\documentclass[final]{lmcs} 
\pdfoutput=1

\usepackage{lastpage}
\lmcsdoi{21}{1}{8}
\lmcsheading{}{\pageref{LastPage}}{}{}%
{Mar.~29,~2023}{Jan.~28,~2025}{}


\usepackage{hyperref}
\usepackage[utf8]{inputenc}  
\usepackage[T1]{fontenc}  

\usepackage{xspace}
\usepackage{tikz}
\usetikzlibrary{automata,positioning,shapes,arrows,decorations,decorations.pathmorphing,decorations.markings,bending}
\usetikzlibrary{shapes.misc}

\tikzset{every loop/.style={looseness=7}, >=latex}
\tikzset{every picture/.style={>=latex}}
\tikzstyle{PlayerMin}=[draw,circle,minimum size=7mm,inner sep=1.5pt]
\tikzstyle{PlayerMax}=[draw,rectangle,minimum size=7mm,inner sep=1.5pt]
\tikzstyle{target}=[circle, minimum size=1mm,inner sep=-2pt]
\tikzstyle{PlayerMinmin}=[draw,circle, minimum size=1.5mm,inner sep=0pt]
\tikzstyle{PlayerMaxmin}=[draw,rectangle,minimum size=1.5mm,inner sep=0pt]
\tikzstyle{leaf}=[draw,diamond,minimum size=7mm,inner sep=1.5pt]

\tikzstyle{strat} =[minimum width=0.1cm,line width=0.01mm,draw=none]
\tikzstyle{vecArrow} = [decoration={markings,mark=at position
	1 with {\arrow[scale=.5,>=latex]{>}}}%,postaction={decorate}
]

\usepackage[obeyFinal]{todonotes}
\usepackage{bbm}
\usepackage{amsfonts,amssymb,amsmath}
\usepackage{wasysym}
\usepackage{stmaryrd}

\usepackage{algorithm}
\usepackage[noend]{algpseudocode}

\algnewcommand{\IfThenElse}[3]{% \IfThenElse{<if>}{<then>}{<else>}
	\State \algorithmicif\ #1\ \algorithmicthen\ #2\ \algorithmicelse\ #3}

\definecolor{DarkOrchid}{HTML}{A4538A}
\definecolor{ForestGreen}{HTML}{009B55}
\definecolor{OrangeRed}{HTML}{ED135A}
\definecolor{DarkBlue}{HTML}{0071BC}
\definecolor{Orange}{HTML}{F58137}

\DeclareMathOperator*{\argmax}{\arg\!\max}
\DeclareMathOperator*{\arginf}{\arg\!\inf}
\DeclareMathOperator*{\argsup}{\arg\!\sup}

\newcommand{\arginfepsilon}[1][]{\ifthenelse{\equal{#1}{}}{{\arginf}^{\varepsilon}}{{\arginf}^{#1}}}
\newcommand{\argsupepsilon}[1][]{\ifthenelse{\equal{#1}{}}{{\argsup}^{\varepsilon}}{{\argsup}^{#1}}}

\newcommand{\R}{\mathbb{R}}
\newcommand{\Rpos}{\R_{\geq 0}}
\newcommand{\Rbar}{\overline\R}
\newcommand{\Z}{\mathbb{Z}}
\newcommand{\N}{\mathbb{N}}

\newcommand{\F}{\mathcal{F}}

\newcommand{\X}{\mathcal{X}}

\newcommand{\game}{\mathcal{G}}
\newcommand{\WTG}{WTG\xspace}
\newcommand{\Cl}{\ensuremath{\mathcal{C}}\xspace}

\newcommand{\clockbound}{\ensuremath{M}\xspace}
\newcommand{\val}{\ensuremath{\nu}\xspace}

\newcommand{\guard}{\ensuremath{I_g}\xspace}

\newcommand{\reset}{\ensuremath{R}\xspace}

\newcommand{\loc}{\ensuremath{\ell}\xspace}

\newcommand{\locinit}{\ensuremath{\loc_{\mathsf{i}}}\xspace}

\newcommand{\Locs}{\ensuremath{L}\xspace}
\newcommand{\LocsMin}{\ensuremath{\Locs_{\MinPl}}\xspace}
\newcommand{\LocsMax}{\ensuremath{\Locs_{\MaxPl}}\xspace}
\newcommand{\LocsT}{\ensuremath{\Locs_t}\xspace}
\newcommand{\LocsUrg}{\ensuremath{\Locs_u}\xspace}
\newcommand{\Trans}{\ensuremath{\Delta}\xspace}
\newcommand{\Transitions}{\Trans}
\newcommand{\trans}{\ensuremath{\delta}\xspace}
\newcommand{\transition}{\trans}
\newcommand{\Transr}{\Trans_{R}}

\newcommand{\QLocs}{Q}
\newcommand{\QLocsT}{\QLocs_t}
\newcommand{\QLocsMin}{\QLocs_{\MinPl}}
\newcommand{\QLocsMax}{\QLocs_{\MaxPl}}
\newcommand{\QLocsUrg}{\QLocs_u}
\newcommand{\qloc}{q}
\newcommand{\qlocinit}{\ensuremath{\qloc_{\mathsf{i}}}\xspace}

\newcommand{\uloc}{\rpath}
\newcommand{\uLocs}{\Locs'}
\newcommand{\uLocsMin}{\LocsMin'}
\newcommand{\uLocsMax}{\LocsMax'}
\newcommand{\uLocsT}{\LocsT'}
\newcommand{\uTrans}{\Trans'}
\newcommand{\utrans}{\trans'}
\newcommand{\uLocsUrg}{\LocsUrg'}

\newcommand{\upath}{\rpath_{\ugame}}

\newcommand{\weight}{\ensuremath{\mathsf{wt}}\xspace}
\newcommand{\weightfin}{\ensuremath{\weight_{t}}\xspace}
\newcommand{\weightT}{\weightfin}
\newcommand{\weightC}{\ensuremath{\weight_\Sigma}\xspace}

\newcommand{\weightP}{\ensuremath{\mathsf{P}}\xspace}

\newcommand{\CTransitions}{\overline\Trans}
\newcommand{\CTrans}{\CTransitions}
\newcommand{\Cweight}{\overline{\weight}\xspace}

\newcommand{\CweightT}{\overline{\weightT}}

\newcommand{\sem}[1]{\ensuremath{\llbracket #1 \rrbracket}\xspace}

\newcommand{\rweight}{\Cweight}

\newcommand{\uweight}{\weight'}
\newcommand{\uweightT}{\weightfin'}

\newcommand{\rpos}{\mathsf{t}_{\geq 0}}
\newcommand{\rneg}{\mathsf{t}_{< 0}}
\newcommand{\revil}{\mathsf{t}_{+\infty}}

\newcommand{\last}{\ensuremath{\mathsf{last}}\xspace}

\newcommand{\id}{\ensuremath{\mathsf{Id}}\xspace}
\newcommand{\proj}{\ensuremath{\mathsf{proj}}\xspace}
\newcommand{\Dproj}{\Trans\proj}
\newcommand{\Pproj}{\Pi\proj}

\newcommand{\play}{\ensuremath{\rho}\xspace}
\newcommand{\rplay}{\play_{\rgame}}
\newcommand{\uplay}{\play_{\ugame}}
\newcommand{\Play}{\ensuremath{\mathsf{Play}}\xspace}

\newcommand{\FPlays}{\ensuremath{\mathsf{FPlays}}\xspace}

\newcommand{\rFPlays}{\FPlays_{\rgame}}
\newcommand{\uFPlays}{\FPlays_{\ugame}}
\newcommand{\uFPlaysEtoile}{\uFPlays^*}
\newcommand{\FPlaysMin}{\FPlays_\MinPl}
\newcommand{\FPlaysMax}{\FPlays_\MaxPl}

\newcommand{\reg}{\ensuremath{I\xspace}}

\newcommand{\rgame}{\ensuremath{{\overline{\game}}}\xspace}
\newcommand{\ugame}{\ensuremath{\mathcal{U}}\xspace}

\newcommand{\rpath}{\ensuremath{\pi}\xspace} 
\newcommand{\ppath}{\ensuremath{\pi}\xspace} 
 
\newcommand{\PPaths}{\ensuremath{\mathsf{FPaths}}\xspace} 
\newcommand{\PPathsMin}{\PPaths_\MinPl}
\newcommand{\PPathsMax}{\PPaths_\MaxPl}

\newcommand{\MinPl}{\ensuremath{\mathsf{Min}}\xspace}
\newcommand{\MaxPl}{\ensuremath{\mathsf{Max}}\xspace}

\newcommand{\Value}{\mathsf{Val}}
\newcommand{\ValueG}{\Value_{\game}}
\newcommand{\ValueU}{\Value_{\ugame}}
\newcommand{\ValueRG}{\Value_{\rgame}}

\newcommand{\Strat}{\ensuremath{\mathsf{Strat}}\xspace}
\newcommand{\StratMin}[1][]{\ifthenelse{\equal{#1}{}}{\ensuremath{\Strat_{\MinPl}}}
	{\ensuremath{\Strat_{\MinPl,#1}}}\xspace}
\newcommand{\StratMax}[1][]{\ifthenelse{\equal{#1}{}}{\ensuremath{\Strat_{\MaxPl}}}
	{\ensuremath{\Strat_{\MaxPl,#1}}}\xspace}

\newcommand{\minstrategy}{\sigma} % for general stochastic strategies
\newcommand{\maxstrategy}{\tau}
\newcommand{\rminstrategy}{\minstrategy_{\rgame}} 
\newcommand{\rmaxstrategy}{\maxstrategy_{\rgame}}
\newcommand{\uminstrategy}{\minstrategy_{\ugame}}
\newcommand{\umaxstrategy}{\maxstrategy_{\ugame}}

\newcommand{\maxstrategyopt}{\maxstrategy^*}

\newcommand{\clockx}{x}
\newcommand{\clocku}{u}
\newcommand{\urgent}{\mathsf u}
\newcommand{\delay}{t}

\newcommand{\PSPACE}{\ensuremath{\mathsf{PSPACE}}\xspace}
\newcommand{\EXPTIME}{\ensuremath{\mathsf{EXPTIME}}\xspace}

\newcommand{\reggame}{\ensuremath{\mathsf{Reg}_\game}}
\newcommand{\bornePseudoPoly}{\kappa}

\newcommand{\maxWeight}{W}
\newcommand{\maxWeightLoc}{W_{\mathsf{loc}}}
\newcommand{\maxWeightTrans}{W_{\mathsf{tr}}}
\newcommand{\maxWeightFinal}{W_{\mathsf{fin}}}
\newcommand{\maxPriceLoc}{\maxWeightLoc}

\newcommand{\moveto}[1]{\ensuremath{\xrightarrow{#1}}\xspace}
\newcommand{\movetoPath}[1]{\ensuremath{\xrightarrow{#1}}\xspace}

\newcommand{\minStratWin}{\Strat}%{\ensuremath{\text{WinStrat}}\xspace} 

\begin{document}
	
	\title[One-Clock Weighted Timed Games with Arbitrary Weights]
	{Decidability of One-Clock Weighted Timed Games with Arbitrary Weights}

	\thanks{We thank the reviewers of the various versions of this article that helped us greatly improve the quality of the results and the writing. This work has been partly funded by the QuaSy project (ANR-23-CE48-0008), and the NCN grant 2019/35/B/ST6/02322.}	%optional

	% affiliations are numbered automatically with a, b, c (see below)
	% use the optional argument to indicate the affiliation(s) of each author
	% omit the argument if there is only one author, or only one affiliation
	\author[B.~Monmege]{Benjamin Monmege\lmcsorcid{0000-0002-4717-9955}}[a]
\author[J.~Parreaux]{Julie Parreaux\lmcsorcid{0009-0009-2744-780X}}[b]
\author[P.-A.~Reynier]{Pierre-Alain Reynier\lmcsorcid{0009-0008-4345-704X}}[a]

% affiliation 1 (automatically numbered a)
\address{Aix Marseille Univ, CNRS, LIS, Marseille, France}	%optional
% write emails for all authors having that affiliation
\email{\{benjamin.monmege,pierre-alain.reynier\}@univ-amu.fr}  %optional

\address{University of Warsaw, Poland}
\email{j.parreaux@uw.edu.pl}
	%% etc.
	
	%% required for running head on odd and even pages, use suitable
	%% abbreviations in case of long titles and many authors:
	
	%%%%%%%%%%%%%%%%%%%%%%%%%%%%%%%%%%%%%%%%%%%%%%%%%%%%%%%%%%%%%%%%%%%%%%%%%%%
	
	%% the abstract has to PRECEDE the command \maketitle:
	%% be sure not to issue the \maketitle command twice!
	\begin{abstract}
		  Weighted Timed Games (\WTG for short) are the most widely used model
		  to describe controller synthesis problems involving real-time
		  issues. Unfortunately, they are notoriously difficult, and
		  undecidable in general. As a consequence, one-clock \WTG{s} have
		  attracted a lot of attention, especially because they are known to
		  be decidable when only non-negative weights are allowed. However,
		  when arbitrary weights are considered, despite several recent works,
		  their decidability status was still unknown. In this paper, we solve
		  this problem positively and show that the value function can be
		  computed in exponential time (if weights are encoded in unary).
	\end{abstract}

\maketitle

\section{Introduction}

% Verification and Synthesis
The task of designing programs is becoming more and more involved.
Developing formal methods to ensure their correctness is thus an
important challenge. Programs sensitive to real-time allow one to
measure time elapsing in order to take decisions. The design of such
programs is a notoriously difficult problem because timing issues may
be intricate, and a posteriori debugging such issues is hard. The
model of timed automata~\cite{AlurDill-94} has been widely adopted as
a natural and convenient setting to describe real-time systems. This
model extends finite-state automata with finitely many real-valued
variables, called clocks, and transitions can check clocks against
lower/upper bounds and reset some clocks.

Model-checking aims at verifying whether a real-time system modelled
as a timed automaton satisfies some desirable property.
Instead of verifying a system, one can try to synthesise
one automatically. A successful approach, widely studied during the
last decade, is one of the two-player games. In this context, a player
represents the \emph{controller}, and an antagonistic player
represents the \emph{environment}. Being able to identify a winning
strategy of the controller, i.e.~a recipe on how to react to
uncontrollable actions of the environment, consists in the synthesis
of a system that is guaranteed to be correct by construction.

In the realm of real-time systems, timed automata have been extended
to timed games~\cite{AsarinMaler-99} by partitioning locations between 
the two players. In a turn-based fashion, the player that must play 
proposes a delay and a transition. The controller aims at satisfying 
some $\omega$-regular objective however the environment player behaves. 
Deciding the winner in such turn-based timed games has been shown to be 
\EXPTIME-complete~\cite{JurdzinskiTrivedi-07}, and a symbolic algorithm 
allowing tool development has been proposed~\cite{BehrmannCDFLL07}.
%settled using the well-known region abstraction,
%yielding to the development of
%tools~\cite{AsarinMaler-99,JurdzinskiTrivedi-07,BehrmannCDFLL07}.
% dec case reachability AM99, ExpTime complete in JT07
% Tool : UppAal TiGa

In numerous application domains, in addition to real-time, other
quantitative aspects have to be taken into account. For instance, one
could aim at minimising the energy used by the system. To address this
quantitative generalisation, weighted (aka priced) timed games (\WTG
for short) have been introduced \cite{BouyerCassezFleuryLarsen-04,BehrmannFehnkerHuneLarsenPetterssonRomijnVaandrager-01}. Locations and transitions are
equipped with integer weights, allowing one to define the accumulated
weight associated with a play.  In this context, one focuses on a
simple, yet natural, reachability objective: given some target
location, the controller, that we now call \MinPl, aims at ensuring
that it will be reached while minimising the accumulated weight. The
environment, that we now call \MaxPl, has the opposite objective:
avoid the target location or, if not possible, maximise the
accumulated weight. This allows one to define the value of the game as
the minimal weight \MinPl can guarantee.  The associated decision
problem asks whether this value is less than or equal to some given
threshold.

In the earliest studies of this problem,
some
semi-decision procedures have been proposed to approximate this value for \WTG{s} with
non-negative weights \cite{AlurBernadskyMadhusudan-04,BouyerCassezFleuryLarsen-04}. In addition, a subclass of strictly non-Zeno cost \WTG{s} for which their
algorithm terminates has been identified in \cite{BouyerCassezFleuryLarsen-04}. This approximation is motivated by the
undecidability of the problem, first shown
in~\cite{BrihayeBruyereRaskin-05}.  This restriction has recently been
lifted to \WTG{s} with arbitrary weights
in~\cite{BusattoGastonMonmegeReynier-17}.

An orthogonal research direction to recover decidability is to reduce
the number of clocks and more precisely to focus on \emph{one-clock
  \WTG{s}}.  Though restricted, a single clock is often sufficient for
modelling purposes.  When only non-negative weights are considered,
decidability has been proven in~\cite{BouyerLarsenMarkeyRasmussen-06}
and later improved 
in~\cite{Rutkowski-11,DueholmIbsenMilterson_13} to obtain exponential 
time algorithms.  Despite several
recent works, the decidability status of one-clock \WTG{s} with arbitrary
weights is still open. In the present paper, we show the decidability
of the value problem for this class. More precisely, we prove that the
value function can be computed in exponential time (if weights are
encoded in unary and not in binary).

Before exposing our approach, let us briefly recap the existing results.
Positive results obtained for one-clock \WTG{s} with non-negative weights are
based on a reduction to so-called \emph{simple \WTG}, where the underlying
timed automata contain no guard, no reset, and the clock value along with the
execution exactly spans the $[0,1]$ interval. In simple \WTG, it is possible to
compute with various techniques inspired by the paradigm of value iteration,
adapted by a computation of the whole value function starting at time $1$ and
going back in time until $0$
\cite{BouyerLarsenMarkeyRasmussen-06,Rutkowski-11,DueholmIbsenMilterson_13},
leading to an exponential-time algorithm. A \PSPACE lower-bound is also known
for related decision problems \cite{FearnleyIbsenJensenSavani-20}.

% Second, a completely different approach is taken
% in~\cite{BrihayeGeeraertsNarayananKrishnaManasaMonmegeTrivedi-14} to
% obtain decidability in pseudo-polynomial time with arbitrary weights
% on transitions and rates on locations restricted to two values
% amongst $\{-d,0,d\}$, with $d\in \mathbb{N}$. 

Recent works
extend the positive results of simple \WTG{s} to arbitrary
weights~\cite{BrihayeGeeraertHaddadLefaucheuxMonmege-15,brihaye2021oneclock},
yielding decidability of \emph{reset-acyclic} one-clock \WTG{s} with
arbitrary weights, with a
\emph{pseudo-polynomial time} complexity (that is polynomial if weights are 
encoded in unary). It is also explained how to extend the result to all \WTG{s} 
where no cyclic play containing a reset may have a negative weight
arbitrarily close to 0. Moreover, it is shown that $\MinPl$ needs memory to play (almost-)optimally, in a very structured way: $\MinPl$ uses \emph{switching strategies}, that are composed of two memoryless strategies, the second one being triggered after a given (pseudo-polynomial) number $\kappa$ of steps.

The crucial ingredient to obtain decidability for non-negative weights or reset-acyclic weighted timed games is to limit the number of reset transitions taken along a play. This is no longer possible in presence of cycles of negative weights containing a reset. There, \MinPl may need to iterate cycles for a number $\kappa$ of times that depends on the desired
precision $\varepsilon$ on the value (to play $\varepsilon$-optimally, $\MinPl$ needs to cycle $O(1/\varepsilon)$ times, see
Example~\ref{ex:value-example}). To rule out these annoying behaviours, we rely on three main ingredients:
\begin{itemize}
\item As there is a single clock, a cyclic path ending with a reset
corresponds to a cycle of configurations. We define the \emph{value} of such a cycle,
that allows us to identify which player may benefit from iterating it.
\item Using the classical region graph construction, we prove 
stronger properties on the value function (it is continuous on the closure
of region intervals). This allows us to prove that \MaxPl has an optimal memoryless 
strategy that avoids cycles whose value is negative (Section~\ref{sec:negative-cycles}).
\item We introduce in Section~\ref{sec:unfolding} a partial unfolding of 
the game, so as to obtain an acyclic \WTG, for which decidability
is known. To do so, we rely on the existence of (almost-)optimal switching strategies for $\MinPl$, allowing us to limit the depth of exploration. Also we
keep track of cycles encountered and handle them according to their value.
We transport the previous result on the existence of a "smart" optimal strategy
for \MaxPl in the context of this unfolding in Section~\ref{sec:negative-cycles-unfolding}. This allows us to show that the unfolding has the same value as the original \WTG in Section~\ref{sec:equal-values}.
\end{itemize}

We finally wrap up the proof in Section~\ref{sec:end-proof}. Along the way, we crucially need that the value function is obtained as a fixed point (indeed the greatest one) of an operator that was already used in many contributions before \cite{AlurBernadskyMadhusudan-04,BouyerCassezFleuryLarsen-04,
BusattoGastonMonmegeReynier-17}. We formally show this statement in Section~\ref{sec:Value-fixpoint}.  

This article is an extended version of the conference article \cite{MonPar22}, with respect to which we have incorporated the full proofs of the result (in particular Section~\ref{sec:Value-fixpoint} is entirely new), in a clarified way.

\section{Weighted timed games}
\label{sec:WTG}

\subsection{Definitions}
We only consider weighted timed games with a single clock, denoted by~$\clockx$. The valuation $\val$ of this clock is a non-negative real number, 
i.e.~$\val \in \Rpos$. On such a clock, transitions of the timed games will be 
able to check some interval constraints, called \emph{guards}, on the clock, 
i.e.~intervals $I$ of real values with closed or open bounds that are 
natural numbers (or $+\infty$). For every interval $I$ having finite bounds $a$ and $b$, we denote 
its closure by $\bar{I} = [a, b]$.

\begin{defi}
  A \emph{weighted timed game} (\WTG for short) is a tuple
  $\game= \langle\QLocsMin, \QLocsMax, \QLocsT, \QLocsUrg, \allowbreak \Transitions,
  \weight,\weightT\rangle$~with
  \begin{itemize}
  \item $\QLocs = \QLocsMin \uplus \QLocsMax \uplus \QLocsT$ a finite set of
    locations split between players $\MinPl$ and $\MaxPl$, and a set of 
    target locations;
  \item $\QLocsUrg \subseteq \QLocsMin \uplus \QLocsMax$ a set of
    \emph{urgent} locations where time cannot be delayed;
  \item $\Transitions$ a finite set of transitions each of the form
    $(\qloc, \reg, \reset, w, \qloc')$, with $\qloc$ and
    $\qloc'$ two locations (with $\qloc \notin \QLocsT$), $\reg$ an 
    interval, $w \in \Z$ the weight of the transition, and $\reset$
    being either~$\{\clockx\}$ when the clock must be reset, or $\emptyset$ 
    when it does not;
  \item $\weight \colon \QLocs \to \Z$ a weight function associating an
    integer weight with each location: for uniformisation of the
    notations, we extend this weight function to also associate with
    each transition the weight it contains,
    i.e.~$\weight\big((\qloc, \reg, \reset, w, \qloc')\big) = w$;
  \item and $\weightT\colon \QLocsT\times\Rpos\to \Rbar$ a function
    mapping each target configuration to a final weight, where
    $\Rbar = \R\cup\{-\infty, +\infty\}$.
  \end{itemize}
\end{defi}

We note that our definition is not usual. Indeed, the addition of 
final weights in \WTG{s} is not standard, 
but we use it in the process of solving those games: in any case, 
it is possible to simply map a given target location to the 
weight~$0$, allowing us to recover the standard definitions of the 
literature. The presence of urgent locations is also unusual: in a 
timed automaton with several clocks, urgency can be modelled with 
an additional clock $\clocku$ that is reset just before entering 
the urgent location and with constraints $\clocku \in [0, 0]$ on 
outgoing transitions. However, when limiting the number
of clocks to one, we regain modelling capabilities by allowing for
such urgent locations. The weight of an urgent location is never used
and will thus not be given in drawings: instead, urgent locations will
be displayed with $\urgent$ inside.

Given a \WTG $\game$, its semantics, denoted by $\sem{\game}$, is 
defined in terms of a game on an infinite transition system whose 
vertices are \emph{configurations} of $\game$, i.e.~the set of 
pairs~$(\qloc, \val) \in \QLocs \times \Rpos$. Configurations are
split into players according to the location $\qloc$, and a
configuration~$(\qloc, \val)$ is a target if $\qloc \in \QLocsT$. 
To encode the delay spent in the current 
location before firing a certain transition, edges linking vertices 
will be labelled by elements of $\Rpos \times \Transitions$. Formally, 
for every delay $\delay \in \Rpos$, transition
$\transition = (\qloc, \reg, \reset, w, \qloc') \in \Transitions$ and
valuation~$\val$, we add a labelled edge
$(\qloc,\val) \moveto{\delay, \transition} (\qloc',\val')$~if
\begin{itemize}
\item $\val+\delay \in \reg$;
\item $\val' = 0$ if $\reset = \{\clockx\}$, and 
  $\val' = \val + \delay$ otherwise;
\item and $\delay = 0$ if $\qloc \in \QLocsUrg$.
\end{itemize}
This edge is given a weight 
$\delay\times \weight(\qloc)+ \weight(\transition)$ 
taking into account discrete and continuous~weights.
Without loss of generality by applying classical 
techniques~\cite[Lemma~5]{BerPet98}, we  
suppose the absence of deadlocks except on target
locations, i.e.~for each location $\qloc\in \QLocs \backslash \QLocsT$
and valuation $\val$, there exist $\delay\in\Rpos$
and $\transition = (\qloc, \reg, \reset, w, \qloc') \in \Transitions$
such that $(\qloc, \val) \moveto{\delay, \transition} 
(\qloc', \val')$ and no transitions start from~$\QLocsT$.

\paragraph{Paths and plays}
We call \emph{path} a finite or infinite sequence of consecutive 
transitions $\qloc_0 \movetoPath{\transition_0} \qloc_1
\movetoPath{\transition_1} \cdots$ where $\trans_0, \trans_1, \ldots\in \Transitions$ and $\qloc_0, \qloc_1, \ldots\in\QLocs$. We 
sometimes denote $\rpath_1 \cdot \rpath_2$ the concatenation of a finite path
$\rpath_1$ ending in location $q$ and another path $\rpath_2$ starting in location $q$. 
We call \emph{play} a finite or infinite sequence of edges in the 
semantics of the game $(\qloc_0,\val_0) 
\moveto{\delay_0,\transition_0} (\qloc_1,\val_1)
\moveto{\delay_1, \transition_1} (\qloc_2,\val_2) \cdots$. 
A play is said to \emph{follow} a path if both use the same sequence of 
transitions. We let $\PPaths$ (resp.~$\FPlays$) be the set of all 
finite paths (resp.~plays).

Given a finite path $\ppath$ or a finite play $\play$, we let $|\ppath|$ 
or $|\play|$ its \emph{length} which is its number of transitions (or edges), 
and $|\ppath|_\trans$ or $|\play|_\trans$ the number of occurrences of 
a given transition $\trans$ in $\ppath$ (or $\play$). More generally, for 
a play~$\play$ and a set $A$ of transitions, we let $|\play|_A$ be 
the number of occurrences of all transitions from $A$ in $\play$, i.e.~$|\play|_A = \sum_{\trans \in A} |\play|_\trans$.  We also let 
$\last(\ppath)$ and $\last(\play)$ be the last location or configuration. 
Finally, $\PPathsMax$ (resp.~$\PPathsMin$) and $\FPlaysMax$
(resp.~$\FPlaysMin$) denote the subset of finite paths or plays 
whose last element belong to player \MaxPl (resp.~\MinPl).

A finite play $\play= (\qloc_0, \val_0) \moveto{\delay_0, \trans_0} (\qloc_1, \val_1) 
\cdots (\qloc_k, \val_k)$ can be associated with the cumulated weight of the edges it traverses: 
\begin{displaymath}
\weightC(\play) = \sum_{i= 0}^{k-1} \big(\weight(\loc_i) \times t_i +
\weight(\trans_i)\big) \,.
\end{displaymath} 
A \emph{maximal play} $\play$ (either infinite or trapped in a 
deadlock that is necessarily a target configuration) is associated 
with a \emph{payoff} $\weightP(\play)$ as follows: the payoff of an 
infinite play (meaning that it never visits a target location) is $+\infty$, 
while the payoff of a finite play, thus ending in a target 
configuration~$(\qloc, \val)$, is $\weightC(\play) + \weightT(\qloc,\val)$. 
By~\cite{BehrmannFehnkerHuneLarsenPetterssonRomijnVaandrager-01}, the set of weights of plays following a given 
path is known to be an interval of values. Moreover, when all the guards 
along the path are closed intervals, this interval has closed bounds.

A \emph{cyclic path} is a finite path that starts and ends in the same location. 
A \emph{cyclic play} is a finite play that starts and ends in the same 
configuration: it necessarily follows a cyclic path, but the reverse might not be 
true since some non-cyclic plays can follow a cyclic path (if they do not end 
in the same valuation as the one in which they start).

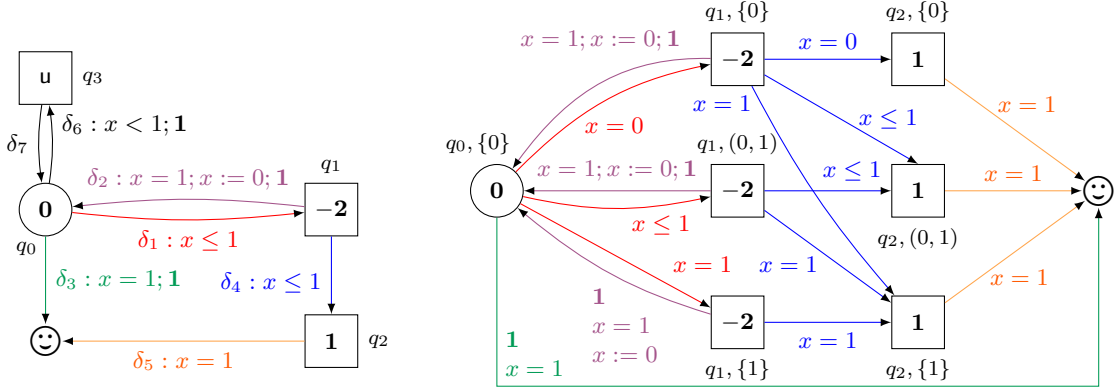
\begin{figure}
	\centering
	\begin{tikzpicture}[xscale=.8,every node/.style={font=\footnotesize}, 
	every label/.style={font=\scriptsize}]
          % Draw the states
    \node[PlayerMax, label={right:$\qloc_3$}] at (0, 1.75)  (s6) {$\urgent$};
	\node[PlayerMin, label={[xshift=-.25cm, yshift=-1.1cm]$\qloc_0$}] at (0, 0)  (s0) {$\mathbf{0}$};
	\node[PlayerMax, label={above:$\qloc_1$}] at (4.75, 0) (s2) {$\mathbf{-2}$};
	\node[PlayerMax, label={right:$\qloc_2$}] at (4.75, -1.75) (s3) {$\mathbf{1}$};
	\node[target] at (0, -1.75) (s4) {$\LARGE \smiley$};
	
	% Connect the states with arrows
	\draw[->]
    (s0) edge[bend right=10] node[right, near end] {$\trans_6: x< 1; \mathbf{1}$} (s6)
    (s6) edge[bend right=10] node[left] {$\trans_7$} (s0)
	(s0) edge[bend right=05, draw=red] node[below] {
		\textcolor{red}{$\trans_1 : x \leq 1$}} (s2)
	(s2) edge[bend right=05, draw=DarkOrchid] node[above] {
		\textcolor{DarkOrchid}{$\trans_2 : x = 1; x:=0; \mathbf{1}$}} (s0)
	(s0) edge[draw=ForestGreen] node[right] {
		\textcolor{ForestGreen}{$\trans_3 : x=1; \mathbf{1}$}} (s4)
	(s2) edge[draw=blue] node[left,yshift=-1mm] {\textcolor{blue}{$\trans_4 : x \leq 1$}} (s3)
	(s3) edge[draw=orange!70] node[below] {\textcolor{orange!70!red}{$\trans_5 : x = 1$}} (s4)
	;
	
	\begin{scope}[xshift=7.5cm,yshift=.25cm]
	\node[PlayerMin, label={[xshift=-.25cm]$\qloc_0,\{0\}$}] at (0, 0)  (s0) {$\mathbf{0}$};
	\node[PlayerMax, label={above:$\qloc_1,\{0\}$}] at (4, 1.75) (s20) {$\mathbf{-2}$};
	\node[PlayerMax, label={above:$\qloc_1,(0,1)$}] at (4, 0) (s21) {$\mathbf{-2}$};
	\node[PlayerMax, label={below:$\qloc_1,\{1\}$}] at (4, -1.75) (s22) {$\mathbf{-2}$};
	\node[PlayerMax, label={above:$\qloc_2,\{0\}$}] at (7, 1.75) (s30) {$\mathbf{1}$};
	\node[PlayerMax, label={below:$\qloc_2,(0,1)$}] at (7, 0) (s31) {$\mathbf{1}$};
	\node[PlayerMax, label={below:$\qloc_2,\{1\}$}] at (7, -1.75) (s32) {$\mathbf{1}$};
	\node[target] at (10, 0) (s4) {$\LARGE \smiley$};
	
	\node[coordinate] at (5, -2.6) (tmp1) {};
	\draw[->, draw=ForestGreen] 
	(s0) |- node[right,yshift=.4cm,xshift=-.2cm] {\textcolor{ForestGreen}{$\begin{array}{l}
			\mathbf{1} \\
			x=1 
			\end{array}$}} (tmp1) -| (s4) ;
	
	\draw[->]
	(s0) edge[bend left=15, draw=red] 
		node[below,yshift=-.1cm,xshift=.1cm] {\textcolor{red}{$x = 0$}} (s20)
	(s0) edge[bend right=10, draw=red] node[below,near end] {\textcolor{red}{$x\leq 1$}} (s21)
	(s0) edge[draw=red] 
		node[above right, near end] {\textcolor{red}{$x = 1$}} (s22)
		
	(s20) edge[bend right=25, draw=DarkOrchid] node[above,yshift=3mm] {
		\textcolor{DarkOrchid}{$x = 1; x:=0; \mathbf{1}$}} (s0)
	(s21) edge[draw=DarkOrchid] node[above] {\textcolor{DarkOrchid}{$x = 1; x:=0; \mathbf{1}$}} (s0)
	(s22) edge[bend left=10,draw=DarkOrchid] 
		node[below,xshift=.2cm] {\textcolor{DarkOrchid}{$\begin{array}{l}
			\mathbf{1}\\ x = 1 \\ x:=0
			\end{array}$}} (s0)
	
	(s20) edge[draw=blue] node[above] {\textcolor{blue}{$x = 0$}} (s30)
	(s20) edge[draw=blue] node[right,xshift=.1cm] {\textcolor{blue}{$x \leq 1$}} (s31.north)
	(s20) edge[bend right=7, draw=blue] 
		node[very near start,left,yshift=.2cm,xshift=-.1cm] {\textcolor{blue}{$x = 1$}} (s32)
	(s21) edge[draw=blue] node[above, near end,xshift=-1mm] {\textcolor{blue}{$x \leq 1$}} (s31)
	(s21) edge[draw=blue] node[left,yshift=-.1cm] {\textcolor{blue}{$x = 1$}} (s32)
	(s22) edge[draw=blue] node[below] {\textcolor{blue}{$x = 1$}} (s32)
	
	(s30) edge[draw=orange!70] 
		node[right, near start, xshift=.1cm] {\textcolor{orange!70!red}{$x = 1$}} (s4)
	(s31) edge[draw=orange!70] node[above] {\textcolor{orange!70!red}{$x = 1$}} (s4)
	(s32) edge[draw=orange!70] node[right, near start] {\textcolor{orange!70!red}{$x = 1$}} (s4)
	;
	\end{scope}
	\end{tikzpicture}
	\caption{On the left, a \WTG with a cyclic path of weight $[-1, 1]$ 
		containing a reset. Its weights are depicted in bold font, and 
		the missing ones are $0$. Locations belonging to \MinPl (resp. \MaxPl) 
		are depicted by circles (resp. squares). Transitions that contain the 
		reset of $\clockx$ are labelled with $\clockx := 0$. The intervals of guards are described, as classically done in timed automata, via equality or inequality constraints on the unique clock $x$.  The target location 
		is ${\Large\smiley}$, whose final weight function is zero. Location 
		$\qloc_3$ is urgent. On the right, the restriction of its closure to 
		locations $\qloc_0, \qloc_1, \qloc_2$ and~${\Large\smiley}$.}
	\label{fig:Ex_transition_2valeurs}
\end{figure}

\begin{exa}
  \label{ex:weight-cycle}
  Plays that follow the cyclic path $\rpath = \qloc_0 \movetoPath{\trans_1} \qloc_1 
  \movetoPath{\trans_2} \qloc_0$ of the WTG
  depicted on the left in \figurename{~\ref{fig:Ex_transition_2valeurs}} have 
  weight between $-1$ (with the play $(\qloc_0, 0) \moveto{0, \trans_1} (\qloc_1, 0) 
  \moveto{1, \trans_2} (\qloc_0, 0)$) and $1$ (with the play $(\qloc_0, 0) 
  \moveto{1, \trans_1} (\qloc_1, 1) \moveto{0, \trans_2} (\qloc_0, 0)$), so 
  $\weightC(\rpath) = [-1, 1]$. Another cyclic path is $\rpath' = \qloc_0 
  \movetoPath{\trans_6} \qloc_3 \movetoPath{\trans_7} \qloc_0$ which goes via an 
  urgent location. In particular, all plays that follow this one are of the form
  $(\qloc_0,\val) \moveto{\delay,\trans_6} (\qloc_3,\val+\delay) 
  \moveto{0,\trans_7} (\qloc_0,\val+\delay)$ with $\val$ and $\val+\delay$ 
  less than $1$: they all have weight $1$.
\end{exa}

\paragraph{Strategies and value}
A \emph{strategy} gives a set of choices to one of the players.
A strategy of $\MinPl$ is a function
$\minstrategy \colon \FPlaysMin\to \Rpos \times \Transitions$ mapping
each finite play $\play$ whose last configuration belongs to $\MinPl$
to a pair $(\delay, \transition)$ of delay and transition, such that
the play $\play$ can be extended by an edge labelled with
$(\delay, \transition)$. A play $\play$ is said to be 
\emph{conforming} to a strategy $\minstrategy$ if the choice made in
$\play$ at each location of $\MinPl$ is the one prescribed by
 $\minstrategy$. Moreover, a finite path $\rpath$ is said to be 
 conforming to a strategy $\minstrategy$ if there exists a finite play
following~$\rpath$ that is conforming to $\minstrategy$.
Similar definitions hold for strategies
$\maxstrategy$ of $\MaxPl$. We let $\StratMin[\game]$ (resp.,
$\StratMax[\game]$) be the set of strategies of \MinPl (resp.,
\MaxPl) in the game $\game$, or simply $\StratMin$ and $\StratMax$ if
the game is clear from the context: we will always use letters
$\minstrategy$ and $\maxstrategy$ to differentiate from strategies of
$\MinPl$ and $\MaxPl$. 

A strategy is said to be \emph{memoryless} if it only depends on the
last configuration of the plays. More formally, \MaxPl's strategy
$\maxstrategy$ is memoryless if for all plays $\play$ and $\play'$
such that $\last(\play)=\last(\play')$, we have
$\maxstrategy(\play)=\maxstrategy(\play')$.

After both players have chosen their strategies $\minstrategy$ and
$\maxstrategy$, each initial configuration $(\qloc,\val)$ gives rise
to a unique maximal play that we denote by
$\Play((\qloc,\val),\minstrategy,\maxstrategy)$. The \emph{value} of
the configuration $(\qloc,\val)$ is then obtained by letting players
choose their strategies as they want, first $\MinPl$ and then
$\MaxPl$, or vice versa since \WTG{s} are known to be
determined~\cite{brihaye2021oneclock}:
\begin{displaymath}
  \Value_\game(\qloc,\val) 
  = \sup_\maxstrategy \inf_\minstrategy
          \weightP(\Play((\qloc,\val),\minstrategy,\maxstrategy))
  = \inf_\minstrategy \sup_\maxstrategy
          \weightP(\Play((\qloc,\val),\minstrategy,\maxstrategy)) \,.
\end{displaymath}

The value of a strategy $\minstrategy$ of $\MinPl$ (symmetric
definitions can be given for strategies $\maxstrategy$ of $\MaxPl$) is
defined as:
\begin{displaymath}
\Value_{\game}^{\minstrategy}(\qloc,\val) = \sup_\maxstrategy
  \weightP(\Play((\qloc,\val),\minstrategy,\maxstrategy)) \,. 
\end{displaymath}
Then, a strategy $\minstrategy^*$ of $\MinPl$ is \emph{optimal} if, 
for all initial configurations $(\qloc,\val)$,
\begin{displaymath}
\Value_{\game}^{\minstrategy^*}(\qloc,\val)\leq \Value_{\game}(\qloc,\val) \,. 
\end{displaymath}
Because of the infinite nature of the timed games, optimal strategies may not
exist: for example, a player may want to let time elapse as much as
possible, but with a delay $\delay<1$ because of a strict guard,
preventing them to obtain the optimal value. We will see in
Example~\ref{example:Max-not-optimal} that this situation can even
happen when all guards contain only \emph{closed} comparisons. We
naturally extend the definition to \emph{almost-optimal strategies},
taking into account small possible errors: we say that a
strategy~$\minstrategy^*$ of $\MinPl$ is\emph{ $\varepsilon$-optimal} 
if, for all initial configurations~$(\qloc,\val)$,
\begin{displaymath}
\Value_{\game}^{\minstrategy^*}(\qloc,\val)\leq
\Value_{\game}(\qloc,\val)+\varepsilon \,.
\end{displaymath}

\begin{exa}\label{ex:value-example}
  We have seen, in Example~\ref{ex:weight-cycle}, that in $q_0$ 
  (on the left in \figurename~\ref{fig:Ex_transition_2valeurs}), 
  $\MinPl$ has no interest in following the cycle
  $\qloc_0 \movetoPath{\trans_6} \qloc_3 \movetoPath{\trans_7} \qloc_0$ 
  since all plays following it have weight $1$. Jumping directly to the target location
  via~$\trans_3$ leads to a weight of $1$. But $\MinPl$ can do better:
  from valuation~$0$, by jumping to $q_1$ after a delay of $t\leq 1$,
  it leaves a choice to $\MaxPl$ to either jump to $q_2$ and the
  target leading to a total weight of $1-t$, or to loop back in $q_0$
  thus closing a cyclic play of weight $-2(1-t)+1=2t-1$. If $t$ is
  chosen too close to $1$, the value of the cycle is greater than $1$,
  and $\MaxPl$ will benefit from it by increasing the total weight. If
  $t$ is chosen smaller than~$1/2$, the weight of the cycle is
  negative, and $\MaxPl$ will prefer to go to the target to obtain a
  weight $1-t$ close to $1$, not very beneficial to $\MinPl$.  Thus,
  $\MinPl$ prefers to play just above~$1/2$, for example at
  $1/2+\varepsilon$. In this case, $\MaxPl$ will choose to go
  to the target with a total weight of $1/2+\varepsilon$. The value of
  the game, in configuration $(q_0,0)$, is thus $\Value_{\game}(q_0,0) =
  1/2$. Not only $\MinPl$ does not have an optimal strategy (but only
  $\varepsilon$-optimal ones, for every $\varepsilon > 0$), but needs
  memory to play $\varepsilon$-optimally, since $\MinPl$ cannot play
  \emph{ad libitum} transition $\trans_2$ with a delay
  $1/2-\varepsilon$: in this case, $\MaxPl$ would prefer staying in
  the cycle, thus avoiding the target. Thus, $\MinPl$ will play the
  transition $\trans_1$ at least $1/4\varepsilon$ times  so that the
  cumulated weight of all the cycles is below $-1/2$, in which case
  $\MinPl$ can safely use transition $\trans_1$ still earning $1/2$ 
  in~total.
\end{exa}

\paragraph{Clock bounding}
Seminal works in \WTG{s} 
\cite{AlurBernadskyMadhusudan-04,BouyerCassezFleuryLarsen-04} have 
assumed that clocks are bounded. This is known to be without loss 
of generality for (weighted) timed 
automata~\cite[Theorem~2]{BehrmannFehnkerHuneLarsenPetterssonRomijnVaandrager-01}: 
it suffices to replace transitions with unbounded delays with 
self-loop transitions periodically resetting the clock. We do not 
know if it is the case for the \WTG{s} defined above since this 
technique cannot be directly applied. This would give too much 
power to player \MaxPl that would then be allowed to loop in a 
location (and thus avoid the target) where an unbounded delay could 
originally be taken before going to the target. 
In~\cite{BouyerCassezFleuryLarsen-04}, since the \WTG{s} are concurrent, 
this new power of \MaxPl is compensated by always giving \MinPl a chance 
to move outside of such a situation. Trying to detect and avoid such 
situations in our turn-based case seems difficult in the presence of 
negative weights since the opportunities of \MaxPl crucially depend 
on the configurations of value $-\infty$ that \MinPl could control 
afterwards: the problem of detecting such configurations (for all 
classes of \WTG{s}) is undecidable~\cite[Prop.~9.2]{Bus19}, which is 
additional evidence to motivate the decision to focus only on bounded 
\WTG{s}. We thus suppose from now on that the clock is bounded by 
a constant $M \in \N$, i.e.~every transition of the \WTG is equipped 
with the interval $[0, M]$.

\paragraph{Regions}
In the following, we rely on the crucial notion of regions introduced in 
the seminal work on timed automata~\cite{AlurDill-94} to obtain a partition of the 
set of valuations $[0, M]$. To reduce the number of regions concerning the more 
usual one of~\cite{AlurDill-94} in the case of a single clock, we define regions by 
a construction inspired by Laroussinie, Markey, and Schnoebelen~\cite{LarMar04}. 
Formally, we call regions of $\game$ the set 
\begin{displaymath}
\reggame = \{(M_i,M_{i+1})\mid 0\leq i\leq k-1\} \cup \{ \{M_i\}\mid 0\leq i\leq k\}
\end{displaymath} 
where $M_0 = 0 < M_1 < \cdots < M_k$ are all the endpoints of the intervals appearing 
in the guards of $\game$ (to which we add $0$ if needed). As usual, if $\reg$ is a 
region, then the time successor of valuations in $\reg$ forms a finite union of regions, 
and the reset $\reg[\clockx := 0] = \{0\}$ is also a region. A region $\reg'$ is 
said to be a \emph{time successor} of the region $\reg$ if there exists 
$\val \in \reg$, $\val' \in \reg'$, and $\delay > 0$ such that $\val' = \val + \delay$. 

%A game $\game$ can be populated with the region information without loss of 
%generality, building what is called the \emph{region game}in~\cite{BusattoGastonMonmegeReynier-17}, the addition of the classical region automaton with information on the owner of locations inherited from $\game$. 
%We let $\mathcal{R}(\game)$, called the region game, be the \WTG with locations 
%$S = \QLocs \times \reggame$ and all transitions 
%$((\qloc, \reg), \guard \cap \reg'', \reset, w, (\loc', \reg'))$ with 
%$(\loc, \guard, \reset, w, \loc') \in \Trans$ such that 
%$\reg''$ is a region time successor of $\reg$ that satisfies the guard $\guard$, 
%and $\reg'$ is the region obtained from $\reg''$ by resetting the clock of $\reset$, 
%i.e.~$\reg'= \reg''$ if $\reset=\emptyset$ and $\reg'= \{0\}$ otherwise. Distribution 
%of locations to players, final locations, and weights are inherited from $\game$.

\paragraph{Final weights}
We also assume that the final weight functions satisfy a sufficient property
ensuring that they can be encoded in finite space: we require final weight functions 
to be piecewise affine with a finite number of pieces and continuous on each
region. More precisely, we assume that cutpoints (the value of the clock in-between two affine pieces) and coefficients are 
rational and given in binary. 

We let $\maxWeightLoc$, $\maxWeightTrans$ and $\maxWeightFinal$ be the
maximum absolute value of weights of locations, transitions and final
functions, i.e.
\begin{displaymath}
\maxWeightLoc = \max_{\qloc\in\QLocs}|\weight(\qloc)| 
\qquad 
\maxWeightTrans = \max_{\trans\in\Transitions}|\weight(\trans)| 
\qquad
\maxWeightFinal = \sup_{\substack{\qloc\in\QLocsT \\
	\weightT(\qloc,\cdot) \notin \{+\infty,-\infty\}}}
\sup_{\val} |\weightT(\qloc,\val)|
\end{displaymath}
We also let $\maxWeight$ be the maximum of $\maxWeightLoc$, 
$\maxWeightTrans$, and $\maxWeightFinal$. 

\subsection{Fixpoint characterisation of the value}

The value function $\Value_\game\colon \QLocs\times \Rpos \to \Rbar$ of \WTG{s} can sometimes be characterised as a fixpoint (and even the greatest fixpoint) of some operator $\F$ defined as follows: for
  all configurations~$(\qloc, \val)$ and all mappings
  $X \colon \QLocs \times \Rpos \to \Rbar$, we let: 
  \begin{displaymath}
  \F(X)(\qloc, \val) =
  \begin{cases}
	\weightT(\qloc, \val) & \text{if } \qloc \in \QLocsT \\
    \inf_{(\qloc,\val) \moveto{\delay, \trans} (\qloc',\val')} 
	\big(\weight(\trans) + \delay\, \weight(\qloc) + X(\qloc', \val')\big) 
                          & \text{if } \qloc \in \QLocsMin \\
    \sup_{(\qloc,\val) \moveto{\delay, \trans} (\qloc',\val')} 
	\big(\weight(\trans) + \delay\, \weight(\qloc) + X(\qloc', \val')\big) 
                          & \text{if } \qloc \in \QLocsMax 
  \end{cases}
  \end{displaymath}

This operator is the basis of the decidability result for (many-clocks) \WTG{s} with non-negative weights with some divergence conditions on the weight of cycles \cite{BouyerCassezFleuryLarsen-04}, since the value iteration algorithm that iterates the operator over an initial well-chosen function is supposed to converge (in finite time) towards the desired value. 
However the proof given in \cite{BouyerCassezFleuryLarsen-04,Bou16} of the claim that $\Value_\game$ is indeed the greatest fixpoint of $\F$ contains flaws since they suppose that the limit of the iterates of $\F$ is a continuous function of $\Rpos$ to prove 
that this limit is the value function. Since the limit of a sequence 
of continuous functions may not be continuous, this fact needs to 
be proven. 

Fortunately, the necessary claim can be recovered in the case of such \emph{divergent} \WTG{s} (at least in the turn-based case that we consider in this article, and not necessarily in the concurrent case studied in \cite{BouyerCassezFleuryLarsen-04}) even in presence of both negative and non-negative weights, as can be recovered from \cite{BusattoGastonMonmegeReynier-18}.
  	
Moreover, in the non-divergent case, with negative weights in \WTG{s}, the 
continuity of the value function is indeed not 
guaranteed~\cite[Remark~3.3]{brihaye2021oneclock}. In particular, this implies that
the proof (even if we somehow obtain the continuity of the limit) can not a priori
be adapted to all \WTG{s} with negative weights. 
  	
In our specific one-clock case, we are able to correct the proof of \cite{BouyerCassezFleuryLarsen-04,Bou16}. 
As the proof is long and technical, and
orthogonal to the rest of the paper, we defer it to Section~\ref{sec:Value-fixpoint}.
We obtain there the following result: 

\begin{thmC}\label{thm:Value-fixpoint}\label{theo:fixedpoint}
The value function of all (one-clock) \WTG{s} is the greatest fixpoint of the operator $\F$. 
\end{thmC}

\subsection{Closure}
A game $\game$ can be populated with the region information without loss of 
generality, building what is called the \emph{region game}
in~\cite{BusattoGastonMonmegeReynier-17}, the addition of the classical region automaton with information on the owner of locations inherited from $\game$. 
To solve one-clock \WTG{s} without reset transitions in~\cite{brihaye2021oneclock}, 
authors do not use the usual region game. Indeed, their method is based 
on a construction that consists in not only enhancing the locations with regions 
(as the region game) but also closing all guards while preserving the value 
of the original game. 
%We update the presentation of this method here. 
%In the following, we call this construction 
%the \emph{closure} of the game, and we start by giving a slightly 
%updated presentation of it.

\begin{defi}
  The \emph{closure} of a \WTG~$\game$ is the \WTG
  $\rgame = \langle \LocsMin, \LocsMax, \LocsT, \LocsUrg, \CTransitions, 
  \Cweight, \CweightT \rangle$ where:
  \begin{itemize}
  \item
    $\Locs = \LocsMin \uplus \LocsMax \uplus \LocsT$ with 
    $\LocsMin = \QLocsMin \times \reggame$, 
    $\LocsMax = \QLocsMax \times \reggame$, 
    $\LocsT = \QLocsT \times \reggame$, and 
    $\LocsUrg = \QLocsUrg \times \reggame$;
    
  \item for all $(\qloc, \reg) \in \Locs$, 
    $\big((\qloc,I), \overline{\guard \cap \reg''}, \reset, w, (\qloc',\reg')\big) 
    \in \CTransitions$ if and only if there exist a transition 
    $(\qloc, \guard, \reset, w, \qloc') \in \Transitions$, and a 
    region $\reg''$ such that $\guard \cap \reg'' \neq \emptyset$, the lower 
    bound of $\reg''$ is a time successor of $\reg$, and $\reg'$ is equal to 
    $\reg''$ if $\reset=\emptyset$ and to $\{0\}$ otherwise: 
    $\overline{\guard \cap \reg''}$ stands for the topological 
    closure of the non-empty interval $\guard \cap \reg''$;
    
  \item for all $(\qloc,I)$, we have
    $\Cweight(\qloc,I) = \weight(\qloc)$;
    
  \item for all $(\qloc,I)\in \LocsT$, for $\val\in I$,
    $\CweightT((\qloc,I),\val)=\weightT(\qloc,\val)$ and extend
    $\val\mapsto\CweightT((\qloc,I),\val)$ by continuity on $\bar{I}$, 
    the closure of the interval $I$. We may also let
    $\CweightT((\qloc,I),\val)=+\infty$ for all
    $\val\notin\overline I$, even though we will never use this in the
    following.
  \end{itemize}
\end{defi}

An example of closure is given in \figurename~\ref{fig:Ex_transition_2valeurs}, which depicts the 
closure (right) of the \WTG (left) restricted to locations 
$\qloc_0, \qloc_1, \qloc_2$, and ${\Large\smiley}$ (we have seen that 
$\qloc_3$ is anyway useless). 

The semantic of the closure is 
obtained by concentrating on the following set of configurations which is an invariant of the 
closure (i.e.~starting from such configuration fulfilling the invariant, we 
can only reach configurations fulfilling the~invariant):
\begin{itemize}
\item configurations $((\qloc,\{M_k\}),M_k)$; 
\item and configurations $((\qloc,(M_k,M_{k+1})),\val)$ with
  $\val\in [M_k, \allowbreak M_{k+1}]$ (and not only in $(M_k,M_{k+1})$ as one
  might expect in the region game).
\end{itemize}

The closure of the guards allows players to mimic a move in $\game$
``arbitrarily close'' to $M_k$ (or $M_{k+1}$) in $(M_k,M_{k+1})$ to be simulated
by jumping on $M_k$ (or $M_{k+1}$) still staying in the region $(M_k,M_{k+1})$. In 
particular, it is shown in \cite{brihaye2021oneclock}
that we can transform an $\varepsilon$-optimal strategy of $\rgame$
into an $\varepsilon'$-optimal strategy of $\game$ with
$\varepsilon' < 2 \varepsilon$ and vice-versa. Thus, the closure of a
\WTG preserves its value.

\begin{lemC}[\cite{brihaye2021oneclock}]
  \label{lem:region_ptg}
  For all \WTG{s} $\game$, $(\qloc,I)\in \QLocs\times\reggame$ and
  $\val\in I$,
  \begin{displaymath}
  \Value_\game(\qloc,\val)=\Value_{\rgame}((\qloc,I),\val) \,.
  \end{displaymath}
\end{lemC}

Moreover, the closure construction also makes the value function 
more manageable for our purpose. Indeed, as shown in
\cite{brihaye2021oneclock}, the mapping
$\val \mapsto \Value_{\game}(\loc,\val)$ is continuous over all
regions, but there might be discontinuities at the borders of the
regions. The closure construction clears this issue by softening the
borders of each region independently:  we show the continuity of the value 
function on each closed region (and not only on the regions) in the closure game 
by following a very similar sketch as the one of
\cite[Theorem~3.2]{brihaye2021oneclock}. The completed proof is given in
Appendix~\ref{app:proof-val_continue}.

\begin{restatable}{lem}{lemValContinue}
  \label{lem:val_continue}
  For all \WTG{s} $\game$ and $(\qloc,I)\in \QLocs\times\reggame$, the
  mapping $\val\mapsto \Value_{\rgame}((\qloc,I),\val)$ is continuous
  over $\overline I$.
\end{restatable}

In \cite{brihaye2021oneclock}, it is also shown that the
mapping $\val\mapsto \Value_{\game}(\loc,\val)$ is piecewise affine on each
region where it is not infinite, that the total number of pieces (and thus of
cutpoints in-between two such affine pieces) is pseudo-polynomial (i.e.~polynomial in the
number of locations and the biggest weight $W$), and that
all cutpoints and the value associated to such a cutpoint are rational numbers. We will only use this result on \emph{reset-acyclic} \WTG{s}, i.e.~that do not contain cyclic 
  paths with a transition with a reset, which we formally cite here:

\begin{thmC}[\cite{brihaye2021oneclock}]
  \label{thm:pseudopoly}
  If $\game$ is a reset-acyclic \WTG, then for all locations $\qloc$, the piecewise affine mapping
  $\val \mapsto \Value_{\game}(\qloc,\val)$ is computable in
  time~polynomial in $|\QLocs|$ and $\maxWeight$. 
\end{thmC}

In \cite{brihaye2021oneclock}, this result is extended to take 
into allow for cyclic paths containing reset transitions when the weight of all the plays following them is not arbitrarily close to $0$ and negative.

\begin{exa}\label{ex:negative-cycle}
  Notice that the game on the left in \figurename~\ref{fig:Ex_transition_2valeurs} 
  does not fulfil this hypothesis: indeed, the play
  $(\qloc_0,0) \moveto{1/2-\varepsilon,\trans_1} (\qloc_1,1/2-\varepsilon) 
    \moveto{1/2+\varepsilon,\trans_2} (\qloc_0,0)$ is a cyclic play that contains a transition with a reset, and of weight
  $-2 \varepsilon$ negative and arbitrarily close to $0$.
\end{exa}

\subsection{Contribution}
In this work, we use a different technique 
of~\cite{brihaye2021oneclock} to push the 
decidability frontier and prove that the value function is computable 
for all \WTG{s} (in particular the one 
of~\figurename~\ref{fig:Ex_transition_2valeurs}):

\begin{thm}
  \label{thm:complexity}
  For all \WTG{s} $\game$ and all locations $\qlocinit$, the mapping
  $\val \mapsto \Value_{\game}(\qlocinit,\val)$ is computable in time
  exponential in $|\QLocs|$ and $\maxWeight$. 
  %, and polynomial in $\maxWeightLoc$ and $\maxWeightFinal$.
\end{thm}

\begin{rem}
	The complexities of Theorems~\ref{thm:pseudopoly} and \ref{thm:complexity} 
	would be more traditionally considered as exponential and doubly-exponential 
	if weights of the \WTG were encoded in binary as usual. In this work, we thus 
	count the complexities as if all weights were encoded in unary and thus consider 
	$\maxWeight$ to be the bound of interest. For Theorem~\ref{thm:pseudopoly}, the 
	obtained bound is classically called \emph{pseudo-polynomial} in the literature. 
\end{rem}

The rest of this article gives the proof of Theorem~\ref{thm:complexity}. 
We fix a \WTG $\game$ and an initial location $\qlocinit$. We let
$\rgame = \langle\LocsMin,\LocsMax,\LocsT,\LocsUrg,\CTrans,
  \Cweight,\CweightT\rangle$ be its closure. We first use
Lemma~\ref{lem:region_ptg}, which allows us to deduce the result by
computing the value functions
$\val \mapsto \Value_{\rgame}((\qlocinit, I),\val)$ for all regions
$I$. Regions $I$ over which
$\val \mapsto \Value_{\rgame}((\qlocinit, I),\val)$ is constantly equal
to $+\infty$ or $-\infty$ are computable in polynomial time, as
explained in \cite{brihaye2021oneclock}. We
therefore remove them from $\rgame$ from now on. We now fix an initial
region~$I_{\mathsf i}$ and let $\locinit=(\qlocinit, I_{\mathsf i})$, and explain 
how to compute $\val \mapsto \Value_{\rgame}((\qlocinit, I_{\mathsf i}),\val)$ 
on the interval $I_{\mathsf i}$. 

As in the non-negative case \cite{BouyerLarsenMarkeyRasmussen-06}, the
objective is to limit the number of transitions with a reset taken
into the plays while not modifying the value of the game. When all
weights are non-negative, this is fairly easy to achieve since,
intuitively speaking, $\MinPl$ has no interest in using any cycles
containing such a transition (since it has non-negative weight and is
thus non-beneficial for $\MinPl$). The game can thus be transformed so
that each transition with a reset is taken at most once. To obtain a
smaller game, it is even possible to simply count the number of
transitions with a reset taken so far in the play and stop the game
(with a final weight $+\infty$) in case the counter goes above the
number of such transitions in the game. The transformed game has
a polynomial number of locations with respect to the original game
and is reset-acyclic, which allows one to solve it by using Theorem~\ref{thm:pseudopoly}, with
a time complexity polynomial in $|\QLocs|$ and $\maxWeight$ (instead of the exponential time
complexity originally achieved in
\cite{BouyerLarsenMarkeyRasmussen-06,Rutkowski-11} with respect to $|\QLocs|$).

The situation is much more intricate in the presence of negative
weights since negative cycles containing a transition with a reset
can be beneficial for $\MinPl$, as we have seen in
Example~\ref{ex:value-example}. Notice that this is still true in the
closure of the game, as can be checked on the right in
\figurename~\ref{fig:Ex_transition_2valeurs}. Moreover, some cyclic paths may
have both plays following it with a positive weight and plays following it with a negative weight, making it difficult to determine whether it is
beneficial to $\MinPl$ (or not). To overcome this situation, we will
consider the point of view of $\MaxPl$, benefiting from the
determinacy of the \WTG. We will show that, in the closure 
$\rgame$, $\MaxPl$ can play \emph{optimally} with \emph{memoryless}
strategies while \emph{avoiding negative cyclic plays}. This will
simplify our further study since, by following this strategy, $\MaxPl$
ensures that only non-negative cyclic plays will be encountered, which 
is not beneficial to $\MinPl$. Therefore, as in
\cite{BouyerLarsenMarkeyRasmussen-06}, we will limit the firing of
transitions with a reset to at most once. However, we are not able to do it
without blowing up exponentially the number of locations of the games. 
Instead, along the unfolding of the
game, we need to record enough information in order to know, in case a
cyclic path ending with a reset is closed, whether this cyclic path
has a potential negative weight (in which case $\MaxPl$ will indeed
not follow it) or non-negative weight (in which case it is not
beneficial for $\MinPl$ to close the cycle). Determining in which case
we are will be made possible by introducing the notion of value of a 
cyclic path in Section~\ref{sec:negative-cycles}. Then, $\MaxPl$ has
even an optimal strategy to avoid closing cyclic paths with a negative
value (which is stronger than only avoiding creating negative cyclic
plays). The unfolding, denoted $\ugame$, will be defined in
Section~\ref{sec:unfolding}. Section~\ref{sec:negative-cycles-unfolding} 
shows that $\MaxPl$ keeps its ability to play without falling in negative "cycles" in the unfolding. This allows us to show in Section~\ref{sec:equal-values} that the unfolding game has a value equal to the closure game. This allows us to wrap up the proof of Theorem~\ref{thm:complexity} in Section~\ref{sec:end-proof}. 

\section{\texorpdfstring{How \MaxPl can control negative cycles}{How Max can control negative cycles}}
\label{sec:negative-cycles}

One of the main arguments of our proof is that, in the closure of a \WTG
$\rgame$, $\MaxPl$ can play \emph{optimally} with \emph{memoryless}
strategies while \emph{avoiding negative cyclic plays}. As already
noticed in~\cite{brihaye2021oneclock}, this is
not always true in all \WTG{s}: $\MaxPl$ may need memory to play
$\varepsilon$-optimally without the possibility of avoiding some negative
cyclic plays.

\begin{figure}
	\centering
	\begin{tikzpicture}[xscale=.8,every node/.style={font=\footnotesize}, 
	every label/.style={font=\scriptsize}]
	% Draw the states
	\node[PlayerMax, label={right:$\qloc_1$}] at (0, 1.75) (s1) {$\mathbf{-1}$};
	\node[PlayerMin, label={right:$\qloc_0$}] at (0, 0)  (s0) {$\mathbf{0}$};
	\node[target] at (0, -1.5) (s3) {$\LARGE \smiley$};
	
	% Connect the states with arrows
	\draw[->]
	(s0) edge[bend right=20,draw=red] node[right] {
		\textcolor{red}{$\begin{array}{c}
		\trans_1 \\ x \leq 2 \\ x := 0
		\end{array}$}} (s1)
	(s1) edge[bend right=20,draw=DarkOrchid] node[left] {
		\textcolor{DarkOrchid}{$\begin{array}{c}
		\trans_2 \\ x \leq 2 \\ \mathbf{1}
		\end{array}$}} (s0)
	(s0) edge[bend left=20, draw=ForestGreen] node[right] {
		\textcolor{ForestGreen}{$\begin{array}{c}
		\trans_3 \\ x \leq 2
		\end{array}$}} (s3)
	(s0) edge[bend right=20, draw=blue] node[left] {
		\textcolor{blue}{$\begin{array}{c}
		\trans_4 \\ x \leq 1 \\ \mathbf{-10}
		\end{array}$}} (s3)
	;

	\begin{scope}[xshift=9.5cm]
	\node[PlayerMax, label={above:$\qloc_1,\{0\}$}] at (0, 1.75) (s1) {$\mathbf{-1}$};
	\node[PlayerMin, label={right:$\qloc_0, \{0\}$}] at (-6 ,0) (s01) {$\mathbf{0}$};
	\node[PlayerMin, label={right:$\qloc_0, (0, 1)$}] at (-3,0)  (s02) {$\mathbf{0}$};
	\node[PlayerMin, label={right:$\qloc_0, \{1\}$}] at (0, 0)  (s03) {$\mathbf{0}$};
	\node[PlayerMin, label={right:$\qloc_0, (1, 2)$}] at (3,0)  (s04) {$\mathbf{0}$};
	\node[PlayerMin, label={below:$\qloc_0, \{2\}$}] at (6,0)  (s05) {$\mathbf{0}$};
	\node[target] at (0, -1.5) (s3) {$\LARGE \smiley$};
	
	% Connect the states with arrows
	\draw[->]
	(s01) edge[bend left=20,draw=red] (s1)
	(s02) edge[draw=red] (s1)
	(s03) edge[bend left=10,draw=red] (s1)
	(s04) edge[draw=red] (s1)
	(s05) edge[bend right=20,draw=red] (s1)
	
	(s1) edge[bend right=30,draw=DarkOrchid] 
		node[above left,yshift=-.15cm] 
		{\textcolor{DarkOrchid}{$x = 0; \mathbf{1}$}} (s01)
	(s1) edge[bend right=10,draw=DarkOrchid] 
		node[left,near end] {\textcolor{DarkOrchid}{$0 \leq x \leq 1; \mathbf{1}$}} (s02)
	(s1) edge[bend left=5,draw=DarkOrchid] 
		node[near end,right] {\textcolor{DarkOrchid}{$x = 1; \mathbf{1}$}} (s03)
	(s1) edge[bend left=10,draw=DarkOrchid] 
		node[near end, right,yshift=.1cm] 
		{\textcolor{DarkOrchid}{$1 \leq x \leq 2; \mathbf{1}$}} (s04)
	(s1) edge[bend left=30,draw=DarkOrchid] 
		node[above right,yshift=-.15cm] 
		{\textcolor{DarkOrchid}{$x = 2; \mathbf{1}$}} (s05)
	
	(s01) edge[bend right=2.5,draw=ForestGreen] (s3)
	(s02) edge[bend left=5,draw=ForestGreen] (s3)
	(s03) edge[bend left=10,draw=ForestGreen] (s3)
	(s04) edge[bend left=5,draw=ForestGreen] (s3)
	(s05) edge[bend left=10,draw=ForestGreen] (s3)
	
	(s01) edge[draw=blue, bend right=10] 
		node[below, near start, xshift=-.25cm] {\textcolor{blue}{$\mathbf{-10}$}} (s3)
	(s02) edge[draw=blue, bend right=5] 
		node[left, near start,xshift=-.25cm] {\textcolor{blue}{$\mathbf{-10}$}} (s3)
	(s03) edge[draw=blue, bend right=10]node[left, near start] 
	{\textcolor{blue}{$\mathbf{-10}$}} (s3)
	;
	\end{scope}
	
	\end{tikzpicture}
	\caption{On the left, a \WTG where $\MaxPl$ needs memory to play
          $\varepsilon$-optimally. On the right, its closure 
          where we merged several transitions by removing unnecessary 
          guards.}
	\label{fig:Max_memless}
\end{figure}
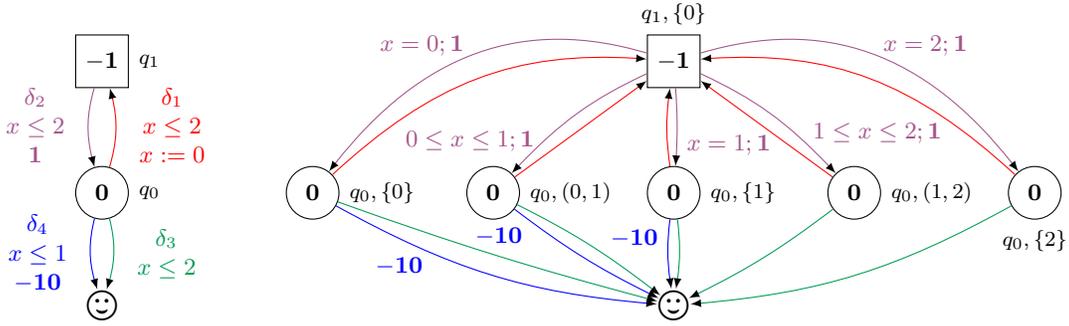

\begin{exa} \label{example:Max-not-optimal} 
	In the \WTG depicted on the left in \figurename~\ref{fig:Max_memless}, 
	we can see that $\Value(\qloc_1, 0) = 0$, but $\MaxPl$ does not have an optimal 
	strategy, needs memory to play $\varepsilon$-optimally, and cannot avoid negative 
	cyclic plays. Indeed, an optimal strategy for \MaxPl always chooses a delay greater 
	than $1$: if at some point, a strategy of \MaxPl chooses a delay less than or equal 
	to~$1$, then \MinPl can always choose $\trans_4$, and the value of this strategy 
	is $-10$. However, \MaxPl must choose a delay closer and closer to $1$. Otherwise, 
	we suppose that there exists $\beta > 0$ such that all delays chosen by the 
	strategy are greater than $1 + \beta$, and $\MinPl$ has a family of strategies 
	that stay longer and longer in the cycle with a weight at most $-\beta$. Thus, 
	the value of this strategy will tend to $-\infty$. In particular, \MaxPl does not 
	have an optimal strategy, and the $\varepsilon$-optimal strategy requires infinite 
	memory to play with delays closer and closer to $1$ (for instance, after the 
	$n$-th round in the cycle, Max delays $\varepsilon/2^n$ time units, to sum
	up, all weights to a value at most $-\varepsilon$).
\end{exa}

Such convergence phenomena needed by $\MaxPl$ do not exist in $\rgame$
since all guards are closed (this is not sufficient alone though) 
and by the regularity of
$\Value$ given by Lemma~\ref{lem:val_continue}.

\begin{exa}
	We consider the closure of the \WTG depicted in
    \figurename~\ref{fig:Max_memless}. The $\varepsilon$-optimal strategy
    (with memory) of $\MaxPl$ in $\game$ translates into an
    optimal memoryless strategy in $\rgame$: in~$(\qloc_1, \{0\})$,
    $\MaxPl$ can delay $1$ time unit and jump into the location
    $(\qloc_0, (1, 2))$. Then, cyclic plays that $\MinPl$
    can create have a zero weight and are thus not profitable for
    either~player.
\end{exa}

To generalise this explanation, we start by defining the value of
cyclic paths ending with a reset in a given \WTG. Intuitively, the value
of this cyclic path is the weight that $\MinPl$ (or $\MaxPl$) can
guarantee regardless of the delays chosen by $\MaxPl$ (or $\MinPl$)
during this one.

\begin{defi}
  \label{def:valCycle}
  Let $\game$ be a \WTG. We define by 
  induction the \emph{value} $\Value_\game^\val(\rpath)$ of a finite 
  path $\rpath$ in $\game$ from an initial valuation $\val$ of the clock: 
  if $\rpath$ has length $0$ (i.e.~if $\rpath \in \QLocs$), we let:
  \begin{displaymath}
  \Value_\game^\val(\rpath)=0 \,.
  \end{displaymath}
  Otherwise, $\rpath$ can be written
  $\qloc_0 \movetoPath{\trans_0} \rpath'$ (with $\rpath'$ starting in
  location $\qloc_1$), and we~let:
  \begin{displaymath}
  \Value_\game^\val(\rpath) = 
  \begin{cases}
      \inf_{(\qloc_0,\val) \moveto{\delay_0, \trans_0} (\qloc_1, \val') \text{ edge of } \sem{\game}} \big(\delay_0 \, \weight(\qloc_0) + \weight(\trans_0) +
      \Value_\game^{\val'}(\rpath')\big) & \text{if } \qloc_0 \in \LocsMin\\
      \sup_{(\qloc_0,\val) \moveto{\delay_0, \trans_0} (\qloc_1, \val') \text{ edge of } \sem{\game}} \big(\delay_0 \, \weight(\qloc_0) + \weight(\trans_0) +
      \Value_\game^{\val'}(\rpath')\big) & \text{if }
      \qloc_0 \in \LocsMax
  \end{cases}
  \end{displaymath} 
  Then, for a cyclic path $\rpath$ of $\game$ 
  ending by a transition with a reset, we let
  $\Value_{\game}(\rpath) = \Value_\game^0(\rpath)$.
\end{defi}

\begin{exa}
  \label{ex:value_cycle}
  Let $\rpath = \qloc_0 \movetoPath{\trans_1} \qloc_1 \movetoPath{\trans_2} 
  \qloc_0$ be the cyclic path of the \WTG $\game$ depicted on the left in 
  \figurename~\ref{fig:Ex_transition_2valeurs}. To evaluate the value of
  $\rpath$, $\MinPl$ only needs to choose a delay $\delay_1 \in [0, 1]$ when
  firing $\trans_1$, while $\MaxPl$ has no choice but to play a delay
  $1 - \delay_1$ when firing~$\trans_2$, generating a finite play $\play$
  of weight $\weightC(\play) = 2t_1 - 1$. We deduce that
  $\Value_\game (\rpath) = \inf_{\delay_1 \in [0, 1]} (2\, \delay_1 - 1) = -1$ 
  (when $\MinPl$ chooses~$\delay_1= 0$).
\end{exa}

A cyclic path with a negative value ensures that \MinPl can always guarantee 
to obtain a cyclic play that follows it with a negative weight, even when there 
are other cyclic plays (that follow it) with a non-negative weight. It is 
exactly those cycles that are problematic for \MaxPl since \MinPl can benefit 
from them. We now show our key lemma: in the closure of a \WTG, \MaxPl can play 
optimally and avoid cyclic paths of negative value.

\begin{lem}
  \label{lem:MaxOpt}
  In a closure \WTG $\rgame$, 
  $\MaxPl$ has a memoryless optimal
  strategy~$\maxstrategy$ such that
  \begin{enumerate}
  \item\label{item:MaxOpt_Weight} all cyclic plays conforming to
    $\maxstrategy$ have a non-negative weight;
  \item\label{item:MaxOpt_Value} all cyclic paths ending by a reset
    conforming to $\maxstrategy$ have a non-negative value.
  \end{enumerate}
\end{lem}
\begin{proof}
  We use Theorem~\ref{theo:fixedpoint} to define the memoryless strategy
  $\maxstrategy$. Indeed, the identity
  $\Value_\rgame = \F(\Value_\rgame)$, applied over configurations
  belonging to $\MaxPl$, suggests a choice of transition and delay to play
  almost optimally. As $\F$ computes a supremum on the set of possible
  (transitions and) delays, this does not directly lead to a specific
  choice: in general, this would give rise to $\varepsilon$-optimal
  strategies and not an optimal one. This is where we rely on the
  continuity of $\Value_\rgame$ (Lemma~\ref{lem:val_continue}) on each
  closure of region to deduce that this supremum is indeed a
  maximum. More precisely, for $\loc\in \LocsMax$, we can write
  $\F(\Value_\rgame)(\loc, \val)$ as:
  \begin{displaymath}
  \max_{\trans\in \CTrans} \sup_{\delay \text{ s.t. }
      (\loc,\val) \moveto{\delay, \trans} (\loc',\val')}
    \big(\Cweight(\trans) + \delay\, \Cweight(\loc) + \Value_\rgame
    (\loc', \val')\big) \,. 
  \end{displaymath} 
  The guard of transition $\trans$ is the closure $\overline \reg$ of a 
  region $\reg \in \reggame$, therefore, $\delay$ is in a closed interval 
  $J$ of values such that $\val+\delay$ falls in $\overline \reg$. Notice that 
  $\val'$ is either $0$ if $\trans$ contains a reset or is $\val+\delay$: 
  in both cases, this is a continuous function of $\delay$. Relying on the 
  continuity of~$\Value_\rgame$, the mapping $\delay \in J \mapsto 
  \Cweight(\trans) + \delay\, \Cweight(\loc) + \Value_\rgame (\loc', \val')$ is 
  thus continuous over a compact set so that its supremum is indeed a maximum.

  We thus let the memoryless strategy $\maxstrategy$ be such that,
  for all configurations $(\loc, \val)$, $\maxstrategy(\loc, \val)$
  is chosen arbitrarily in:
  \begin{equation}
    \argmax_{\trans \in \CTrans} 
    \argmax_{\delay \text{ s.t. } \loc,\val\xrightarrow{\delay, \trans} \loc',\val'} 
    \big(\Cweight(\trans) +  \delay \, \Cweight(\loc)  +
    \Value_\rgame(\loc', \val')\big)
    \label{eq:max_opt}
  \end{equation}
  This mapping $\maxstrategy$ is then extended into a memoryless strategy, defining it over 
  finite plays by only
  considering  the last configuration of the play. To conclude the proof, 
  we show that $\maxstrategy$ is an optimal strategy that satisfies 
  the two properties of the lemma.  
  
  We first show that $\maxstrategy$ is an  optimal strategy by proving that
  $\Value_\rgame^{\maxstrategy}(\loc, \val) \geq \Value_\rgame(\loc,
  \val)$ for all configurations $(\loc, \val)$. In particular, we show
  that for all plays $\play$ from $(\loc, \val)$ conforming to
  $\maxstrategy$, we have $\weightP(\play) \geq \Value_\rgame(\loc,
  \val)$. We remark that if $\play$ does not reach
  $\LocsT$, then $\weightP(\play) = +\infty$, and the inequality is
  satisfied. Now, we suppose that $\play$ reaches $\LocsT$, and we
  reason by induction on the length of $\play$ to show that
  $\weightP(\play) \geq \Value_\rgame(\loc, \val)$ for all plays $\play$
  starting in a configuration $(\loc, \val)$ reaching $\LocsT$.
  If $\play$ has length $0$, it starts directly in $\loc \in \LocsT$, and 
  $\weightP(\play) = \CweightT(\loc, \val) = \Value_\rgame(\loc,
  \val)$.  Otherwise,
  $\play = (\loc, \val) \moveto{\delay, \trans} \play'$, with
  $\play'$ starting in a configuration $(\loc',\val')$.  In
  particular, by inductive hypothesis, we have:
  \begin{displaymath}
  \weightP(\play) = \rweight(\trans) + \delay \, \rweight(\loc) +
  \weightP(\play') \geq \rweight(\trans) + \delay \, \rweight(\loc)
  + \Value_\rgame(\loc', \val') \,.
  \end{displaymath}
  Now, if $\loc \in \LocsMin$, then we conclude by using that
  $\Value_\rgame$ is a fixed point of $\F$, i.e. 
  \begin{displaymath}
	\weightP(\play) \geq 
	\!\!\!\!\!\inf_{(\loc, \val) \moveto{\delay, \trans} (\loc',\val')} \!\!\!\!\!
	\big(\rweight(\trans) + \delay \, \rweight(\loc) + 
	\Value_\rgame(\loc', \val')\big)  
	= \Value_\rgame(\loc, \val) \,.
  \end{displaymath}
  Otherwise, we suppose that $\loc \in \LocsMax$ and $(\delay, \trans)$ is 
  defined by $\maxstrategy$. Thus, by~\eqref{eq:max_opt} and using again
  that $\Value_\rgame$ is a fixed point of $\F$, we obtain that
  \begin{displaymath}
  	\weightP(\play) \geq 
  	\!\!\!\!\! \sup_{(\loc, \val) \moveto{\delay, \trans} (\loc',\val')} \!\!\!\!\!
  	\big(\rweight(\trans) + \delay \, \rweight(\loc, \reg)  + 
  	\Value_\rgame(\loc', \val')\big)
  	= \Value_\rgame(\loc, \val) \,.
  \end{displaymath}
  
  Finally, we conclude the proof by showing that $\maxstrategy$ satisfies 
  the two properties of the~lemma.
  \begin{enumerate}
  	\item Let $\play = (\loc_1, \val_1) \moveto{\delay_1, \trans_1} \cdots
  	(\loc_k, \val_k) \moveto{\delay_k, \trans_k} (\loc_{k+1}, \val_{k+1}) = 
  	(\loc_1, \val_1)$ be a cyclic play conforming to $\maxstrategy$. 
  	We show that $\weightC(\play) \geq 0$ by claiming that for all 
  	$i\in \{1,\ldots,k\}$, 
  	\begin{equation}
  	\label{eq:MaxOpt-1}
  	\Value_\rgame(\loc_i, \val_i) \leq \rweight(\trans_i) + 
  	\delay_i \, \rweight(\loc_i) + \Value_\rgame(\loc_{i+1},
  	\val_{i+1})
  	\end{equation} 
  	Indeed by summing this inequality along $\play$, we obtain:
  	\begin{displaymath}
  	\sum_{i=1}^k \Value_\rgame(\loc_i, \val_i) 
  	\leq \sum_{i=1}^k \big(\rweight(\trans_i) + 
  	\delay_i\,\rweight(\loc_i) + 
  	\Value_\rgame(\loc_{i+1}, \val_{i+1})\big)
  	\end{displaymath}
  	i.e., since $\play$ is a cyclic play,
  	\begin{displaymath}
  	\weightC(\play) = \sum_{i=1}^k \big(\rweight(\trans_i) + 
  	\delay_i \, \rweight(\loc_i)\big) \geq 0 \,.
  	\end{displaymath}
  	To conclude this point, we show~\eqref{eq:MaxOpt-1}. 
  	For $i\in \{1,\ldots,k\}$, we distinguish two cases.
  	First, we suppose that $\loc_i \in \LocsMin$ and we conclude as
	$\Value_\rgame$ is a fixed point of $\F$: 
	\begin{align*}
	\Value_\rgame(\loc_i, \val_i)&= 
	\inf_{(\loc_i, \val_i) \moveto{\delay, \trans} (\loc',\val')}
	\big(\rweight(\trans) + \delay \, \rweight(\loc_i) + 
	\Value_\rgame(\loc', \val')\big) \\
	&\leq \rweight(\trans_i) + \delay_i \, \rweight(\loc_i) + 
	\Value_\rgame(\loc_{i+1}, \val_{i+1}) \,.
	\end{align*}
	Otherwise, $\loc_i \in \LocsMax$ then, as $\Value_\rgame$ is a fixed point
	of $\F$ and by using~\eqref{eq:max_opt}, we have:
	\begin{align*}
	\Value_\rgame(\loc_i, \val_i) &= 
	\sup_{(\loc_i, \val_i) \moveto{\delay, \trans} (\loc',\val')}
	\big(\rweight(\trans) + \delay \, \rweight(\loc_i) + 
	\Value_\rgame(\loc', \val')\big) \\
	&\leq \rweight(\trans_i) + \delay_i \, \rweight(\loc_i) + 
	\Value_\rgame(\loc_{i+1}, \val_{i+1}) \,.
	\end{align*}

  	\item Let $\rpath = \loc_0 \movetoPath{\trans_0} \loc_1 \cdots \loc_k 
  	\movetoPath{\trans_k} \loc_0$ be a cyclic path conforming to $\maxstrategy$ 
  	such that $\trans_k$ contains a reset. 
  	By grouping all infimum/supremum together in the
  	previous definition, we can see that $\ValueRG(\rpath)$ can be 
  	rewritten~as:
  	\begin{displaymath}
  	\inf_{\big(f_i\colon (t_0,\ldots,t_{i-1}) \mapsto 
  		t_i\big)_{\substack{0\leq i \leq k \\ \loc_i \in \LocsMin}}} \,
  	\sup_{\big(f_i\colon (t_0,\ldots,t_{i-1}) \mapsto
  		t_i\big)_{\substack{0\leq i \leq k \\ \loc_i \in \LocsMax}}}
  	\weightC(\play)
  	\end{displaymath} 
  	where $\play$ is the finite play $(\loc_0,0) \moveto{t_0=f_0,\trans_0} 
  	(\loc_1,\val_1) \moveto{t_1=f_1(t_0),\trans_1} \cdots
  	\moveto{t_k=f_k(t_0,\ldots,t_{k-1}),\trans_k} (\loc_0,0)$. Notice
  	that the mapping $f_i$, chosen by the player owning location $\loc_i$,
  	describes the delay before taking the transition $\trans_i$ as a
  	function of the previously chosen delays.
  	In particular, for all $\varepsilon>0$, there exists
  	$\smash{\big(f_i\colon (t_0,\ldots,t_{i-1}) \mapsto
  	t_i\big)_{\substack{0\leq i \leq k \\ \loc_i \in \LocsMin}}}$ such
  	that for all $\big(f_i\colon (t_0,\ldots,t_{i-1}) \mapsto
  	t_i\big)_{\substack{0\leq i \leq k \\ \loc_i \in \LocsMax}}$:
  	\begin{displaymath}
  	\weightC(\play)\leq \ValueRG(\rpath)+\varepsilon
  	\end{displaymath}
  	with $\play$ the finite play described above. Since $\rpath$ is
  	conforming to $\maxstrategy$, a particular choice of delays
  	$\big(f_i\colon (t_0,\ldots,t_{i-1}) \mapsto
  	t_i\big)_{\substack{0\leq i \leq k \\ \loc_i \in \LocsMax}}$ is
  	given by $\maxstrategy$ itself. In this case, the latter finite
  	play $\play$ is conforming to $\maxstrategy$. By the previous
  	item, we know that $\weightC(\play)\geq 0$, therefore:
  	\begin{displaymath}
  	\ValueRG(\rpath) \geq -\varepsilon \,.
  	\end{displaymath}
  	Since this holds for all $\varepsilon>0$, we deduce that
  	$\ValueRG(\rpath)\geq 0$ as expected.\qedhere
  \end{enumerate}
\end{proof}

As a side note, it is tempting to strengthen
Lemma~\ref{lem:MaxOpt}.\eqref{item:MaxOpt_Value} 
so as to ensure that all plays following a cyclic path ending by a reset conforming to $\maxstrategy$ have
a non-negative weight. Unfortunately, this does
not hold, as shown in the following example.

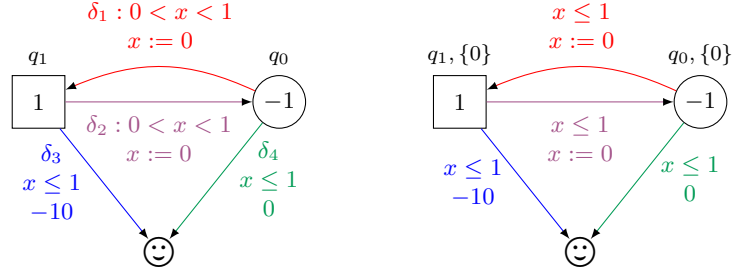
\begin{figure}
	\centering
	\begin{tikzpicture}[xscale=.8,every node/.style={font=\footnotesize}, 
	every label/.style={font=\scriptsize}]
	% Draw the states
	\node[PlayerMax, label={above:$\qloc_1$}] at (0, 0) (s1) {$1$};
	\node[PlayerMin, label={above:$\qloc_0$}] at (4, 0)  (s0) {$-1$};
	\node[target] at (2, -2) (s3) {$\LARGE \smiley$};
	
	% Connect the states with arrows
	\draw[->]
	(s0) edge[bend right=20,draw=red] node[above] {
		\textcolor{red}{$\begin{array}{c}
			\trans_1: 0 < x < 1 \\ x := 0
			\end{array}$}} (s1)
	(s1) edge[draw=DarkOrchid] node[below] {
		\textcolor{DarkOrchid}{$\begin{array}{c}
			\trans_2: 0 < x < 1 \\ x := 0
			\end{array}$}} (s0)
	(s1) edge[draw=blue] node[left] {
		\textcolor{blue}{$\begin{array}{c}
			\trans_3  \\ x \leq 1 \\ -10 
			\end{array}$}} (s3)
	(s0) edge[draw=ForestGreen] node[right] {
		\textcolor{ForestGreen}{$\begin{array}{c}
			\trans_4  \\ x \leq 1 \\ 0 
			\end{array}$}} (s3)
	;
	
	\begin{scope}[xshift=7cm]
	\node[PlayerMax, label={above:$\qloc_1, \{0\}$}] at (0, 0) (s1) {$1$};
	\node[PlayerMin, label={above:$\qloc_0, \{0\}$}] at (4, 0)  (s0) {$-1$};
	\node[target] at (2, -2) (s3) {$\LARGE \smiley$};
	
	% Connect the states with arrows
	\draw[->]
	(s0) edge[bend right=20, draw=red] node[above] {
		\textcolor{red}{$\begin{array}{c}
			x \leq 1\\ x:=0
			\end{array}$}} (s1)
	(s1) edge[draw=DarkOrchid] node[below] {
		\textcolor{DarkOrchid}{$\begin{array}{c}
			x \leq 1\\ x:= 0
			\end{array}$}} (s0)
	(s1) edge[draw=blue] node[left]{
		\textcolor{blue}{$\begin{array}{c}
			x \leq 1 \\ -10 
			\end{array}$}} (s3)
	(s0) edge[draw=ForestGreen] node[right]{
		\textcolor{ForestGreen}{$\begin{array}{c}
			x \leq 1 \\ 0 
			\end{array}$}} (s3)
	;
	\end{scope} 
	\end{tikzpicture}
	\caption{On the left, a \WTG such that its closure on the right
		contains a cyclic path of value $0$, but some cyclic paths of negative weight.  Moreover, $\MaxPl$ uses the cyclic path to play
		optimally.}
	\label{fig:cycleNeg}
\end{figure}

\begin{exa}
	We consider the closure of the \WTG depicted in 
	\figurename~\ref{fig:cycleNeg}. Let $\rpath = (\qloc_0, \{0\}) 
	\movetoPath{\trans_1} \allowbreak (\qloc_1, \{0\}) 
	\movetoPath{\trans_2} (\qloc_0, \{0\})$. It is a cyclic path such that plays following it have a weight in $[-1, 1]$. To evaluate 
	the value of $\rpath$ in $\rgame$, $\MinPl$ and $\MaxPl$ need to choose delays 
	$t_1, t_2\in [0, 1]$ when firing $\trans_1$ and $\trans_2$. We obtain 
	a set of finite plays $\play$ parametrised by $t_1$ and
	$t_2$ of weight $\weightC(\play) = -t_1 + t_2$. We deduce that 
	$\ValueRG (\rpath) = \inf_{t_1} \sup_{t_2} (t_2 - t_1) = 0$ (when
	$\MinPl$ and $\MaxPl$ choose $t_1 = t_2 = 1$). In particular, from 
	the configuration $((\qloc_0,\{0\}), 0)$, the cyclic path $\rpath$ 
	is not interesting for \MinPl since he only can guarantee the weight 
	$0$. Thus, he must play transition $\trans_4$ after a delay of $1$ 
	unit of time to lead to a value of $-1$. To play optimally, \MaxPl 
	must avoid the transition $\trans_3$, i.e.~all optimal strategies of 
	\MaxPl play in the previous cyclic path $\rpath$ that has a non-negative value but such that certain plays following it have a negative weight. 
\end{exa}
   
Finally, we note that Lemma~\ref{lem:MaxOpt} does not allow us to 
conclude on the decidability of the value problem since we use the 
unknown value $\Value_\rgame$ to define the optimal strategy.

\section{Definition of the unfolding}\label{sec:unfolding}

To compute $\Value_\rgame$, we now define the partial \emph{unfolding} 
of the \WTG $\rgame$ by allowing only one 
occurrence of each cyclic path (from $\rgame$) ending by a reset. In 
particular, when a transition with a reset is taken 
for the first time, we go into a new copy of the \WTG, from which, 
if this transition happens to be chosen one more time, we stop the 
play by jumping into a new target location. The final weight of this 
target location is determined by the value of the cyclic path (ending 
with a reset) that would have just been closed. If the cyclic path has 
a negative value, then we go in a leaf  $\rneg$ of final weight 
$-\infty$ since this is a desirable cycle for \MinPl. Otherwise, we go 
in a leaf $\rpos$ of final weight big enough to be an undesirable 
behaviour for \MinPl, i.e.~$|\Locs|(\maxWeightTrans+ \clockbound \, 
\maxWeightLoc) + \maxWeightFinal$ (for technical reasons that will 
become clear later, we can not simply put a final weight $+\infty$).

A single transition with a reset can be part of two distinct cyclic 
paths, one of negative value and the other of non-negative value, 
as demonstrated in Example~\ref{ex:Ex_cumule-value}. Thus, knowing 
the last transition of the cycle is not enough to compute the value 
of the cyclic path. Instead, we need to record the whole cyclic path: 
copying the game (as done in the non-negative 
setting~\cite{BouyerLarsenMarkeyRasmussen-06}) is not enough. Our 
unfolding needs to remember the path followed so far: their locations 
are thus finite paths of $\rgame$.

\begin{exa}
	\label{ex:Ex_cumule-value}
	In \figurename~\ref{fig:cumule_value}, we have depicted a
	\WTG (left) and a portion of its closure (right), where $\utrans_2$
	is contained in a cyclic path of negative value:
	\begin{displaymath}
	(\qloc_0, \{0\}) \movetoPath{\utrans_3} (\qloc_2, \{0\}) \movetoPath{\utrans'_4} 
	(\qloc_0, \{1\}) \movetoPath{\utrans'_1} (\qloc_1, \{0\}) \movetoPath{\utrans_2} 
	(\qloc_0, \{0\})
	\end{displaymath}
	and another cyclic path of non-negative (zero) value:
	\begin{displaymath}
	(\qloc_0, \{0\}) \movetoPath{\utrans_1} (\qloc_1, \{0\}) \movetoPath{\utrans_2} 
	(\qloc_0, \{0\}) \,.
	\end{displaymath}
\end{exa}

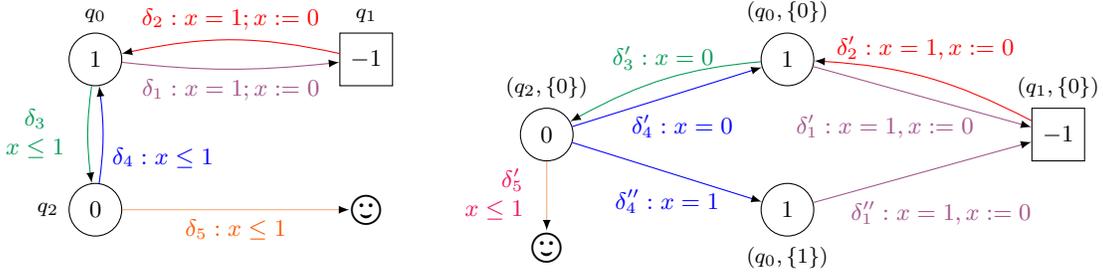
\begin{figure}[tbp]
	\centering
	\begin{tikzpicture}[xscale=.8,every node/.style={font=\footnotesize}, 
	every label/.style={font=\scriptsize}]
	% Draw the states
	\node[PlayerMax, label={above:$\qloc_1$}] at (4.5, 0) (s1) {$-1$};
	\node[PlayerMin, label={above:$\qloc_0$}] at (0, 0)  (s0) {$1$};
	\node[PlayerMin, label={left:$\qloc_2$}] at (0, -2)  (s2) {$0$};
	\node[target] at (4.5, -2) (s3) {$\LARGE \smiley$};
	
	% Connect the states with arrows
	\draw[->]
	(s0) edge[bend right=5,draw=DarkOrchid] node[below] {
		\textcolor{DarkOrchid}{$\trans_1: x = 1; x := 0$}} (s1)
	(s1) edge[bend right=10,draw=red] node[above] {
		\textcolor{red}{$\trans_2 : x = 1; x := 0$}} (s0)
	(s2) edge[draw=Orange!70] node[below] {
		\textcolor{orange!70!red}{$\trans_5 : x \leq 1$}} (s3)
	(s0) edge[bend right=10,draw=ForestGreen] node[left] {
		\textcolor{ForestGreen}{$\begin{array}{c}
			\trans_3 \\ x \leq 1
			\end{array}$}} (s2)
	(s2) edge[bend right=10,draw=blue] node[right, near start] {
		\textcolor{blue}{$\trans_4 : x \leq 1$}} (s0)
	;
	
	\begin{scope}[xshift=9.5cm]
	\node[PlayerMax, label={above:$(\qloc_1, \{0\})$}] at (6.5, -1) (s1) {$-1$};
	\node[PlayerMin, label={above:$(\qloc_0, \{0\})$}] at (2, 0)  (s01) {$1$};
	\node[PlayerMin, label={below:$(\qloc_0, \{1\})$}] at (2, -2)  (s03) {$1$};
	\node[PlayerMin, label={above:$(\qloc_2,\{0\})$}] at (-2, -1)  (s21) {$0$};
	\node[target] at (-2, -2.5) (s31) {$\LARGE \smiley$};
	
	% Connect the states with arrows
	\draw[->]
	(s01) edge[draw=DarkOrchid] node[below, xshift=-.5cm,yshift=-.1cm] 
		{\textcolor{DarkOrchid}{$\utrans_1: x=1, x:=0$}} (s1)
	(s03) edge[draw=DarkOrchid] node[below,xshift=.25cm,yshift=-.25cm]  	
		{\textcolor{DarkOrchid}{$\utrans'_1: x=1, x:=0$}} (s1)
	(s1) edge[bend right=10,draw=red] 
	node[above,yshift=.1cm] {\textcolor{red}{$\utrans_2: x=1, x:=0$}} (s01)
	(s01) edge[bend right=10,draw=ForestGreen] node[above]
		{\textcolor{ForestGreen}{$\utrans_3: x = 0$}} (s21)
	(s21) edge[draw=blue] 
	node[below,xshift=.2cm,yshift=-.1cm] {\textcolor{blue}{$\utrans_4: x = 0$}} (s01)
	(s21) edge[draw=Orange!70] node[left] 
	{\textcolor{OrangeRed}{$\begin{array}{r}
		\trans_5' \\ x \leq 1
		\end{array}$}} (s31)
	(s21) edge[draw=blue] node[below,yshift=-.1cm] 
		{\textcolor{blue}{$\utrans'_4: x = 1$}} (s03)
	;
	
	\end{scope}
	\end{tikzpicture}
	\caption{A \WTG (left) and a portion of its closure (right)
		where $\utrans_2$ belongs to a cyclic path of non-negative
		value and another cyclic path of negative value.}
	\label{fig:cumule_value}
\end{figure}

In order to obtain a finite acyclic unfolding, we also need to stop cyclic 
paths without resets. To do so, we will rely on a property of 
reset-acyclic \WTG{s}. For such \WTG{s}, it can be 
shown the existence of an $\varepsilon$-optimal strategy for \MinPl 
with a particular shape \cite{brihaye2021oneclock} defined as follows: 

\begin{defiC}
	\label{def:switching}
	A \emph{switching} strategy $\minstrategy$ is described by two
	memoryless strategies $\minstrategy^1$, and~$\minstrategy^2$, 
	as well as a switching threshold $\bornePseudoPoly'$. The strategy 
	$\minstrategy$ then consists in playing strategy $\minstrategy^1$ 
	until either we reach a target location or the finite play has 
	a length of at least $\bornePseudoPoly'$, in which case we switch 
	to strategy~$\minstrategy^2$.
\end{defiC}

Intuitively, $\minstrategy^1$ aims at reaching a cyclic play with negative 
weight, while $\minstrategy^2$ is an attractor to the target. As a  
consequence, we can estimate the maximal number of steps needed 
by~$\minstrategy^2$ to reach the target. Combining this with the 
switching threshold $\bornePseudoPoly'$, we can deduce a threshold 
$\bornePseudoPoly$ that upper bounds the number of
steps under the switching strategy $\minstrategy$ to reach the target. 
Moreover, we can explicitly give the pseudo-polynomial bound $\bornePseudoPoly$ 
since it is given by the previous work of \cite{brihaye2021oneclock}. 
From~\cite[Lemma 3.9]{brihaye2021oneclock}, we know~that 
\begin{displaymath}
\bornePseudoPoly' = O(|\Locs| \times (\maxWeightLoc + |\minstrategy^1| \times 
\maxWeightTrans|\Locs|) + |\minstrategy^1|)
\end{displaymath}
where $|\minstrategy^1|$ is the size of this strategy, i.e.~the number 
of cutpoints in $\ValueRG$ (by~\cite[Theorem 5.9]{brihaye2021oneclock}). Moreover, 
by~\cite[Theorem 5.13]{brihaye2021oneclock}, we have a bound over the number of 
cutpoints in $\ValueRG$, i.e.~$|\minstrategy^1| = O(\maxWeightTrans^4|\Locs|^9)$.
Thus, we deduce that the switching threshold $\bornePseudoPoly'$ is approximated by 
\begin{displaymath}
\bornePseudoPoly' = 
O\left(|\Locs| \times \left[\maxWeightLoc + \maxWeightTrans^4 |\Locs|^9 
\times \maxWeightTrans|\Locs|\right] + \maxWeightTrans^4 |\Locs|^9\right) = 
O\left(|\Locs|^{11} (\maxWeightLoc + \maxWeightTrans^5)\right) \,.
\end{displaymath}
Then, we fix $\bornePseudoPoly''$ to be the number of turns taken by $\minstrategy^2$ 
to reach the target location, which is polynomial in the number of locations of 
the region automaton underlying the game, thus polynomial in the number of 
locations of the game (since there is only one-clock). Overall, this gives a 
definition for $\bornePseudoPoly$ as:
\begin{displaymath}
\bornePseudoPoly = \bornePseudoPoly' + \bornePseudoPoly'' = 
O\left(|\Locs|^{12} (\maxWeightLoc + \maxWeightTrans^5)\right)
\end{displaymath} 
that is polynomial in $|\QLocs|$ (as $|\Locs|$ is polynomial 
in $|\QLocs|$) and in $\maxWeight$. Thus, we obtain the following result.

\begin{lemC}[\cite{brihaye2021oneclock}]
	\label{lem:switching}
	Let $\game$ be a reset-acyclic \WTG. $\MinPl$ has an 
	$\varepsilon$-optimal switching strategy~$\minstrategy$ such that 
	all plays conforming to $\minstrategy$ reach the target within 
	$\bornePseudoPoly$ steps. Moreover, $\bornePseudoPoly$ is 
	polynomial in~$|\QLocs|$ and $\maxWeight$.
\end{lemC}

\begin{figure}
	\centering
	\begin{tikzpicture}
	\draw[fill=gray!20] (-.2,-.6) circle (.9);
	\draw[fill=gray!20] (-3.5,-2.55) circle (.9);
	\draw[fill=gray!20] (3.35,-2.5) circle (.9);
	\draw[fill=gray!20] (-2.1,-4.6) circle (.95);
	\draw[fill=gray!20] (4.4,-4.6) circle (.9);
	
	\begin{scope}[scale=.3]
	\node[PlayerMinmin] at (0, 0) (s1) {};
	\node[PlayerMaxmin] at (-1, -1) (s2) {};
	\node[PlayerMinmin] at (1, -1) (s3) {};
	\node[PlayerMinmin] at (-1.75, -2.25) (s4) {};
	\node[PlayerMinmin] at (-.25, -2.25) (s5) {};
	\node[PlayerMaxmin] at (1, -2.25) (s6) {};
	\node[PlayerMaxmin] at (-2.25, -3.5) (s7) {};
	\node[PlayerMaxmin] at (-.25, -3.5) (s9) {};
	\node[PlayerMinmin] at (-1.25, -3.5) (s10) {};
	
	\draw[-]
	(s1) edge (s2)
	(s1) edge (s3)
	(s2) edge (s4)
	(s2) edge (s5)
	(s3) edge (s6)
	(s4) edge (s7)
	(s4) edge (s10)
	(s5) edge (s9)
	;
	
	\path[vecArrow]
	(s1) edge [decorate] (s2)
	(s1) edge [decorate] (s3)
	(s2) edge [decorate] (s4)
	(s2) edge [decorate] (s5)
	(s3) edge [decorate] (s6)
	(s4) edge [decorate] (s7)
	(s5) edge [decorate] (s9)
	(s4) edge [decorate] (s10)
	;
	\end{scope}
	
	\begin{scope}[xshift=-3.5cm,yshift=-2cm,scale=.3]
	\node[PlayerMinmin] at (0, 0) (s1) {};
	\node[PlayerMaxmin] at (-1, -1) (s2) {};
	\node[PlayerMinmin] at (1, -1) (s3) {};
	\node[PlayerMinmin] at (-1.75, -2.25) (s4) {};
	\node[PlayerMinmin] at (-.1, -2.25) (s5) {};
	\node[PlayerMaxmin] at (2, -1) (s6) {};
	\node[PlayerMaxmin] at (-1.75, -3.5) (s7) {};
	\node[PlayerMaxmin] at (1.1, -2.25) (s9) {};
	\node[PlayerMinmin] at (-.65, -3.5) (s10) {};
	
	\draw[-]
	(s1) edge (s2)
	(s1) edge (s3)
	(s2) edge (s4)
	(s2) edge (s5)
	(s3) edge (s6)
	(s4) edge (s7)
	(s4) edge (s10)
	(s5) edge (s9)
	;
	
	\path[vecArrow]
	(s1) edge [decorate] (s2)
	(s1) edge [decorate] (s3)
	(s2) edge [decorate] (s4)
	(s2) edge [decorate] (s5)
	(s3) edge [decorate] (s6)
	(s4) edge [decorate] (s7)
	(s5) edge [decorate] (s9)
	(s4) edge [decorate] (s10)
	;
	\end{scope}
	
	\begin{scope}[xshift=3.5cm,yshift=-2cm,scale=.3]
	\node[PlayerMaxmin] at (0, 0) (s1) {};
	\node[PlayerMinmin] at (1, -1) (s2) {};
	\node[PlayerMinmin] at (-1.25, 0) (s3) {};
	\node[PlayerMinmin] at (.25, -2) (s4) {};
	\node[PlayerMaxmin] at (1.75, -2) (s5) {};
	\node[PlayerMaxmin] at (-1.25, -1) (s6) {};
	\node[PlayerMinmin] at (.25, -3) (s7) {};
	\node[PlayerMaxmin] at (-1.25, -3) (s8) {};
	\node[PlayerMinmin] at (-2.5, -1) (s9) {};
	\node[PlayerMaxmin] at (-1.25, -2) (s10) {};
	
	\draw[-]
	(s1) edge (s2)
	(s1) edge (s3)
	(s2) edge (s4)
	(s2) edge (s5)
	(s3) edge (s6)
	(s4) edge (s7)
	(s4) edge (s10)
	(s7) edge (s8)
	(s6) edge (s9)
	;
	
	\path[vecArrow]
	(s1) edge [decorate] (s2)
	(s1) edge [decorate] (s3)
	(s2) edge [decorate] (s4)
	(s2) edge [decorate] (s5)
	(s3) edge [decorate] (s6)
	(s4) edge [decorate] (s7)
	(s4) edge [decorate] (s10)
	(s6) edge [decorate] (s9)
	(s7) edge [decorate] (s8)
	;
	\end{scope}
	
	\begin{scope}[xshift=-2cm,yshift=-4cm,scale=.3]
	\node[PlayerMaxmin] at (0, 0) (s1) {};
	\node[PlayerMinmin] at (-1, -1) (s2) {};
	\node[PlayerMinmin] at (1, -1) (s3) {};
	\node[PlayerMinmin] at (-1.75, -2) (s4) {};
	\node[PlayerMaxmin] at (-.25, -2) (s5) {};
	\node[PlayerMaxmin] at (1, -2) (s6) {};
	\node[PlayerMinmin] at (-1.75, -3) (s7) {};
	\node[PlayerMaxmin] at (-1.75, -4) (s8) {};
	\node[PlayerMaxmin] at (-.75, -3) (s10) {};
	\node[PlayerMaxmin] at (2.25, -2) (s9) {};
	
	\draw[-]
	(s1) edge (s2)
	(s1) edge (s3)
	(s2) edge (s4)
	(s2) edge (s5)
	(s3) edge (s6)
	(s4) edge (s7)
	(s7) edge (s8)
	(s4) edge (s10)
	(s6) edge (s9)
	;
	
	\path[vecArrow]
	(s1) edge [decorate] (s2)
	(s1) edge [decorate] (s3)
	(s2) edge [decorate] (s4)
	(s2) edge [decorate] (s5)
	(s3) edge [decorate] (s6)
	(s4) edge [decorate] (s7)
	(s4) edge [decorate] (s10)
	(s6) edge [decorate] (s9)
	(s7) edge [decorate] (s8)
	;
	\end{scope}
	
	\begin{scope}[xshift=4.5cm,yshift=-4cm,scale=.3]
	\node[PlayerMinmin] at (0, 0) (s1) {};
	\node[PlayerMaxmin] at (-1, -1) (s2) {};
	\node[PlayerMinmin] at (1, -1) (s3) {};
	\node[PlayerMinmin] at (-.25, -2.25) (s4) {};
	\node[PlayerMinmin] at (-1.5, -2.25) (s5) {};
	\node[PlayerMaxmin] at (1, -2.25) (s6) {};
	\node[PlayerMaxmin] at (0, -3.5) (s7) {};
	\node[PlayerMaxmin] at (-2.75, -2.25) (s9) {};
	\node[PlayerMinmin] at (-1, -3.5) (s10) {};
	
	\draw[-]
	(s1) edge (s2)
	(s1) edge (s3)
	(s2) edge (s4)
	(s2) edge (s5)
	(s3) edge (s6)
	(s4) edge (s7)
	(s4) edge (s10)
	(s5) edge (s9)
	;
	
	\path[vecArrow]
	(s1) edge [decorate] (s2)
	(s1) edge [decorate] (s3)
	(s2) edge [decorate] (s4)
	(s2) edge [decorate] (s5)
	(s3) edge [decorate] (s6)
	(s4) edge [decorate] (s7)
	(s4) edge [decorate] (s10)
	(s5) edge [decorate] (s9)
	;
	\end{scope}

	\node[target,scale=.75] at (0,-4.75) (target) {\Huge $\smiley$};
	\node[leaf,scale=.75] at (-4, -6) (rneg) {\small $\rneg$};
	\node[leaf,scale=.75] at (1.5, -6) (rpos) {\small $\rpos$};
	\node[leaf,scale=.75] at (-.1, -2.5) (revil) {$\revil$};

	%Bubble 1
	\node[PlayerMaxmin,draw=none] at (-0.66,-1.09) (s1) {};
	\node[PlayerMaxmin,draw=none] at (-.37,-1.05) (s01) {};
	\node[PlayerMaxmin,draw=none] at (-.09,-1.05) (s02) {};
	\node[PlayerMinmin,draw=none] at (-3.5,-2) (s2) {};
	\node[PlayerMinmin,draw=none] at (.3,-.3) (s3) {};
	\node[PlayerMaxmin,draw=none] at (3.5,-2) (s4) {};
	\node[PlayerMaxmin,draw=none] at (.3,-.675) (s5) {};
	
	%Bubble 2
	\node[PlayerMaxmin,draw=none] at (-4.025,-3.05) (s6) {};
	\node[PlayerMaxmin,draw=none] at (-3.7,-3.045) (s61) {};
	\node[PlayerMaxmin,draw=none] at (-3.175,-2.7) (s62) {};
	\node[PlayerMinmin,draw=none] at (-3.2,-2.3) (s8) {};
	\node[PlayerMaxmin,draw=none] at (-2,-4) (s9) {};
	\node[PlayerMaxmin,draw=none] at (-2.9,-2.3) (s10) {};
	
	%Bubble 3
	\node[PlayerMaxmin,draw=none] at (4,-2.6) (s11) {};
	\node[PlayerMaxmin,draw=none] at (3.13,-2.615) (s111) {};
	\node[PlayerMaxmin,draw=none] at (2.75,-2.3) (s112) {};
	\node[PlayerMinmin,draw=none] at (3.55,-2.905) (s12) {};
	\node[PlayerMinmin,draw=none] at (4.5,-4) (s13) {};
	\node[PlayerMaxmin,draw=none] at (3.13,-2.9) (s14) {};
	
	%Bubble 4
	\node[PlayerMinmin,draw=none] at (-2.525,-4.9) (s15) {};
	\node[PlayerMaxmin,draw=none] at (-1.3,-4.6) (s151) {};
	\node[PlayerMaxmin,draw=none] at (-2.225,-4.9) (s152) {};
	\node[PlayerMaxmin,draw=none] at (-2.1,-4.6) (s16) {};
	\node[PlayerMaxmin,draw=none] at (-2.55,-5.2) (s17) {};
	
	%Bubble 5
	\node[PlayerMaxmin,draw=none] at (4.5,-5.05) (s18) {};
	\node[PlayerMaxmin,draw=none] at (3.7,-4.675) (s181) {};
	\node[PlayerMaxmin,draw=none] at (4.25,-5.05) (s182) {};
	\node[PlayerMinmin,draw=none] at (4.8,-4.29) (s20) {};
	\node[PlayerMaxmin,draw=none] at (4.8,-4.67) (s19) {};
	
	\draw[->]
	%Bubble 1
	(s1) edge[draw=DarkBlue] node[above] 
	{\scriptsize \textcolor{DarkBlue}{$\begin{array}{c}
			\trans_1 \\ x:= 0
			\end{array}$}} (s2)
	(s01) edge[draw=DarkOrchid, dashed] (revil.110) 
	(s02) edge[draw=DarkOrchid, dashed] (revil) 
	(s3) edge[draw=red] node[above,xshift=.1cm] 
	{\scriptsize \textcolor{red}{$\begin{array}{c}
			\trans_2 \\ x:= 0
			\end{array}$}} (s4)
	(s5) edge[draw=ForestGreen,bend left=30] (target)
	
	%Bubble 2
	(s6) edge[draw=DarkBlue] node[left] 
	{\scriptsize \textcolor{DarkBlue}{$\begin{array}{c}
			\trans_1 \\ x:= 0
			\end{array}$}} (rneg)	
	(s61) edge[draw=DarkOrchid, dashed] (revil.210) 
	(s62) edge[draw=DarkOrchid, dashed] (revil) 
	(s8) edge[draw=red] node[above,xshift=.5cm,yshift=-.2cm, near end] 
	{\scriptsize \textcolor{red}{$\begin{array}{c}
			\trans_2 \\ x:= 0
			\end{array}$}} (s9)
	(s10) edge[bend left,draw=ForestGreen] (target)
	
	%Bubble 3
	(s11) edge[bend left=20,draw=DarkBlue] node[right] 
	{\scriptsize \textcolor{DarkBlue}{$\begin{array}{c}
			\trans_1 \\ x:= 0
			\end{array}$}} (s13)
	(s111) edge[draw=DarkOrchid, dashed] (revil) 
	(s112) edge[draw=DarkOrchid, dashed] (revil.10) 
	(s12) edge[draw=red] node[left, yshift=-0.4cm,xshift=-.1cm] 
	{\scriptsize \textcolor{red}{$\begin{array}{c}
			\trans_2 \\ x:= 0
			\end{array}$}} (rpos)
	(s14) edge[bend right=15,draw=ForestGreen] (target)
	
	%Bubble 4
	(s15) edge[draw=red] node[below] 
	{\scriptsize \textcolor{red}{$\begin{array}{c}
			\trans_2 \\ x:= 0
			\end{array}$}} (rneg)
	(s151) edge[draw=DarkOrchid, dashed] (revil) 
	(s152) edge[draw=DarkOrchid, dashed, bend right=50] (revil) 
	(s16) edge[bend right,draw=DarkBlue] node[below,xshift=-.25cm] 
	{\scriptsize \textcolor{DarkBlue}{$\begin{array}{c}
			\trans_1 \\ x:= 0
			\end{array}$}} (rpos)
	(s17) edge[bend right=40,draw=ForestGreen] (target)
	
	%Bubble 5
	(s18) edge[draw=DarkBlue, bend left=05] node[below] 
	{\scriptsize \textcolor{DarkBlue}{$\begin{array}{c}
			\trans_1 \\ x:= 0
			\end{array}$}} (rpos)
	(s181) edge[draw=DarkOrchid, dashed] (revil.325) 
	(s182) edge[draw=DarkOrchid, dashed, bend left=10] (revil) 
	(s20) edge[draw=red, in=-50, out=-45] node[right] 
	{\scriptsize \textcolor{red}{$\begin{array}{c}
			\trans_2 \\ x:= 0
			\end{array}$}} (rpos)
	(s19) edge[in=-10, out=-90,draw=ForestGreen] (target)
	;
	\end{tikzpicture}
	\caption{Scheme of the unfolding of a closure of a \WTG.}
	\label{fig:unfolding}
\end{figure}
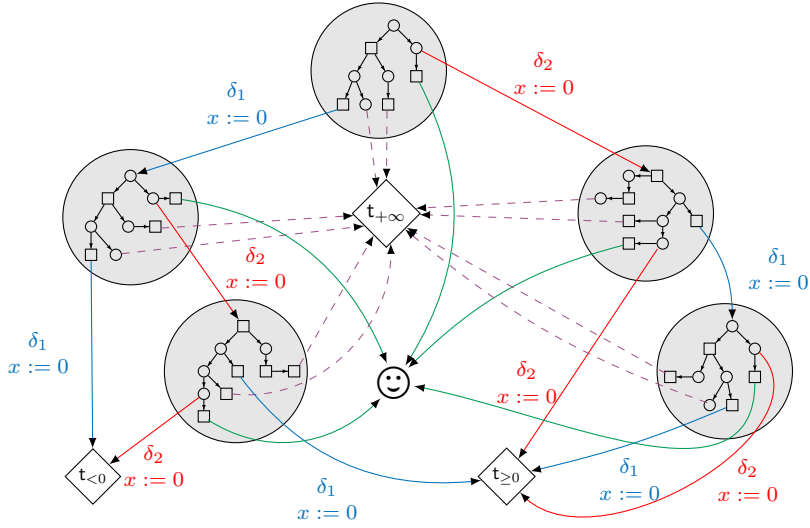

Now, between two transitions with a reset, we obtain a reset-acyclic \WTG. 
As a consequence, since \MinPl can play almost optimally using a switching 
strategy, we can bound the number of steps between two transitions with a 
reset by $\bornePseudoPoly$. This property allows us to avoid incorporating 
cycles in the unfolding: we cut the unfolding when the play becomes longer 
than $\bornePseudoPoly$ since the last seen transition with a reset. In 
this case, we will jump into a new target location, $\revil$, whose final 
weight is equal to $+\infty$ since it is an undesirable behaviour for \MinPl.

The scheme of the unfolding is depicted in \figurename~\ref{fig:unfolding} 
when the closure of a \WTG contains two transitions with a reset, 
$\trans_1$ and $\trans_2$, each belonging to several cycles of different 
values (negative and non-negative). Inside each grey component, only 
transitions with no reset are unfolded for at most $\bornePseudoPoly$ steps 
by only keeping, in the current location, the path followed so far. 
Transitions with a reset induce a change of components: these are in between 
the components. The second time they are visited, the value of the cycle 
it closes is computed, and we jump in $\rneg$ or $\rpos$ depending on the 
sign of the~value.

\begin{defi}
  \label{def:depliage}
  The \emph{unfolding} of $\rgame$ from the initial location $\locinit$ is the 
  (a priori infinite) \WTG 
  $\ugame = \langle\uLocsMin, \uLocsMax, \uLocsT, \uLocsUrg,
  \uTrans, \uweight, \uweightT\rangle$ with
  $\uLocsMin \subseteq \PPathsMin$, $\uLocsMax \subseteq \PPathsMax$,
  $\uLocsT \subseteq \LocsT \cup \{\rpos, \rneg, \revil\}$ such that 
  \begin{itemize} 
  \item $\uLocs = \uLocsMin\uplus\uLocsMax\uplus\uLocsT$ and $\uTrans$
    are the smallest sets such that $\locinit\in \uLocs$ and for all
    $\uloc\in\uLocsMin\uplus\uLocsMax$ and $\trans \in \Trans$, if
    $\textsc{Next}(\uloc, \trans) = (\uloc', \utrans)$ then
    $\uloc'\in\uLocs$ and $\utrans\in\uTrans$ (where $\textsc{Next}$
    is defined in Algorithm~\ref{algo:Next});
  \item
    $\uLocsUrg = \{\uloc\in\uLocs \mid \last(\uloc) \in \LocsUrg\}$;
  \item for all $\uloc \notin \uLocsT$,
    $\uweight(\uloc) = \Cweight(\last(\uloc))$;
  \item for all $\uloc \in \uLocsT$, for all
    $\val$,
    \begin{align*}
      \uweightT(\uloc, \val) &= \CweightT(\uloc, \val) \quad
                               \text{if } \uloc\in\LocsT & 
      \uweightT(\rpos,\val) &= |\Locs|(\maxWeightTrans+ \clockbound \,
                              \maxWeightLoc) + \maxWeightFinal\\
      \uweightT(\rneg,\val) &= - \infty &
     \uweightT(\revil,\val) &= + \infty \,.
    \end{align*}
  \end{itemize}
\end{defi}

As expected, the definition of \textsc{Next} guarantees that a ``new'' 
target location is reached when the length between two resets is too 
long or when a transition with a reset appears two times. Moreover, the 
length of the path in a location that is not a target, given by the 
application of $\textsc{Next}$, strictly increases.
This allows us to show that $\ugame$ is a finite and acyclic~\WTG.

\begin{algorithm}[tbp]
	\begin{algorithmic}[1] 		
		\Function{\textsc{Next}}{$\uloc, \trans = (\loc_1, \reg, \reset, w, \loc_2)$}:
		\Comment{\quad$\last(\uloc) = \loc_1$}
		\If{$\loc_2 \in \LocsT$} 
		$\uloc' := \loc_2$ \label{line:uloc_1} 
		\ElsIf{$\reset = \{\clockx\}$}
		\If{$|\uloc|_\trans = 0$}
		$\uloc' := \uloc \movetoPath{\trans} \loc_2$ \label{line:reset}
		\Else 
		\{ $\uloc := \uloc_1 \movetoPath{\trans} \uloc_2$\;
		\IfThenElse{$\Value_{\rgame}(\uloc_2 \movetoPath{\trans} \loc_2) \geq 0$}
		{$\uloc' := \rpos$\label{line:rpos}}
		{$\uloc' := \rneg$\label{line:rneg}}\}
		\EndIf
		\Else 
		\{ $\uloc := \uloc_1 \cdot \uloc_2$ where $\uloc_2$ contains no 
		reset and $|\uloc_2|$ is maximal\;
		\IfThenElse{$|\uloc_2| = \bornePseudoPoly$}
		{$\uloc' := \revil$ \label{line:revil2}}
		{$\uloc' := \uloc \movetoPath{\trans} \loc_2$\label{line:plusone2}}
		\}
		\EndIf
		\State $\utrans := (\uloc, \reg, \reset, w, \uloc')$ 
		\Comment{\quad $\Dproj(\utrans) := \trans$} \label{line:dproj}
		\State \Return $(\uloc', \utrans)$
		\EndFunction
	\end{algorithmic}
	\caption{Function $\textsc{Next}$ that maps pairs
		$(\uloc, \trans)\in\PPaths_{\rgame}\times \CTrans$ to pairs
		$(\uloc', \utrans)$ composed of a finite path $\uloc'$ of
		$\rgame$ (or $\rpos$, or $\rneg$, or $\revil$) and a new
		transition $\utrans$ of the unfolding $\ugame$.}
	\label{algo:Next}
\end{algorithm}

\begin{lem}
  \label{lem:sizeU}
  The \WTG $\ugame$ is acyclic and has a 
  finite set of locations of cardinality at most exponential in $|\QLocs|$ 
  and $\maxWeight$.
\end{lem}
\begin{proof}
	We start by proving that $\ugame$ is an acyclic \WTG. The function 
	\textsc{Next} never removes a transition from a path 
	$\uloc \in \uLocsMin \cup \uLocsMax$ that is given as input. In particular, 
	the function \textsc{Next} produces a transition from $\uloc$ to $\uloc'$ 
	that is an extension of $\uloc$ (i.e.~$|\uloc'| > |\uloc|$) or a target
	location. Thus, all paths in $\ugame$ are acyclic, i.e.~$\ugame$ is acyclic.
	
	Now, we prove that $\ugame$ is a finite \WTG by proving that 
	for all locations $\uloc \in \uLocsMin\cup\uLocsMax$, the length of the path is 
	upper-bounded by $(|\Transr| + 1)\bornePseudoPoly$, 
	where we let $\Transr$ be the subset of transitions with a reset. 
	Locations of $\ugame$ are 
	built by successive applications of $\textsc{Next}$. We 
	show by induction (on the number of such applications) that every 
	location $\uloc \in \uLocsMin \cup \uLocsMax$ can be decomposed as follows:
	$\uloc = \uloc'_0 \movetoPath{\delta_0} \uloc'_1 \movetoPath{\delta_1}
	\cdots \movetoPath{\delta_k} \uloc'_k$
	where transitions $\delta_i$ belong to~$\Transr$ and are pairwise distinct, 
	and where $\uloc'_i$ belong to $\PPaths_{\rgame}$ and have length at most
	$\bornePseudoPoly$. As a direct consequence of this property, we have 
	$k\leq |\Transr|$, and we easily deduce the expected bound on $|\uloc|$, and thus the desired bound on the number of locations of $\ugame$ by the bound on~$\bornePseudoPoly$ of Lemma~\ref{lem:switching}. 
	
	We now proceed to the induction on the number of applications of 
	$\textsc{Next}$. As a base case, we have $\uloc = \locinit$, and the 
	property trivially holds. Assume now that the property holds for $\uloc$, 
	with a decomposition
	$\uloc = \uloc'_0 \movetoPath{\trans_0} \uloc'_1 \movetoPath{\trans_1} 
	\cdots \movetoPath{\delta_k} \uloc'_k $. We fix some 
	$\trans = (\loc_1, \reg, \reset, w, \loc_2) \in \Trans$ with $\last(\uloc) = 
	\loc_1$, and we consider 
	$\textsc{Next}(\uloc, \trans) = (\uloc', \utrans)$ with
	$\uloc' \in \uLocsMin\cup\uLocsMax$.  We distinguish cases according to
	the definition of $\textsc{Next}$. Observe that
	$\uloc'\in \uLocsMin\cup\uLocsMax$ excludes the case when
	$\textsc{Next}$ sets $\uloc'$ in lines~\ref{line:uloc_1},
	\ref{line:rpos}, or \ref{line:revil2} (when $|\uloc'_k| = \bornePseudoPoly$). 
	The following cases may occur:
	\begin{itemize}
		\item if $\textsc{Next}$ sets $\uloc'$ in line~\ref{line:reset}, then 
		we have $\trans \in \Transr$, $|\uloc|_\trans=0$, and
		$\uloc' = \uloc \movetoPath{\trans} \loc_2$. Hence, a correct
		decomposition of $\uloc'$ is obtained by adding $\trans$ and an empty path $\uloc'_{k+1}$
		to the ones of $\uloc$.
		\item if $\textsc{Next}$ sets $\uloc'$ in line~\ref{line:plusone2},
		while $\uloc'\neq \revil$ (by hypothesis), then
		$|\uloc'_k|<\bornePseudoPoly$. A correct decomposition of $\uloc'$ is
		then obtained from the one of $\uloc$ by replacing $\uloc'_k$ with 
		$\uloc'_k \movetoPath{\trans} \loc_2$.
		\qedhere
	\end{itemize}	
\end{proof}

\section{\texorpdfstring{How \MaxPl can control negative "cycles" in the unfolding}{How Max can control negative "cycles" in the unfolding}}
\label{sec:negative-cycles-unfolding}

In this section, we show another good property of $\ugame$, mimicking the one of Lemma~\ref{lem:MaxOpt} that \MaxPl has an optimal memoryless strategy in $\rgame$ avoiding cyclic plays with a negative weight. In the unfolding 
$\ugame$, there are no cyclic plays, but we will be able to obtain a similar result: \MaxPl can play optimally with a memoryless strategy while making sure that in-between two occurrences of the same transition with a reset (which would result in a cyclic play in the original game $\rgame$), the play has a non-negative weight. More formally, we want to obtain the following lemma, that we show in the rest of the section: 

\begin{lem}
  \label{lem:U_MaxOpt}
  In the \WTG $\ugame$, $\MaxPl$ has a memoryless optimal
  strategy~$\maxstrategy$ such that if a finite play 
  $\play = \play_1 \moveto{\delay_1,\utrans_1} \play_2
  \moveto{\delay_2, \utrans_2} (\rpos,0)$ is conforming to
  $\maxstrategy$ with $\Dproj(\utrans_1) = \Dproj(\utrans_2) 
  \in \Transr$ (i.e.~the same transition with a reset in the original \WTG $\rgame$), then
  $\weightC(\play_2 \moveto{\delay_2, \utrans_2} (\rpos,0)) \geq 0$.
\end{lem}

\begin{rem}
	The fact that $\ugame$ is acyclic is crucial in this result: we can 
	not guarantee that the value of the path ending in $\rpos$ is 
	non-negative if we would have defined $\ugame$ with grey components 
	(in \figurename{~\ref{fig:unfolding}}) containing cyclic paths 
	without a reset. Indeed, the values of cyclic paths are not preserved 
	by concatenation. 
	For instance, in the \WTG $\game$ depicted on the left of \figurename~\ref{fig:cumule_value}, 
	we can see that $\ValueG(\qloc_0 \movetoPath{\trans_1} \qloc_1 
	\movetoPath{\trans_2} \qloc_0) = 0$ ($\MinPl$ and $\MaxPl$ must delay $1$ 
	in each location), and $\ValueG(\qloc_0 \movetoPath{\trans_3} \qloc_2
	\movetoPath{\trans_4} \qloc_0) = 0$. However, when we
	concatenate these two cyclic paths, we obtain the cycle $\qloc_0 \movetoPath{\trans_3} \qloc_2 \movetoPath{\trans_4}  \qloc_0
	\movetoPath{\trans_1} \qloc_1 \movetoPath{\trans_2} \qloc_0)$ of value $-1$. 
%	
%	In particular, we can compose 
%	two cyclic paths with a non-negative value and obtain a cyclic path 
%	with a negative value when \MinPl can use a cyclic path to reach a 
%	new valuation with a cost less than waiting in its location.
\end{rem}

The proof of Lemma~\ref{lem:U_MaxOpt} essentially consists in applying 
the result of Lemma~\ref{lem:MaxOpt} in $\ugame$ to define a 
memoryless optimal strategy with the desired property. However, Lemma~\ref{lem:MaxOpt} 
holds only in the closure of \WTG{s} (Example~\ref{example:Max-not-optimal} 
gives a counter-example when the \WTG is not a closure), i.e.~a priori, the 
result holds only on the closure $\overline{\ugame}$ of $\ugame$.  
Nevertheless, we notice that the guards of all transitions of $\ugame$ 
come from $\rgame$ and the regions of $\ugame$ are thus the ones of $\rgame$. 
Therefore, apart from target locations, only locations $(\uloc, I)$ in
$\overline{\ugame}$ with $\uloc$ ending in a location of the form 
$(\loc, I)$ are reachable. Thus $\ugame$ and $\overline{\ugame}$ are the 
same \WTG, and the result of Lemma~\ref{lem:MaxOpt} transfers to $\ugame$ as well. 

The second argument of the proof is checking that $\ugame$ preserves the value 
of paths, i.e.~the value in $\ugame$ of a finite path $\upath$, 
$\ValueU^{\val}(\upath)$, is equal to the value in $\rgame$ of its projection 
given by $\Dproj$. In particular, we define a new projection 
function $\Pproj$ as an extension over finite paths of $\Dproj$ such that for all 
finite paths in $\ugame$ with at least one transition, 
$\upath = \uloc_1 \movetoPath{\utrans} \upath'$ with 
$\uloc_1 \in \uLocs \setminus \uLocsT$, we let $\Pproj(\upath)$ be equal to:
\begin{displaymath}
\begin{cases}
\last(\uloc_1) \movetoPath{\Dproj(\utrans)} \loc_2 & 
\text{if $\upath' \in \uLocs$ and 
$\Dproj(\utrans) = (\last(\uloc_1), \reg, \reset, w, \loc_2)$}  \\
\last(\uloc_1) \movetoPath{\Dproj(\utrans)} \Pproj(\upath') & 
\text{otherwise .} 
\end{cases}
\end{displaymath}
We note that, $\Pproj(\upath)$ is always a finite path in $\rgame$ with the 
same length of $\upath$ and it 
satisfies the following properties: 
\begin{lem}
	\label{lem:Pproj}
	Let $\upath \in \PPaths_{\ugame}$ be a path with at 
	least one transition, then 
	\begin{enumerate}
		\item\label{item:Pproj_val} for all valuations $\val$, 
		$\ValueU^{\val}(\upath) = \ValueRG^{\val}(\Pproj(\upath))$; 
		\item\label{item:Pproj_suffix} if $\last(\upath) \notin \uLocsT$, 
		then $\Pproj(\upath)$ is a suffix of $\last(\upath)$.
	\end{enumerate} 
\end{lem}
\begin{proof}
	\begin{enumerate}
		\item We prove this property when the first location of $\upath$ belongs to \MinPl. The case where it belongs to \MaxPl is analogous 
		when we replace the infimum by a supremum. We reason by 
		induction on the length of $\upath$. First, we suppose that $\upath$ 
		contains exactly one transition, i.e.~$\upath = \uloc_1 \movetoPath{\utrans} \uloc_2$ and :
		\begin{displaymath}
		\ValueU^{\val}(\upath) = \inf_\delay \big(\delay\,\uweight(\uloc_1) + 
		\uweight(\utrans) + \ValueU^{\val'}(\uloc_2) \big) 
		\end{displaymath} 
		where $\val' = \val + \delay$ if $\utrans$ does not contain a reset, or 
		$\val' = 0$ otherwise. Since the value of an empty path is null, we have 
		$\ValueU^{\val'}(\uloc_2) = 0 = \ValueRG^{\val'}(\loc_2)$ where 
		$\loc_2$ is given by $\Dproj(\utrans)$. Moreover, as $\ugame$ preserves the weight 
		of transitions and locations, we obtain~that 
		\begin{displaymath}
		\ValueU^{\val}(\upath) 
		= \inf_\delay \big(\delay\,\weight(\last(\uloc_1)) + 
		\weight(\Dproj(\utrans)) + \ValueRG^{\val'}(\loc_2) \big) \,.
		\end{displaymath}
		By definition of $\Pproj$, we remark that  
		$\Pproj(\upath) = \last(\uloc_1) \movetoPath{\Dproj(\utrans)} \loc_2$.
		Thus, since $\ugame$ preserves transitions with a reset, we deduce~that 
		$\ValueU^{\val}(\upath) = \ValueRG^{\val}(\Pproj(\upath))$.

		Now, we suppose that $\upath = \uloc_1 \movetoPath{\utrans} \upath'$ with $\upath'$ a path in $\ugame$. Since $\ugame$ preserves the weight of transitions and locations, we have:
		\begin{align*}
		\ValueU^{\val}(\upath) 
		&= \inf_\delay \big(\delay \,\uweight(\uloc_1) + 
		\uweight(\utrans) + \ValueU^{\val'}(\upath') \big) \\
		&= \inf_\delay \big(\delay \,\weight(\last(\uloc_1)) + 
		\weight(\Dproj(\utrans)) + \ValueU^{\val'}(\upath') \big)
		\end{align*}
		where $\val' = \val + \delay$ if $\utrans$ does not contain a reset, or 
		$\val' = 0$ otherwise. Now, the induction hypothesis applied to $\upath'$
		implies that 
		\begin{displaymath}
		\ValueU^{\val}(\upath) 
		= \inf_t \big(t\,\weight(\last(\uloc_1)) + 
		\weight(\Dproj(\utrans)) + \ValueRG^{\val'}(\Pproj(\upath')) \big) 
		\end{displaymath}
		Finally, we obtain that $\ValueU^{\val}(\upath) = \ValueRG^{\val}(\Pproj(\upath))$, 
		since $\ugame$ preserves transitions with a reset and 
		$\Pproj(\upath) = \last(\uloc_1) \movetoPath{\Dproj(\utrans)} \Pproj(\upath')$.

		\item We reason by induction on the length of $\upath$. First, we suppose 
		that $\upath$ contains only one transition, i.e.~$\upath = \uloc_1 \movetoPath{\utrans} \uloc_2$ with 
		$\uloc_2 \in \uLocsMin \cup \uLocsMax$, and  
		$\Pproj(\upath) = \last(\uloc_1) \movetoPath{\Dproj(\utrans)} \loc_2$ 
		where $\loc_2$ is given by $\Dproj(\utrans)$. By definition of \textsc{Next}, since $\uloc_2 \notin \uLocsT$, we note that 
		$\uloc_2 = \uloc_1 \movetoPath{\Dproj(\utrans)} \loc_2$. 
		Since $\last(\uloc_1)$ is a suffix of $\uloc_1$, it follows that 
		$\Pproj(\upath)$ is a suffix of $\last(\upath) = \uloc_2$. 
		
		Otherwise, we suppose that $\upath = \uloc_1 \movetoPath{\utrans} \upath'$ 
		with $\last(\upath') \notin \uLocsT$, and 
		$\Pproj(\upath) = \last(\uloc_1) \movetoPath{\Dproj(\utrans)} \Pproj(\upath')$. 
		By induction hypothesis, 
		$\Pproj(\upath')$ is a suffix of $\last(\upath')$, i.e.~there exists a finite path $\rpath$ 
		of $\rgame$ such that 
		$\last(\upath') = \rpath \cdot \Pproj(\upath')$. Now, we remark that 
		$\rpath = \uloc_1$, since each application of \textsc{Next} (that does not reach 
		a target location) adds exactly one transition in the path of the next 
		location: $\utrans$ is a transition between $\uloc_1$ and $\uloc_2$ where 
		$\uloc_2 = \uloc_1 \movetoPath{\Dproj(\utrans)} \loc_2$ is the first 
		location of $\upath'$. 
		Finally, we obtain a suffix of $\last(\upath)$ since $\last(\uloc_1)$ is a 
		suffix of $\uloc_1$.
		\qedhere
	\end{enumerate}
\end{proof}

\noindent Finally, we have the tools to finish the proof of Lemma~\ref{lem:U_MaxOpt}. 
%\begin{proof}[Proof of Lemma~\ref{lem:U_MaxOpt}]
	As explained before, we apply the result of Lemma~\ref{lem:MaxOpt} in $\ugame$ 
	(since the closure of $\ugame$ describes the same \WTG as $\ugame$)
	to obtain a memoryless optimal strategy~$\maxstrategy$ for \MaxPl. 
	It remains to show that if $\play = \play_1 \moveto{\delay,\utrans_1} \play_2
	\moveto{\delay_2, \utrans_2} (\rpos,0)$ is conforming to
	$\maxstrategy$ with $\Dproj(\utrans_1) = \Dproj(\utrans_2)$ containing 
	a reset, then 
	$\weightC(\play_2 \moveto{\delay_2, \utrans_2} (\rpos,0)) \geq 0$. 
	Let $\upath$ be the path of $\ugame$ followed by 
	$\play_2 \moveto{\delay_2, \utrans_2} (\rpos,0)$. We start by
	claiming that
	\begin{equation}
	\label{eq:U_OptMax-1}
	\weightP(\play_2 \moveto{\delay_2, \utrans_2} (\rpos,0)) \geq
	\ValueU(\upath) + \weightT(\rpos, 0) 
	\end{equation}
	where this is the place where the fact that $\weightT(\rpos, 0)$ is not equal to $+\infty$ is crucial.
	 
	Equation~\eqref{eq:U_OptMax-1} and 
	Lemma~\ref{lem:Pproj}.\eqref{item:Pproj_val} (with $\val = 0$) allow us to conclude as follows. First, 
	\begin{displaymath}
	\weightC(\play_2 \moveto{\delay_2, \utrans_2} (\rpos,0)) =
	\weightP(\play_2 \moveto{\delay_2, \utrans_2} (\rpos,0)) -
	\weightT(\rpos, 0) \geq  \ValueU(\upath) \geq \ValueRG(\rpath) 
	\end{displaymath} 
	where $\rpath = \Pproj(\upath)$. Then, since 
	$\upath = \upath' \movetoPath{\utrans_2} \rpos$ 
	where $\play_2$ follows $\upath'$, we deduce that $\rpath'$ is a suffix of 
	$\upath'$ where $\rpath = \rpath' \movetoPath{\trans} \loc$ (by 
	Lemma~\ref{lem:Pproj}.\eqref{item:Pproj_suffix} applied to $\upath'$). 
	In particular, the definition of $\textsc{Next}$ on $\rpath'$ and $\trans$ 
	(as $\rpos$ is reached) guarantees that 
	$\ValueRG(\rpath) \geq 0$. 
	
	To conclude the proof, we need to show~\eqref{eq:U_OptMax-1}. We 
	reason by induction on suffixes $\play'$ of 
	$\play_2 \moveto{\delay_2, \utrans_2} (\rpos,0)$ 
	showing that 
	\begin{displaymath}
	\weightP(\play') \geq \ValueU^{\val'}(\upath') + 
	\weightT(\rpos, 0)
	\end{displaymath}
	where $\upath'$ is the path followed by~$\play'$, 
	and $\val'$ is the first valuation of $\play'$.
	For the suffix $\play' = (\rpos, 0)$, then
	\begin{displaymath}
	\weightP(\play') = \weightT(\rpos, 0) = \ValueU^0(\rpos) +
	\weightT(\rpos, 0) \,.
	\end{displaymath}
	Otherwise, we suppose that
	$\play' = (\uloc, \val') \moveto{\delay, \utrans} \play''$. In particular, 
	we fix $\val'' = \val' + \delay$ the first valuation of $\play''$ 
	($\val'' \neq 0$ since $\play_2$ does not contain a transition with a reset) and 
	$\upath' = \uloc \movetoPath{\utrans} \upath''$ with $\play''$ follows $\upath''$. 
	Moreover, we deduce that
	\begin{align*}
	\weightP(\play') 
	&= \delay \, \uweight(\uloc) + \uweight(\utrans) + \weightP(\play'') \\
	&\geq \delay \, \uweight(\uloc) + \uweight(\utrans) +
	\ValueU^{\val'+\delay}(\upath'') + \weightT(\rpos, 0)  
	\qquad \text{(by induction hypothesis)}\,.
	\end{align*} 
	To conclude the induction case, we distinguish two cases. 
	\begin{itemize}
		\item If $\uloc \in \uLocsMin$, then
		\begin{align*}
		\weightP(\play') 
		&\geq \inf_{\delay \text{ s.t. } (\uloc,\val') \moveto{\delay,\utrans} 
			(\uloc', \val'+\delay)} 
		\left(\delay \, \uweight(\uloc) + \uweight(\utrans) + 
		\ValueU^{\val'+\delay}(\upath'') + \weightT(\rpos,0)\right) \\
		&= \inf_{\delay \text{ s.t. } (\uloc,\val') \moveto{\delay,\utrans} 
			(\uloc', \val' + \delay)} \left(\delay \, \uweight(\uloc) + 
		\uweight(\utrans) + \ValueU^{\val'+\delay}(\upath'')\right) + 
		\weightT(\rpos,0) \\ 
		&= \ValueU^{\val'}(\upath') + \weightT(\rpos,0) \,.
		\end{align*}
		
		\item If $\uloc \in \uLocsMax$, since $\maxstrategy$ chooses $\utrans$,
		we can deduce that
		\begin{align*}
		\weightP(\play') 
		&\geq \sup_{\delay \text{ s.t. } (\uloc,\val') \moveto{\delay,\utrans} 
			(\uloc', \val' + \delay)} \left(\delay \, \uweight(\uloc) + 
		\uweight(\utrans) + \ValueU^{\val'+\delay}(\upath'') + 
		\weightT(\rpos,0)\right) \\
		&= \sup_{\delay \text{ s.t. } (\uloc,\val') \moveto{\delay,\utrans} 
			(\uloc',\val'+\delay)} \left(\delay \, \uweight(\uloc) + 
		\uweight(\utrans) + \ValueU^{\val'+\delay}(\upath'')\right) + 
		\weightT(\rpos,0) \\ 
		&= \ValueU^{\val'}(\upath') + \weightT(\rpos,0) \,.
		\end{align*}
	\end{itemize}
	Since it is obtained after transition $\utrans_1$ that resets the
	clock, the first valuation of 
	$\play_2 \moveto{\delay_2, \utrans_2} (\rpos,0)$ is $0$. Thus, by induction, 
	we obtain \eqref{eq:U_OptMax-1} as expected. 
%\end{proof}
 
\section{Value of the unfolding}
\label{sec:equal-values}

The most difficult part of the proof of Theorem~\ref{thm:complexity} 
is to show that the unfolding preserves the value from $\rgame$. Remember 
that we have fixed an initial location 
$\locinit = (\qlocinit, I_{\mathsf{i}})$ to build $\ugame$.

\begin{thm}
  \label{theo:value}
  For all $\val \in I_{\mathsf{i}}$,
  $\ValueRG(\locinit, \val) = \ValueU(\locinit, \val)$.
\end{thm}

\begin{figure}[t]
	\centering
	\begin{tikzpicture}[xscale=1.2,every node/.style={font=\footnotesize}]
	\node[PlayerMax, draw=none] at (0, 0) (Rplay) {$\rFPlays$};
	\node[PlayerMax, draw=none] at (4, 0) (Uplay) {$\uFPlaysEtoile$};
	\node[PlayerMax, draw=none] at (0, -2) (rStrat)
        {$\Rpos\times \CTrans$};
	\node[PlayerMax, draw=none] at (4, -2) (uStrat)
        {$\Rpos\times \uTrans$};
	
	\draw[->,rounded corners]
        (Rplay) edge[bend right=10] node[below] {$\Phi$} (Uplay)
		(Rplay) edge node[left] {$\rmaxstrategy$} (rStrat)
		(Uplay) edge[bend right=10] node[above] {$\proj$} (Rplay) 
		(Uplay) edge node[right] {$\umaxstrategy$} (uStrat)
		(uStrat) edge[bend right=10] node[above] {$\id \times \Dproj$} (rStrat)
		(rStrat) edge[bend right=10,dotted] node[below] {$\id \times \textsc{Next}$} (uStrat);
	\end{tikzpicture}
	\caption{Scheme showing the links between the different
          objects defined for the proof of Theorem~\ref{theo:value}
          where $\uFPlaysEtoile$ is the set of finite plays of
          $\ugame$ avoiding target locations $\rpos$ and $\rneg$.}
	\label{fig:schema_preuve}
\end{figure}
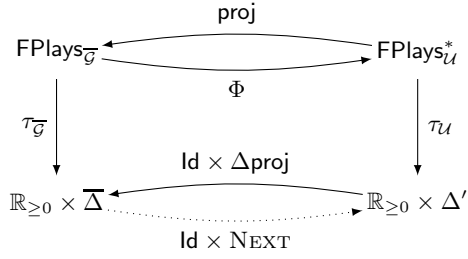

We prove Theorem~\ref{theo:value} in this section, splitting the 
proof into two inequalities.

\paragraph{First inequality}
We prove first that $\ValueRG(\locinit, \val) \leq \ValueU(\locinit,\val)$, which can be rewritten as: 
\begin{displaymath}
\ValueRG(\locinit, \val)  \leq
\sup_{\umaxstrategy} \ValueU^{\umaxstrategy}(\locinit,\val) \,. 
\end{displaymath}
We must thus show that $\MaxPl$ can
guarantee to always do at least as good in $\ugame$ as in~$\rgame$. We
thus fix an optimal strategy $\rmaxstrategy$ in $\rgame$
obtained by Lemma~\ref{lem:MaxOpt}: in
particular, $\ValueRG(\locinit, \val) =
\ValueRG^{\rmaxstrategy}(\locinit,\val)$. We show the existence of a strategy
$\umaxstrategy$ in $\ugame$ such that $\ValueRG^{\rmaxstrategy}(\locinit, \val) \leq
\ValueU^{\umaxstrategy}(\locinit, \val)$, i.e. for all plays $\play$ conforming
to $\umaxstrategy$, there exists a play conforming to~$\rmaxstrategy$
with a weight at most the weight of $\play$. As it is depicted in
\figurename~\ref{fig:schema_preuve}, the strategy~$\umaxstrategy$ is
defined via a \emph{projection} of plays of $\ugame$ in $\rgame$: we
use the mapping $\textsc{Next}$ to send back transitions of $\CTrans$
to~$\uTrans$.

%
%\begin{prop}
%	\label{prop:leq}
%	Let $\rmaxstrategy$ be a memoryless optimal strategy of 
%	$\MaxPl$ in $\rgame$. There exists a strategy~$\umaxstrategy$ 
%	in $\ugame$ such that 
%	$\ValueRG^{\rmaxstrategy}(\locinit, \val) \leq
%	\ValueU^{\umaxstrategy}(\locinit, \val)$.
%\end{prop} 
More formally, the projection operator $\proj$ 
projects finite plays of $\ugame$ starting in $\locinit$
(since these are the only ones we need to take care of) to finite
plays of $\rgame$. For this reason, from now on, $\uFPlays$ and
$\rFPlays$ denote the subsets of plays that start in location
$\locinit$. Moreover, we limit ourselves to projecting plays of
$\ugame$ that do not reach the targets $\rneg$ and $\rpos$, since
otherwise there is no canonical projection in $\rgame$. We
thus let $\uFPlaysEtoile$ be all such finite plays of $\uFPlays$ that
do not end in $\rneg$ or $\rpos$. 
The projection function $\proj\colon \uFPlaysEtoile \to \rFPlays$ is
defined inductively on finite plays $\play \in \uFPlaysEtoile$ by
letting~$\proj(\play)$~be
\begin{displaymath}
\begin{cases}
(\locinit, \val) & \text{if $\play = (\locinit, \val) \in \uLocs$;} \\
\proj(\play') \moveto{\delay, \Dproj(\utrans)} (\last(\uloc), \val) 
& \text{if $\play = \play' \moveto{\delay, \utrans} (\uloc, \val)$;} \\
\proj(\play') \moveto{\delay, \Dproj(\utrans)} (\loc', \val) 
& \text{if $\play = \play' \moveto{\delay, \utrans} (\revil, \val)$ 
	and $\Dproj(\utrans) = (\loc, \reg, \reset, w, \loc')$.}
\end{cases}
\end{displaymath}
It fulfils the following properties:
\begin{lem}
	\label{lem:proj}
	For all plays $\play \in \uFPlaysEtoile$, 
	\begin{enumerate}
		\item\label{item:proj_last} if $\last(\play)=(\uloc,\val)$ 
		with $\uloc \neq \revil$,
		then $\last(\proj(\play)) = (\last(\uloc),\val)$;
		\item\label{item:proj_weight}
		$\weightC(\play) = \weightC(\proj(\play))$;
		\item\label{item:proj_strong} if $\last(\play) = (\uloc,\val)$ with
		$\uloc\notin\LocsT$, then $\proj(\play)$ follows $\uloc$.
	\end{enumerate}
\end{lem}
\begin{proof}
	\begin{enumerate}
		\item Since $\uloc \neq \revil$, this is direct from a case 
		analysis on the definition of $\proj$.
		
		\item We reason by induction on the length of
		$\play \in \uFPlaysEtoile$. First, we suppose that 
		$\play = (\locinit, \val)$, then we have $\proj(\play) = \play$ and   
		$\weightC(\play) = 0 = \weightC(\proj(\play))$.  Now, we suppose
		that $\play = \play' \moveto{\delay, \utrans} (\uloc, \val)$,
		with $\play' \in \uFPlaysEtoile$ ending in location $\uloc'$ such 
		that $\uloc' \notin \uLocsT$. Then,
		\begin{align*}
		\weightC(\play) 
		&= \weightC(\play') + \delay \, \weight'(\uloc') + \weight'(\utrans) \\ 
		&= \weightC(\play') + 
		\delay\, \Cweight(\last(\uloc')) + \Cweight(\Dproj(\utrans))
		\end{align*}
		since $\ugame$ preserves the weights of $\rgame$, i.e.~$\weight'(\uloc') = \Cweight(\last(\uloc'))$, and
		$\weight'(\utrans) = \Cweight(\Dproj(\utrans))$. Moreover, 
		the induction hypothesis applied to $\play'$ implies that 
		\begin{align*}
		\weightC(\play) 
		&= \weightC(\proj(\play')) + 
		\delay\, \Cweight(\last(\uloc')) + \Cweight(\Dproj(\utrans)) \\
		&= \weightC(\proj(\play')) + 
		\delay\, \Cweight(\last(\proj(\play'))) + \Cweight(\Dproj(\utrans))
		\end{align*}
		since $\last(\uloc') = \last(\proj(\play'))$ by the first item 
		(as $\uloc' \neq \revil$). Finally, by the definition of $\proj(\play)$, we
		conclude that $\weightC(\play) = \weightC(\proj(\play))$.
		
		\item We reason by induction on the length of $\play \in \uFPlaysEtoile$ 
		that does not reach a target location. If $\play = (\locinit, \val)$, 
		the property is trivial. Now, we suppose that 
		$\play = \play' \moveto{\delay, \utrans} (\uloc, \val)$,
		with $\play' \in \uFPlaysEtoile$ ending in a configuration $(\uloc', \val')$ 
		such that $\uloc' \notin \uLocsT$. In particular, we have 
		$\proj(\play) = \proj(\play') \moveto{\delay, \trans} (\last(\uloc), \val)$ 
		with $\trans = \Dproj(\utrans)$, and, by the induction hypothesis,
		$\proj(\play')$ follows $\uloc'$. Moreover, we have
		$\textsc{Next}(\uloc', \trans) = (\uloc, \utrans)$ such that $\uloc$ 
		must be obtained from $\uloc'$ on lines~\ref{line:reset} or~\ref{line:plusone2} 
		of Algorithm~\ref{algo:Next}, i.e.~$\uloc = \uloc' \movetoPath{\trans} \loc_2$ 
		where $\loc_2$ is given by $\trans$. Thus, we deduce that
		$\proj(\play)$ follows $\uloc$.
		\qedhere
	\end{enumerate}
\end{proof}

Now, for all plays $\play \in \uFPlaysEtoile$ such that 
$\last(\play) = (\uloc, \val)$ and $\uloc \in \uLocsMax$ 
(for plays not starting in~$\locinit$, the decision over $\play$ is irrelevant), we define 
a strategy $\umaxstrategy$ for \MaxPl in $\ugame$ by
\begin{equation*}
%\label{eq:umaxstrategy}
\umaxstrategy(\play) = (\delay, \utrans) \quad
\text{if } \rmaxstrategy(\proj(\play)) = (\delay, \trans) \text{ and } 
\textsc{Next}(\uloc, \trans) = (\uloc', \utrans)
\end{equation*} 
We note that this is a valid decision for $\MaxPl$: we apply 
the same delay (since delays chosen in~$\rmaxstrategy$ and 
$\umaxstrategy$ are the identical) from the same configuration (as 
$\last(\proj(\play)) = (\last(\rpath), \val)$, by
Lemma~\ref{lem:proj}.\eqref{item:proj_last}), through the same guard
(since guards of $\trans$ and $\trans'$ are identical). Thus, whether or not
the location $\uloc$ is urgent (i.e.~$\last(\rpath)$ is urgent), the
decision $(\delay, \utrans)$ gives rise to an edge in $\sem{\ugame}$.
Moreover, since the definition of $\umaxstrategy$ relies on the 
projection, it is of no surprise that:

\begin{lem}
	\label{lem:proj_conf}
	Let $\play\in \uFPlaysEtoile$ be a play conforming to
	$\umaxstrategy$. Then, $\proj(\play)$ is conforming 
	to~$\rmaxstrategy$.
\end{lem}
\begin{proof}
	We reason by induction on the length of $\play$. If
	$\play = (\locinit, \val)$, then $\proj(\play) = (\locinit, \val)$,
	and the property is trivial. Otherwise, we suppose that 
	$\play = \play' \moveto{\delay, \utrans} (\uloc, \val)$ and 
	$\proj(\play) = \proj(\play') \moveto{\delay, \trans} (\last(\uloc), \val)$ 
	where $\trans = \Dproj(\utrans)$. By the induction hypothesis, 
	$\proj(\play')$ is conforming to $\rmaxstrategy$. Letting
	$\last(\proj(\play')) = (\loc', \val')$, we conclude by distinguishing
	two cases. First, if $\loc' \in \LocsMin$,
	we directly conclude that $\proj(\play)$ is conforming to
	$\rmaxstrategy$ too.  Otherwise, we suppose that $\loc' \in \LocsMax$. 
	Since $\play$ is conforming to $\umaxstrategy$ and $\play'$ also 
	belongs to $\MaxPl$ (by Lemma~\ref{lem:proj}.\eqref{item:proj_last}), we have
	$\umaxstrategy(\play') = (\delay, \utrans)$. In particular, by
	definition of $\umaxstrategy$,
	$\rmaxstrategy(\proj(\play')) = (\delay, \Dproj(\utrans)) =
	(\delay,\trans)$. Thus, $\rplay$ is conforming to~$\rmaxstrategy$.
\end{proof}

Finally, we prove $\ValueRG^{\rmaxstrategy}(\locinit, \val) \leq 
\ValueU^{\umaxstrategy}(\locinit, \val)$ by showing that for
all plays $\uplay$ from $(\locinit, \val)$ conforming to
$\umaxstrategy$, there exists a play~$\rplay$ from $(\locinit,\val)$
conforming to $\rmaxstrategy$ such that
$\weightP(\rplay) \leq \weightP(\uplay)$. 
%We can not 
%directly use the projection operator since some plays conforming to 
%$\umaxstrategy$ may end up in $\rneg$ or $\rpos$. 
%In particular, we 
%treat the ones ending in $\rpos$ by using the final weight function 
%that we have chosen for $\rpos$ (bigger than any
%acyclic play of $\rgame$). Moreover, we show that there cannot be such plays
%$\uplay$ ending in $\rneg$ since they would contradict
%Lemma~\ref{lem:MaxOpt}.\eqref{item:MaxOpt_Value}. 
%
%\begin{proof}[Proof of Proposition~\ref{prop:leq}]
	If $\uplay$ does
	not reach a target location of $\ugame$ or reaches target $\revil$,
	then $\weightP(\uplay) = +\infty$, and for all plays $\rplay$
	conforming to $\rmaxstrategy$, we have
	$\weightP(\rplay) \leq +\infty = \weightP(\uplay)$. Now, we 
	suppose that $\uplay$ reaches a target location different from $\revil$.
	\begin{itemize}
		\item If the target location reached by $\uplay$ is not in
		$\{\rpos,\rneg\}$, then $\uplay\in \uFPlaysEtoile$, and we can use
		the projector operator to let $\rplay = \proj(\uplay)$. It 
		is conforming to $\rmaxstrategy$ (by Lemma~\ref{lem:proj_conf}). 
		Moreover, by letting $\last(\uplay) = (\uloc, \val)$ (with 
		$\uloc \neq \revil$ by hypothesis), we have 
		$\uweightT(\uloc, \val) = \CweightT(\last(\uloc), \val)$
		since $\last(\rplay) = (\last(\uloc), \val)$, by  
		Lemma~\ref{lem:proj}.\eqref{item:proj_last}. 
		We conclude that $\weightP(\rplay) = \weightP(\uplay)$, since 
		$\proj$ preserves the weight (by
		Lemma~\ref{lem:proj}.\eqref{item:proj_weight}).
		
		\item If the target location reached by $\uplay$ is $\rpos$, then 
		we decompose $\uplay$ as $\uplay = \uplay^1 \moveto{\delay, \utrans} 
		(\rpos, \val)$ with $\uplay^1 \in \uFPlaysEtoile$ and 
		$(\rpath', \val') = \last(\uplay^1)$. Since the
		value in $\rgame$ is supposed to be finite (we removed configurations of value $+\infty$ or $-\infty$), $\MinPl$ can always
		guarantee to reach the target, i.e.~there exists an (attractor) memoryless
		strategy $\minstrategy_{\rgame}$ that guarantees to reach $\LocsT$.
		Now, let $\rplay = \rplay^1\rplay^2$ be such that
		$\rplay^1 = \proj(\uplay^1) \moveto{\delay, \trans} (\loc, \val)$ 
		with $\trans = \Dproj(\utrans)$ 
		and $\rplay^2$ be the play from $(\loc, \val)$ conforming to 
		$\rmaxstrategy$ and $\minstrategy_{\rgame}$. 
		 To conclude this 
		case, we prove that $\rplay$ is conforming to $\rmaxstrategy$ and 
		$\weightP(\rplay) \leq \weightP(\uplay)$.
		
		First, since $\proj(\uplay^1)$ is conforming to $\rmaxstrategy$ 
		(by Lemma~\ref{lem:proj_conf}), then $\rplay^1$ is conforming 
		to $\rmaxstrategy$ if and only if its last move is. If
		$\rpath'\in\uLocsMin$, then $\proj(\uplay^1)$ belongs to \MinPl 
		(by Lemma~\ref{lem:proj}.\eqref{item:proj_last}) and $\rplay^1$ 
		is conforming to $\rmaxstrategy$. Otherwise, we suppose that 
		$\rpath'\in\uLocsMax$, then $\umaxstrategy(\uplay^1) = 
		(\delay,\utrans)$ and $\textsc{Next}(\rpath',\trans) = 
		(\rneg, \utrans)$. Thus, since $\proj(\uplay^1)$ belongs to \MaxPl 
		(by Lemma~\ref{lem:proj}.\eqref{item:proj_last}) and by the 
		construction of $\umaxstrategy$, we deduce that 
		$\rmaxstrategy(\proj(\uplay^1)) = (\delay, \trans)$, i.e.~$\rplay^1$ is conforming to $\rmaxstrategy$. Finally, we conclude 
		that $\rplay$ is conforming to $\rmaxstrategy$ by the choice of 
		$\rplay^2$. 
		
		Now, we prove that $\weightP(\rplay) \leq \weightP(\uplay)$. First, we 
		remark that 
		\begin{displaymath}
		\weightP(\rplay)
		= \weightC(\rplay^1) + \weightP(\rplay^2) 
		= \weightC(\proj(\uplay^1)) + \delay \,
		\Cweight(\last(\rpath')) + \Cweight(\trans) + \weightP(\rplay^2) \,. 
		\end{displaymath} 
		In particular, since $\Cweight(\last(\rpath')) = \uweight(\rpath')$ 
		(by definition of $\ugame$) and also by using 
		Lemma\ \ref{lem:proj}.\eqref{item:proj_weight}, we obtain:
		\begin{displaymath}
		\weightP(\rplay)
		= \weightC(\uplay^1) + \delay \, \uweight(\rpath') + \uweight(\utrans) + 
		\weightP(\rplay^2)
		= \weightC(\uplay) + \weightP(\rplay^2) \,.
		\end{displaymath}
		Moreover, the length of $\rplay^2$ is bounded by $|\Locs|$ (since it is 
		conforming to an attractor, and since regions are already encoded in $\rgame$) and each of its edges has a weight bounded in
		absolute values by $\maxWeightTrans+ \clockbound \, \maxWeightLoc$. 
		By adding its final weight, we obtain:
		\begin{displaymath}
		\weightP(\rplay) \leq \weightC(\uplay) + |\Locs|(\maxWeightTrans+ \clockbound \,
		\maxWeightLoc) + \maxWeightFinal \,.
		\end{displaymath}
		Now, we remark that $\uplay$ reaches $\rpos$, and its
		weight is thus:
		\begin{displaymath}
		\weightP(\uplay) = \weightC(\uplay) + 
		|\Locs|(\maxWeightTrans+ \clockbound \,
		\maxWeightLoc) + \maxWeightFinal \,.
		\end{displaymath}
		Therefore, $\weightP(\rplay) \leq \weightP(\uplay)$.
		
		\item Finally, we prove that the case where the target location 
		reached by $\uplay$ is $\rneg$ is not possible. As before we 
		decompose $\uplay$ as $\uplay = \uplay^1 \moveto{\delay, \utrans} 
		(\rpos, \val)$ with $\uplay^1 \in \uFPlaysEtoile$ and 
		$(\rpath', \val') = \last(\uplay^1)$. We consider 
		$\rplay^1 = \proj(\uplay^1) \moveto{\delay, \trans} (\loc, \val)$ 
		with $\trans = \Dproj(\utrans)$ that is conforming to $\rmaxstrategy$ 
		(by the same reasoning than the previous case) and we prove that 
		$\rplay^1$ finishes with a play that follows the cyclic path with 
		negative value that contradicts 
		Lemma~\ref{lem:MaxOpt}.\eqref{item:MaxOpt_Value}. By definition of $\ugame$,  
		we have $\textsc{Next}(\rpath', \trans) = (\rneg, \utrans)$ with
		$|\rpath'|_{\trans}> 0$, so by letting $\rpath' = \rpath_1 
		\movetoPath{\trans} \rpath_2$ with $|\rpath_2|_\trans = 0$, 
		we have $\ValueRG(\rpath_2 \moveto{\trans} \loc_2) <0$ where 
		$\loc_2$ is given by $\trans$. Moreover, since $\proj(\uplay^1)$ follows 
		$\rpath$ (by Lemma~\ref{lem:proj}.\eqref{item:proj_strong}), $\rplay^1$ 
		follows $\rpath \movetoPath{\trans} \loc_2$ that contains a cyclic path 
		$\rpath_2 \movetoPath{\trans} \loc_2$ with a negative~value. 
	\end{itemize}
	To conclude the proof, we have shown that for all plays $\uplay$ from 
	$(\locinit, \val)$ conforming to $\umaxstrategy$, we can build a play 
	$\rplay$ from $(\locinit, \val)$ conforming to $\rmaxstrategy$ such that 
	$\weightP(\rplay) \leq \weightP(\uplay)$. In particular, 
	\begin{align*}
	\ValueRG^{\rmaxstrategy}(\locinit, \val) 
	&= \inf_{\rmaxstrategy \in \StratMin[\rgame]} 
	\weightP(\Play((\locinit, \val), \rminstrategy, \rmaxstrategy))\\
	&\leq \inf_{\umaxstrategy \in \StratMin[\ugame]}  
	\weightP(\Play((\locinit, \val), \uminstrategy, \umaxstrategy))\\ 
	&\leq\ValueU^{\umaxstrategy}(\locinit, \val) \,.
	\end{align*}
%\end{proof}

%%%%%%%%%%%%%%%

\paragraph{Second inequality}
We then prove the reciprocal inequality
$\ValueRG(\locinit, \val) \geq \ValueU(\locinit,\val)$ that can be rewritten as:
\begin{displaymath}
\ValueRG(\locinit, \val)  \geq
\sup_{\umaxstrategy} \ValueU^{\umaxstrategy}(\locinit,\val) \,.
\end{displaymath}
It thus amounts to showing that $\MaxPl$ 
can guarantee to always do at least as good in $\rgame$ as in $\ugame$. 
We thus fix the optimal strategy $\umaxstrategy$ in~$\ugame$ given by Lemma~\ref{lem:U_MaxOpt}, and show that
$\ValueRG(\locinit, \val) \geq \ValueU^{\umaxstrategy}(\locinit, \val)$.

To do so, we show that there exists a strategy 
$\rmaxstrategy$ in~$\rgame$ such that for a particular play~$\play$ 
conforming to $\rmaxstrategy$, there exists a play conforming to
$\umaxstrategy$ with a weight at most the weight of $\play$. 
As depicted in \figurename~\ref{fig:schema_preuve}, the
strategy $\rmaxstrategy$ is defined via a function $\Phi$ that maps
plays of $\rgame$ into plays of $\ugame$. 
Intuitively, this function removes all
cyclic plays ending with a reset from plays in $\rgame$.
%
%Considering for $\umaxstrategy$ the memoryless optimal strategy of
%$\MaxPl$ in $\ugame$ satisfying the conditions of
%Lemma~\ref{lem:U_MaxOpt}. Therefore, we show that 
%\begin{prop}
%	\label{prop:geq}
%	$\ValueRG(\locinit, \val) \geq
%	\ValueU^{\umaxstrategy}(\locinit, \val)$.
%\end{prop}
%We first define
%the function~$\Phi$, mapping plays of $\rgame$ in plays of
%$\ugame$. 
%It needs to take care of the appearance of more than one
%occurrence of a transition with a reset in plays of $\rgame$. 
Formally, 
it is defined by induction on the length of the plays by letting 
$\Phi(\locinit, \val) = (\locinit, \val)$, and for all plays 
$\play \in \rFPlays$, letting $\play' = \play 
\moveto{\delay, \trans} (\loc, \val)$,
\begin{enumerate}
	\item if $\Phi(\play)$ ends in $\revil$, we fix 
	$\Phi(\play') = \Phi(\play)$;
	
	\item else, if $\trans$ contains a reset and
	$\Phi(\play) = \play_1 \moveto{\delay', \trans'} \play_2$ with
	$\Dproj(\trans') = \trans$, letting $\uloc$ the first location 
	of $\play_2$, we fix $\Phi(\play') = \play_1 
	\moveto{\delay', \trans'} (\uloc,0)$;
	
	\item otherwise, letting $\textsc{Next}(\uloc,\trans) = (\uloc',\trans')$ 
	with $\uloc$ the last location of $\Phi(\play)$, we fix 
	$\Phi(\play') = \Phi(\play) \moveto{\delay, \trans'} (\uloc', \val)$.
\end{enumerate}
This function satisfies the following properties:
\begin{lem}
	\label{lem:Phi}
	For all plays $\play\in \rFPlays$, if $\last(\Phi(\play)) = (\uloc, \val)$ 
	with $\uloc\neq \revil$, then we have $\uloc\notin\{\rneg,\rpos\}$ and
	\begin{displaymath}
	\last(\play) =
	\begin{cases}
	(\last(\uloc), \val) & \text{if } \uloc\notin \LocsT \,;\\
	(\uloc, \val) & \text{otherwise .}
	\end{cases}
	\end{displaymath}
\end{lem}
\begin{proof}
	We show the property by induction on the length of $\play$. 
	If $\play = (\locinit, \val)$, then $\Phi(\play) = \play$ and 
	the property holds. Otherwise, we let
	$\play' = \play \moveto{\delay, \trans} (\loc, \val)$, and we 
	suppose that the property holds for $\play$ (since it does not end 
	in $\uLocsT$) and we follow the definition of $\Phi$.
	\begin{enumerate}
		\item If $\Phi(\play)$ ends in $\revil$, we have
		$\Phi(\play') = \Phi(\play)$ and this case is thus not possible
		(since $\Phi(\play')$ is supposed to not end in $\revil$).
		
		\item Else, if $\trans$ contains a reset and
		$\smash[t]{\Phi(\play) = \play_1 \moveto{\delay', \utrans} \play_2}$ 
		with $\Dproj(\utrans) = \trans = (\loc, \reg, \reset, w, \loc')$ 
		and $\last(\play_1) = (\uloc_1, \val)$, 
		we have $\Phi(\play') = \play_1 \moveto{\delay', \utrans} (\uloc',0)$, 
		by letting $\uloc'$ the first location of $\play_2$. Moreover, we have 
		$\textsc{Next}(\uloc_1, \trans) = (\uloc', \utrans)$. Now, by 
		definition of \textsc{Next}, if
		$\loc' \in \LocsT$, then $\uloc' = \loc' \in \LocsT$. Thus, we 
		conclude that $\last(\play') = (\loc', 0) = (\last(\Phi(\play')), 0)$ 
		as expected. Otherwise, $\loc' \notin \LocsT$ and we have 
		$\last(\Phi(\play)) = (\uloc, 0)$. We note that 
		$\uloc \notin \{\rneg, \rpos\}$ since $\play_1$ 
		does not contain a transition $\utrans_1$ such that 
		$\Dproj(\utrans_1) = \trans$ (otherwise, in $\Phi(\rho)$, we 
		would have already fired twice the transition $\trans$ with a 
		reset, before trying to fire it a third time). Thus 
		$\uloc= \uloc' \movetoPath{\trans} \loc'$, and we conclude.
		
		\item Otherwise, $\Phi(\play') = \Phi(\play) \moveto{\delay, \trans'} 
		(\uloc', \val)$ if $\textsc{Next}(\uloc,\trans) = (\uloc', \trans')$ 
		with $\uloc$ the last location of $\Phi(\play)$. Once again, we are
		in a case where $\uloc' = \uloc \movetoPath{\trans} \loc'$, by 
		letting $\trans = (\loc, \reg, \reset, w, \loc')$. Thus, we
		conclude as~before. \qedhere
	\end{enumerate}
\end{proof}

Now, we define $\rmaxstrategy$ such that its behaviour is the same as
the one given by $\umaxstrategy$ after the application of $\Phi$ on
the finite play, i.e.~after the removal of all cyclic paths between the same 
transition with a reset. Formally, for all plays
$\play \in \rFPlays$, we let $\rmaxstrategy(\play)$ be defined as any
valid move $(\delay,\trans)$ if $\Phi(\play)$ ends in $\revil$, and
otherwise,
\begin{equation*}
%\label{eq:rmaxstrategy}
\rmaxstrategy(\play) =
(\delay, \Dproj(\utrans)) \qquad \text{if }
\umaxstrategy(\Phi(\play)) = (\delay, \utrans)
\end{equation*}
This is a valid decision for \MaxPl. First, by Lemma~\ref{lem:Phi},
$\last(\play) = (\last(\rpath), \val)$ when
$\last(\Phi(\play)) = (\rpath, \val)$. Moreover, delays chosen in
$\rmaxstrategy$ and $\umaxstrategy$ are the same, and the guards of
$\utrans$ and $\Dproj(\utrans)$ are identical. Thus, whether or not
the location~$\rpath$ is urgent, the decision $(\delay,
\Dproj(\utrans))$ gives rise to an edge in~$\sem{\rgame}$.
Since the definition of $\rmaxstrategy$ relies on the operation
$\Phi$, it is again not surprising that:
\begin{lem}
	\label{lem:Phi_conf}
	Let $\play \in \rFPlays$ be a play conforming to
	$\rmaxstrategy$. Then $\Phi(\play)$ is conforming to
	$\umaxstrategy$.
\end{lem}
\begin{proof}
	We reason by induction on the length of $\play$. If
	$\play = (\locinit, \val)$, then $\Phi(\play) = (\locinit,\val)$ 
	and the property is trivial. Otherwise, we suppose that
	$\play' = \play \moveto{\delay, \trans} (\loc, \val)$. By
	the induction hypothesis, $\Phi(\play)$ conforms to 
	$\umaxstrategy$. 
	\begin{enumerate}
		\item If $\Phi(\play)$ ends in $\revil$, we have
		$\Phi(\play')=\Phi(\play)$ that is conforming to $\umaxstrategy$.
		
		\item If $\trans$ contains a reset and $\Phi(\play) = \play_1 
		\moveto{\delay', \trans'} \play_2$ with $\Dproj(\trans') = \trans$, 
		letting $\uloc$ be the first location of $\play_2$, we have
		$\Phi(\play') = \play_1 \moveto{\delay', \trans'} (\uloc,0)$. This is
		a prefix of $\Phi(\play)$ that is conforming to $\umaxstrategy$. Thus, 
		$\Phi(\play')$ is conforming to $\umaxstrategy$ too.
		
		\item Otherwise, $\Phi(\play') = \Phi(\play) 
		\moveto{\delay, \trans'} (\uloc', \val)$ if 
		$\textsc{Next}(\uloc,\trans) = (\uloc',\trans')$ with $\uloc$ the
		last location of $\Phi(\play)$. If $\Phi(\play)$ ends in a location
		of $\MinPl$, since it is conforming to $\umaxstrategy$, so does
		$\Phi(\play')$. Otherwise, $\rmaxstrategy(\play) = (\delay, \trans)$
		which implies that $\umaxstrategy(\Phi(\play)) = (\delay, \utrans')$
		with $\Dproj(\utrans') = \trans$, meaning that
		$\textsc{Next}(\uloc, \trans) = (\uloc', \utrans')$,
		i.e.~$\utrans' = \utrans$: in this case too, $\Phi(\play')$ 
		is conforming to~$\umaxstrategy$.
		\qedhere
	\end{enumerate}
\end{proof}

Finally, we prove that
$\ValueRG(\locinit, \val) \geq \ValueU^{\umaxstrategy}(\locinit, \val)$.
Notice that we do not aim at comparing
$\ValueU^{\umaxstrategy}(\locinit, \val)$ with
$\ValueRG^{\rmaxstrategy}(\locinit, \val)$ but instead directly with
$\ValueRG(\locinit, \val)$. This is helpful here since we do not need
to start with any play $\play$ conforming to $\rmaxstrategy$. Instead,
we pick a special play, choosing well the strategy followed by
$\MinPl$. Indeed, we suppose that $\MinPl$ follows an 
$\varepsilon$-optimal (switching) strategy $\minstrategy$ in $\rgame$, 
as given in 
\cite{brihaye2021oneclock}. 
As we explained before in Definition~\ref{def:depliage}, in \WTG{s} 
without resets, this ensures that in all plays $\rplay$ conforming to
$\minstrategy$, the target is reached fast enough (with a number of
transitions bounded by $\bornePseudoPoly$). We can easily enrich the
result of
\cite{brihaye2021oneclock}
to take into account resets. Indeed, as performed in
\cite[Theorem~6.6]{brihaye2021oneclock} to show
that all one-clock \WTG{s} have an (a priori non-computable) value 
function that is piecewise affine with a finite number of cutpoints, 
we can replace each transition with a reset by a new transition 
jumping in a fresh target location of value given by the value function 
we aim at computing. From a strategy perspective, this means that in each
component of our unfolding (in-between two transitions with a reset),
$\MinPl$ follows a switching strategy. Notice that such strategies are
a priori not knowing to be computable (since we cannot perform the 
transformation described above, using the value function), but we 
use only its existence in this proof.

We thus consider an $\varepsilon$-optimal strategy $\minstrategy$ for
$\MinPl$ in $\rgame$ such that in all plays $\rplay$ conforming
to $\minstrategy$, in-between two transitions with a reset and after
the last such transition, the number of transitions is bounded by
$\bornePseudoPoly$.
	We now fix the special play $\play$ from $(\locinit, \val)$
	conforming to $\minstrategy$ and $\rmaxstrategy$. It reaches a 
	target since $\minstrategy$ is $\varepsilon$-optimal and
	$\ValueRG(\locinit, \val) \neq +\infty$. We show that
	\begin{equation}
	\label{eq:star}
	\exists \uplay \in \uFPlays \text{ conforming to } \umaxstrategy 
	\quad
	\weightP(\uplay) \leq \weightP(\play) \tag{$\star$}
	\end{equation}
	As a consequence, we obtain:
	\begin{displaymath}
	\ValueU^{\umaxstrategy}(\locinit, \val)
	= \inf_{\uminstrategy \in \StratMin[\ugame]} 
	\weightP(\Play((\locinit, \val),
	\uminstrategy, \umaxstrategy)) 
	\leq \weightP(\uplay) 
	\leq \weightP(\play) 
	\leq \ValueRG(\locinit, \val) + \varepsilon \,.
	\end{displaymath}
	Since this holds for all $\varepsilon>0$, we have
	$\ValueU^{\umaxstrategy}(\locinit, \val) \leq 
	\ValueRG(\locinit, \val)$ as expected.
	
	To show \eqref{eq:star}, we proceed by induction on the
	prefixes~$\play'$ of $\play$, proving that \eqref{eq:star} 
	holds or that $\Phi(\play')$ does not end in~$\revil$ and
	$\weightC(\Phi(\play')) \leq \weightC(\play')$. Indeed, at 
	the end of the induction, we therefore obtain \eqref{eq:star} or that 
	$\Phi(\play)$ does not end in~$\revil$ and
	$\weightC(\Phi(\play)) \leq \weightC(\play)$. In the case where 
	\eqref{eq:star} does not hold, we fix $\uplay = \Phi(\play)$ and 
	$\last(\uplay) = (\uloc, \val)$. In particular, we have 
	$\uloc \in \LocsT$, and $\last(\play) = (\uloc, \val)$: by 
	Lemma~\ref{lem:Phi}, if $\uloc \notin \LocsT$, then 
	$\last(\play) = (\last(\uloc), \val)$, with 
	$\last(\uloc) \notin \LocsT$ that contradicts the fact that 
	$\play$ reaches the target. Therefore,
	\begin{displaymath}
	\weightP(\uplay) = \weightP(\Phi(\play)) 
	= \weightC(\Phi(\play)) + \uweightT(\uloc,\val)
	\leq \weightC(\play) + \weightT(\uloc,\val) 
	= \weightP(\play) \,.
	\end{displaymath} 
	Since~$\uplay$ is conforming to $\umaxstrategy$ (by 
	Lemma~\ref{lem:Phi_conf}), we obtain \eqref{eq:star} here too. 
	
	Finally, we proceed to the proof by induction. First, we suppose that 
	$\play' = (\locinit, \val)$ and $\weightC(\Phi(\play')) = 0 = 
	\weightC(\play')$. Otherwise, we suppose that
	$\play' = \play'' \moveto{\delay, \trans} (\loc, \val)$. By
	induction on $\play''$, if \eqref{eq:star} does not (already) hold,
	we know that $\Phi(\play'')$ does not end in $\revil$ and
	$\weightC(\Phi(\play'')) \leq \weightC(\play'')$. We follow the three
	cases of the definition of $\Phi(\play')$.
	\begin{enumerate}
		\item We cannot have $\Phi(\play'')$ ending in $\revil$ by
		hypothesis.
		
		\item Suppose now that $\trans$ contains a reset and
		$\Phi(\play'') = \play_1 \moveto{\delay', \utrans} \play_2$ with
		$\Dproj(\utrans) = \trans$. Letting $\uloc$ the first location of
		$\play_2$, we have $\Phi(\play') = \play_1 \moveto{\delay', \utrans} 
		(\uloc,0)$. Thus
		\begin{equation}
		\label{eq:2}
		\weightC(\Phi(\play')) = \weightC(\Phi(\play'')) -
		\weightC(\play_2) \leq \weightC(\play'') -
		\weightC(\play_2) 
		\end{equation}
		Let $(\uloc', \val') = \last(\play_2)$, and
		$\uplay = \Phi(\play'') \moveto{\delay, \utrans'} (\uloc'',0)$, 
		with $\textsc{Next}(\uloc',\trans) = (\uloc'',\utrans')$. Notice 
		that $\uplay$ is conforming to $\umaxstrategy$, since 
		$\Phi(\play'')$ does and if $\uloc'$ belongs to $\MaxPl$, this 
		follows directly from the definition of $\rmaxstrategy$ from 
		$\umaxstrategy$ (since $\rmaxstrategy(\rplay'') = (\delay, \trans)$, and
		$\Phi(\play'') \notin \revil$). Moreover, it 
		contains twice a transition with a reset coming from the same 
		transition $\trans$ of $\rgame$, therefore
		$\uloc'' \in \{\rneg, \rpos\}$. If $\uloc'' = \rneg$,
		$\weightP(\uplay) = -\infty$ and \eqref{eq:star} holds. Otherwise, if
		$\uloc'' = \rpos$, by Lemma~\ref{lem:U_MaxOpt} applied on $\uplay$,
		$\weightC(\play_2 \moveto{\delay, \utrans'} (\rpos,0)) \geq 0$, i.e.~$\weightC((\uloc',\val') \moveto{\delay, \utrans'} (\rpos,0)) \geq 
		-\weightC(\play_2)$. Combined with \eqref{eq:2}, we obtain~that
		\begin{align*}
		\weightC(\Phi(\play')) 
		&\leq \weightC(\play'') + \weightC((\uloc',\val') 
			\moveto{\delay, \utrans'} (\rpos,0)) \\
		&= \weightC(\play'') + \delay\,\uweight(\uloc') + \uweight(\utrans') \\
		&= \weightC(\play'') + \delay\,\Cweight(\loc') + \Cweight(\trans) 
		= \weightC(\play')
		\end{align*}
		where we let $\loc'$ be the last location of $\play''$, which is
		also the last location of $\uloc'$. 
				
		\item Otherwise, $\Phi(\play') = \Phi(\play'') 
		\moveto{\delay, \utrans} (\uloc', \val)$ if 
		$\textsc{Next}(\uloc, \trans) = (\uloc', \utrans)$ with
		$\uloc$ the last location of $\Phi(\play'')$. In this case,
		\begin{align*}
		\weightC(\Phi(\play')) 
		&= \weightC(\Phi(\play'')) + \delay \, \uweight(\uloc) + \uweight(\utrans) \\
		& \leq \weightC(\play'') + \delay \,\Cweight(\loc') + \Cweight(\trans) 
		= \weightC(\play')
		\end{align*} 
		where we let $\loc'$ e the last location of $\play''$.
	\end{enumerate}
	This ends the proof by induction. 

\section{Main decidability result}
\label{sec:end-proof}

By using the unfolding, we are now able to conclude the proof of Theorem~\ref{thm:complexity}, i.e.~to compute the value function of $\game$ in exponential time with respect to $|\QLocs|$ and $\maxWeight$. 

  Remember (by Lemma~\ref{lem:region_ptg}) that we only need to explain 
  how to compute
  $\val \mapsto \Value_{\rgame}((\qlocinit, I_{\mathsf{i}}),\val)$
  over $I_{\mathsf{i}}$. By Theorem~\ref{theo:value}, this is
  equivalent to computing
  $\val \mapsto \Value_{\ugame}((\qlocinit, I_{\mathsf{i}}),\val)$
  over $I_{\mathsf{i}}$. We now explain why this is doable.
	
  First, the definition of $\ugame$ is effective: we can compute it
  entirely, making use of Lemma~\ref{lem:sizeU} showing that it is a
  finite \WTG. The only non-trivial part is the test of the sign of
  $\ValueRG(\ppath_2 \movetoPath{\trans} \loc_2)$ in line~\ref{line:reset} 
  of Algorithm~\ref{algo:Next} to determine in which target location we
  jump. Since $\ppath_2 \movetoPath{\trans} \loc_2$ is a finite path,
  we can apply Theorem~\ref{thm:pseudopoly} to compute the value of
  the corresponding game, which is exactly the value
  $\ValueRG(\ppath_2 \movetoPath{\trans} \loc_2)$. 
  The complexity of computing the  value of a path is polynomial in 
  the length of this path (that is exponential in $|\QLocs|$ 
  and $\maxWeight$, by Lemma~\ref{lem:sizeU}) and polynomial 
  in $|\QLocs|$ and $\maxWeight$ (notice that weights of $\rgame$ 
  are the same as the ones in $\game$): this is thus of complexity 
  exponential in $|\QLocs|$ and $\maxWeight$. Since $\ugame$ has an
  exponential number of locations with respect to $|\QLocs|$ and 
  $\maxWeight$, the total time required to compute~$\ugame$
  is also exponential with respect to $|\QLocs|$ and $\maxWeight$. 

  Lemma~\ref{lem:sizeU} ensures that $\ugame$ is acyclic, so
  we can apply Theorem~\ref{thm:pseudopoly} to compute the value
  mapping
  $\val \mapsto \Value_{\ugame}((\qlocinit, I_{\mathsf{i}}),\val)$ as
  a piecewise affine and continuous function. It requires a complexity
  polynomial in the number of locations of $\ugame$, and in $\maxWeight$
  (since weights of $\ugame$ all 
  come from $\game$). Knowing the previous bound on the number of locations 
  of~$\ugame$, this complexity translates into an exponential time complexity 
  with respect to $|\QLocs|$ and $\maxWeight$, as announced.

\section{\texorpdfstring{The value function of (one-clock) WTGs is the greatest fixpoint of $\F$}{The value function of (one-clock) WTGs is the greatest fixpoint of F}}
\label{sec:Value-fixpoint}

We finally prove Theorem~\ref{thm:Value-fixpoint}, i.e.~that the value function of all  
(one-clock) \WTG{s} is the greatest fixpoint of the operator~$\F$. 
A natural way to prove this theorem would be to use the fixpoint theory and, 
more precisely, Kleene's theorem characterising the greatest 
fixpoint\footnote{A careful reader will remark that Kleene's theorem 
characterises the least fixpoint for increasing sequences of elements in a complete partial order (CPO). 
Intuitively, we use this version with a reverse order (the CPO admits 
upper-bounds instead of lower-bounds). 
Formally, we can fit the hypothesis of Kleene's theorem by considering the 
operator $-\F$.} 
as the limit of a sequence of iterates of $\F$ before showing that the limit is equal to the value of 
$\game$.
To be applicable, this theorem requires $\F$ to be \emph{Scott-continuous} over a complete partial order (CPO)~\cite[Chapter 8]{Winskel_93}, 
i.e. monotonous and such that for all non-increasing sequence $(X_i)_{i\in \N}$ of elements of the CPO, $\inf_i \F(X^i) = \F(\inf_i X^i)$. 
Unfortunately, although the operator $\F$ can be shown to be monotonous (see Lemma~\ref{lem:Fproperties}.\ref{item:Fproperties_monotonic}), 
it is not necessarily Scott-continuous as demonstrated by the following example. 

\begin{exa}
	We consider the \WTG depicted on the left of 
	\figurename{~\ref{fig:fixpoint_sequence}}. 
	We prove that $\F$ is not Scott-continuous by exhibiting a 
	non-increasing sequence $(X_i)_i$ of functions (that are continuous over regions) 
	such that $\inf_i X_i > \F(\inf_i X_i)$. The sequence is depicted in the right of \figurename{~\ref{fig:fixpoint_sequence}} and can be 
	defined, for all configurations $(\qloc, \val)$ and all $i \in \N$, by:
	\begin{displaymath}
	X_i(\qloc, \val) = 
	\begin{cases}
	\frac{1}{2^i} & \text{if } 0 < \val \leq \frac{2^i-1}{2^i}\\
	(2^i - 1)\val + (2 - 2^i) &  \text{if } \frac{2^i-1}{2^i} < \val < 1
	\end{cases}
	\end{displaymath}
	We consider the location $\qloc_0$ of \MaxPl. For all $i \in \N$ and 
	all valuations $\val$, we have: 
	\begin{displaymath}
	\F(X_i)(\qloc_0, \val) = 
	\sup_{\delay} (\delay\, \weight(\qloc_0) + \weight(\trans) + X_i(\qloc_1, \val + \delay)) = 
	\sup_{\delay} X_i(\qloc_1, \val + \delay) = 1
	\end{displaymath} 
	and thus $\inf_i \F(X_i)(\qloc_0, \val) = 1$.
	Moreover, since $\inf_i X_i$ is  constant function whose value is $0$ in the interval $(0 ,1)$, 
	we deduce that $\F(\inf_i X_i)(\qloc_0, \val) = 0$ for all configurations $(\qloc_0, \val)$. 
	Thus, we deduce that $\F(\inf_i X_i)(\qloc_0, \val) < \inf_i \F(X_i)(\qloc_0, \val)$.	
	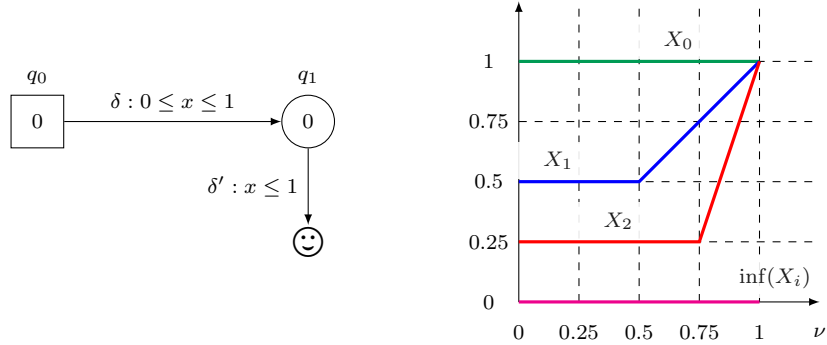
\begin{figure}[tbp]
		\begin{tikzpicture}[node distance=3cm,auto,>=latex,
		every node/.style={font=\scriptsize},scale=.8]
		
		\node[PlayerMax, label={above:$\qloc_0$}] at (0, 0) (s1) {$0$};
		\node[PlayerMin, label={above:$\qloc_1$}] at (4.5, 0)  (s0) {$0$};
		\node[target] at (4.5, -2) (s3) {$\LARGE \smiley$};
		
		% Connect the states with arrows
		\draw[->]
		(s1) edge node[above] {$\trans : 0 \leq x \leq 1$} (s0)
		(s0) edge node[left] {$\trans' : x \leq 1$} (s3)
		;
		
		\begin{scope}[xshift=8cm, yshift=-3cm]
		\draw[->] (0,0) -- (0,5);
		\draw[->] (0,0) -- (5,0);
		
		\draw[dashed] (1,0) -- (1,5);
		\draw[dashed] (2,0) -- (2,5);
		\draw[dashed] (3,0) -- (3,5);
		\draw[dashed] (4,0) -- (4,5);
		
		\draw[dashed] (0,1) -- (5,1);
		\draw[dashed] (0,2) -- (5,2);
		\draw[dashed] (0,3) -- (5,3);
		\draw[dashed] (0,4) -- (5,4);
		
		\node at (5,-.5) {$\val$};
		
		\node at (0,-.5) {$0$};
		\node at (1,-.5) {$0.25$};
		\node at (2,-.5) {$0.5$};
		\node at (3,-.5) {$0.75$};
		\node at (4,-.5) {$1$};
		
		\node at (-.5,0) {$0$};
		\node at (-.5,1) {$0.25$};
		\node at (-.5,2) {$0.5$};
		\node at (-.5,3) {$0.75$};
		\node at (-.5,4) {$1$};
		
		\draw[ForestGreen,very thick] (0,4) -- (4,4);
		\draw[blue,very thick] (0,2) -- (2,2) -- (4,4);
		\draw[red,very thick] (0,1) -- (3,1) -- (4,4);
		\draw[magenta,very thick] (0,0) -- (4,0);
		
		\node[rectangle,fill=white,fill opacity=.9,text opacity=1] at 
		(2.5,4.35) {$X_0$};
		\node[rectangle,fill=white,fill opacity=.9,text opacity=1]  at 
		(.5,2.35) {$X_1$};
		\node[rectangle,fill=white,fill opacity=.9,text opacity=1]  at 
		(1.5,1.35) {$X_2$};
		\node[rectangle,fill=white,fill opacity=.9,text opacity=1]  at 
		(4.1,0.4) {$\inf(X_i)$};
		\end{scope}
		\end{tikzpicture}
		\caption{On the left, a \WTG in which $\F$ is not Scott-continuous, for instance when 
			we consider the non-increasing sequence of continuous functions on 
			$(0, 1)$ depicted on the right for all locations.}
		\label{fig:fixpoint_sequence}
	\end{figure}

\end{exa}

We thus design a more pedestrian 
proof only using non-increasing sequences $(V_i)$ 
defined by an iteration of the operator $\F$ (as in~\cite{Tarski-55})
that \emph{uniformly} converge over each region, i.e.~the restriction of the sequence 
to each region uniformly converges. 
In particular, we adapt and correct the sketch given in~\cite{Bou16} 
for concurrent hybrid games with only non-negative weights to the context 
of (one-clock) \WTG{s} with negative weights.
As in~\cite{Bou16}, our proof is split into two parts: 
\begin{enumerate}
	\item In Section~\ref{subsec:fixpoint_Kleene}, we prove that  
	the sequence $V_i$ of iterates of $\F$ (used in the \emph{value iteration}-based algorithm 
	of~\cite{BouyerCassezFleuryLarsen-04}) converges toward the greatest fixpoint of $\F$. 
	In \cite{Bou16}, it is proved that all non-increasing sequences of functions 
	that uniformly converge over each region are a fixpoint of the operator $\F$. 
	The key argument of \cite{Bou16} is to prove the uniform convergence of the sequence $V_i$ by using Dini's theorem. 
	We show that this is legal by showing that functions $V_i$ are all $k$-Lipschitz-continuous for the same constant $k$ (which requires us to restrict to one-clock WTGs).
	
	\item In Section~\ref{subsec:fixpoint_Val}, we prove that the 
	sequence $V_i$ of iterates of $\F$ converges to the value function. 
	The key argument in~\cite{Bou16} is to remark that the mapping obtained 
	after $i$ applications of $\F$ is a value function when we consider only plays of length $i$.
	We formalise this intuition by inductively defining a strategy of \MinPl that will 
	increase the length $i$ of the plays such that its value is upper bounded by $V_i$.
	To do it, in Section~\ref{subsec:fixpoint_restriction}, we start by proving that a fixpoint of a restriction 
	of $\F$ under a given strategy of \MinPl is the value of this strategy.
\end{enumerate}

We now fix a (one-clock) \WTG $\game$. We have supposed that the final weight functions are continuous over each region. Without loss of generality, we may also suppose that final weights $\weightT(\qloc,\val)$ are different from $+\infty$, for all configurations $(\qloc, \val)$ with $\qloc \in \QLocsT$. To do so, it suffices to forbid the jump into a region $I$ where the final weight function is constant equal to $+\infty$, by modifying the guard on the incoming transitions.

\subsection{\texorpdfstring{Restriction of $\F$ according to a strategy of \MinPl}{Restriction of F according to a strategy of Min}}
\label{subsec:fixpoint_restriction}

Before to start the proof of Theorem~\ref{theo:fixedpoint}, we establish a partial result
when we have fixed the strategy of \MinPl. 
In particular, we give a link between the value $\Value^\minstrategy$ of a 
strategy $\minstrategy$ of \MinPl and the restriction $\F^{\minstrategy}$ of $\F$ according to this strategy, by replacing 
the infimum for locations of \MinPl with the choice given by 
$\minstrategy$. However, since we can not (and we do not want to) suppose that 
$\minstrategy$ is memoryless, we need to extend the functions that 
$\F^{\minstrategy}$ use. Formally, $\F^{\minstrategy}$ is a new operator 
over functions $X \colon \FPlays \to \Rbar$ such that 
$\F^{\minstrategy}(X)(\play)$ is equal to 
\begin{displaymath} 
\begin{cases}
\weightfin(\qloc, \val) & \text{if $\qloc \in \QLocsT$} \\
\weight(\trans) + \delay \,\weight(\qloc) + 
X(\play \moveto{\delay, \trans} (\qloc', \val'))
& \text{if } \qloc \in \QLocsMin \text{ and } 
\minstrategy(\play) = (\delay, \trans) \\  
\sup_{(\qloc, \val) \moveto{\delay, \trans} (\qloc', \val')} 
\left(\weight(\trans) + \delay \,\weight(\qloc) + 
X(\play \moveto{\delay, \trans} (\qloc', \val')) \right) 
& \text{if } \qloc \in \QLocsMax \\ 
\end{cases}
\end{displaymath}
where $\last(\play) = (\qloc, \val)$. 

Since the value function $\Value^\minstrategy$ has only been defined for configurations, 
we need to extend it over all finite plays. To do that, we define the weight of a 
play given by two strategies ($\minstrategy$ and $\maxstrategy$) from a given 
finite play $\play$ by the weight of the unique play $\play'$ conforming to 
$\minstrategy$ and $\maxstrategy$ from the last configuration of $\play$ 
(when $\minstrategy$ and $\maxstrategy$ are initialised by $\play$), i.e. 
\begin{displaymath}
\weightP(\Play(\play, \minstrategy, \maxstrategy)) = \weight(\play') \,.
\end{displaymath}
Even if the weight of $\play$ is not taken into account in the 
weight of $\Play(\play, \minstrategy, \maxstrategy)$,
we observe that $\Play(\last(\play), \minstrategy, \maxstrategy)$ 
does not describe the same accumulated weight as $\Play(\play, \minstrategy, \maxstrategy)$ 
(since $\minstrategy$ and $\maxstrategy$ may use some memory).
We thus let, for all finite plays $\play$, 
\begin{displaymath}
\Value^{\minstrategy}(\play) = 
\sup_{\maxstrategy} \weightP(\Play(\play, \minstrategy, \maxstrategy)) \,.
\end{displaymath} 

\pagebreak
\begin{lem}
	\label{lem:ValueStrat}
	$\Value^{\minstrategy}$ is a fixpoint of
	$\F^{\minstrategy}$.
\end{lem}
\begin{proof}
	Let $\play$ be a finite play and $(\qloc, \val)$ its last configuration. 
	If $\qloc \in \QLocsT$, then for all 
	strategies $\maxstrategy$ of $\MaxPl$, 
	$\weightP(\Play(\play, \minstrategy, \maxstrategy)) = 
	\weightfin(\qloc, \val)$. 
	Thus, we obtain that 
	\begin{displaymath}
	\Value^{\minstrategy}(\play) = 
	\sup_{\maxstrategy} \weightP(\Play(\play, \minstrategy, \maxstrategy)) = 
	\weightfin(\qloc, \val) = \F^{\minstrategy}(\Value^{\minstrategy})(\play) 
	\end{displaymath}
	where the second equality follows by applying 
	the supremum over strategies of $\MaxPl$.
	
	Now, we suppose that $\qloc \in \QLocsMin$, and let
	$\minstrategy(\play) = (\delay, \trans)$. 
	Thus, for all strategies $\maxstrategy$ of $\MaxPl$, we obtain that 
	\begin{displaymath}
	\weightP(\Play(\play, \minstrategy, \maxstrategy)) = 
	\weight(\trans) + \delay \, \weight(\qloc) + 
	\weightP(\Play(\play \moveto{\delay, \trans} (\qloc', \val'), 
	\minstrategy, \maxstrategy)) \,.
	\end{displaymath}
	In particular, by applying the supremum over 
	strategies of \MaxPl, we obtain that 
	\begin{displaymath}
	\Value^{\minstrategy}(\play)
	% = \sup_{\maxstrategy} \weightP(\Play(\play, \minstrategy, \maxstrategy)) 
	= \sup_{\maxstrategy} \big(\weight(\trans) + \delay \, \weight(\qloc) + 
	\weightP(\Play(\play \moveto{\delay, \trans} (\qloc', \val'), 
	\minstrategy, \maxstrategy))\big) \,.
	\end{displaymath}
	We note that the choice $(\delay, \trans)$ of $\minstrategy$
	depends only on $\play$ that is independent of the chosen strategy of $\MaxPl$. 
	Thus, we deduce that
	\begin{align*}
	\Value^{\minstrategy}(\play) 
	&= \weight(\trans) + \delay \, \weight(\qloc) + 
	\sup_{\maxstrategy} \weightP(\Play(\play \moveto{\delay, \trans} (\qloc', \val'), 
	\minstrategy, \maxstrategy))  \\
	&= \weight(\trans) + \delay\, \weight(\qloc) + 
	\Value^{\minstrategy}(\play \moveto{\delay, \trans} (\qloc', \val')) 
	= \F^{\minstrategy}(\Value^{\minstrategy})(\play)\,.
	\end{align*}
	
	Finally, we suppose that $\qloc \in \QLocsMax$ and we reason by double 
	inequalities. We begin by showing that 
	$\F^{\minstrategy}(\Value^{\minstrategy})(\play) \leq \Value^{\minstrategy}(\play)$. 
	Let $\varepsilon > 0$, by the definition of $\F^{\minstrategy}(\Value^{\minstrategy})(\play)$, 
	we obtain the existence of an edge $(\qloc, \val) \moveto{\delay, \trans} (\qloc', \val')$ 
	such~that 
	\begin{displaymath}
	\F^{\minstrategy}(\Value^{\minstrategy})(\play) \leq 
	\weight(\trans) + \delay\,\weight(\qloc) + 
	\Value^{\minstrategy}(\play \moveto{\delay, \trans} (\qloc', \val')) 
	+ \frac{\varepsilon}{2} \,.
	\end{displaymath}
	Similarly, by the definition of $\Value^{\minstrategy}$, there exists 
	a strategy $\maxstrategyopt$ 
	for $\MaxPl$ such that
	\begin{displaymath}
	\Value^{\minstrategy}(\play \moveto{\delay, \trans} (\qloc', \val')) \leq 
	\weightP(\Play(\play \moveto{\delay, \trans} (\qloc', \val'), \minstrategy, 
	\maxstrategyopt) + \frac{\varepsilon}{2} \,.
	\end{displaymath} 
	In particular, by combining these two inequalities, we obtain:    
	\begin{displaymath}
	\F^{\minstrategy}(\Value^{\minstrategy})(\play) \leq 
	\weight(\trans) + \delay\,\weight(\qloc) + 
	\weightP(\Play(\play \moveto{\delay, \trans} (\qloc', \val'), 
	\minstrategy, \maxstrategyopt) + \varepsilon \,.
	\end{displaymath}
	We consider a new strategy $\maxstrategy$ for \MaxPl 
	defined such that $\maxstrategy(\play) = (\delay, \trans)$ and 
	$\maxstrategy(\play') = \maxstrategyopt(\play')$, for all 
	finite plays $\play' \neq \play$. In particular, since $\maxstrategy$ 
	and $\maxstrategyopt$ make the same choice for all plays that extend 
	$\play \moveto{\delay, \trans} (\qloc', \val')$, we obtain that 
	$\weightP(\Play(\play \moveto{\delay, \trans} 
	(\qloc', \val'), \minstrategy, \maxstrategyopt) = 
	\weightP(\Play(\play \moveto{\delay, \trans} (\qloc', \val'), 
	\minstrategy, \maxstrategy)$. 
	Thus, we deduce that  
	\begin{displaymath}
	\F^{\minstrategy}(\Value^{\minstrategy})(\play) 
	\leq \weightP(\Play(\play, \minstrategy, \maxstrategy)) + \varepsilon 
	\leq \sup_{\maxstrategy'} 
	\big(\weightP(\Play(\play, \minstrategy, \maxstrategy'))\big) + \varepsilon  
	= \Value^{\minstrategy}(\play) + \varepsilon  \,.
	\end{displaymath}
	Since this inequality holds for all $\varepsilon > 0$, it follows that 
	$\F^{\minstrategy}(\Value^{\minstrategy})(\play) \leq \Value^{\minstrategy}(\play)$.  
	
	Conversely, we prove that $\Value^{\minstrategy}(\play) \leq 
	\F^{\minstrategy}(\Value^{\minstrategy})(\play)$. Let $\varepsilon > 0$, and 
	as for the previous inequality, there exists a strategy $\maxstrategyopt$ of $\MaxPl$ such that 
	\begin{displaymath}
	\Value^{\minstrategy}(\play) - \varepsilon \leq 
	\weightP(\Play(\play, \minstrategy, \maxstrategyopt)) \,. 
	\end{displaymath}
	In particular, by letting $(\delay, \trans) = \maxstrategyopt(\play)$, we 
	deduce that 
	\begin{align*}
	\Value^{\minstrategy}(\play) - \varepsilon &\leq
	\weight(\trans) + \delay \,\weight(\qloc) + 
	\weightP(\Play(\play \moveto{\delay, \trans}(\qloc', \val'), 
	\minstrategy, \maxstrategyopt)) \\ 
	&\leq \weight(\trans) + \delay \,\weight(\qloc) + 
	\Value^{\minstrategy}(\play \moveto{\delay, \trans}(\qloc', \val')) \\
	&\leq \sup_{(\qloc, \val) \moveto{\delay, \trans} (\qloc', \val')} 
	\big(\weight(\trans) + \delay \,\weight(\qloc) + 
	\Value^{\minstrategy}(\play \moveto{\delay, \trans}(\qloc', \val')) \big) 
	= \F^{\minstrategy}(\Value^{\minstrategy})(\play) \,.
	\end{align*}
	Finally, since this inequality holds for all $\varepsilon > 0$, we obtain that 
	$\Value^{\minstrategy}(\play) \leq \F^{\minstrategy}(\Value^{\minstrategy})(\play)$.
\end{proof}

\subsection{\texorpdfstring{Iterates of $\F$ uniformly converge to the greatest fixpoint of $\F$}{Iterates of F uniformly converge to the greatest fixpoint of F}}
\label{subsec:fixpoint_Kleene}

We now prove the first result needed in the proof of Theorem~\ref{theo:fixedpoint}.
In particular, we consider the sequence $(V_i)_i$ of functions $\QLocs \times \Rpos\to \Rbar$ defined,
for all $i\in \N$ and for all configurations $(\qloc, \val)$ by
\begin{displaymath}
V_i(\qloc, \val) = 
\begin{cases}
+\infty & \text{if $i = 0$ and $\qloc \notin \QLocsT$} \\
\weightfin(\qloc, \val) & \text{if $i = 0$ and $\qloc \in \QLocsT$} \\
\F(V_{i-1})(\qloc, \val) & \text{otherwise.}
\end{cases}
\end{displaymath}

\begin{prop}
	\label{prop:fixpoint}
	$\inf_i V_i$ is the greatest fixpoint of $\F$.
\end{prop}

This section is devoted to the proof of this proposition.
In particular, our proof relies on the following technical 
results\footnote{These results hold for all \WTG{s} and not only one-clock \WTG.} 
providing sufficient condition on the limit of the sequence $(V_i)_i$ to be a fixpoint of $\F$.

\begin{lem}
	\label{lem:Fproperties}
	\begin{enumerate}
		\item\label{item:Fproperties_monotonic} 
		$\F$ is monotonous\footnote{A function $f\colon X\to Y$ over partial orders $X$ and $Y$ is monotonous if for all $x\leq x'$ in $X$, we have $f(x) \leq f(x')$ in $Y$.} over $\QLocs \times \Rpos \to \Rbar$, where the partial order over $\QLocs \times \Rpos$ is the pointwise order over $\QLocs$ and the usual order over $\Rpos$.
		\item\label{item:Fproperties_constant} 
		For all $X\colon \QLocs \times \Rpos\to \Rbar$ and $a \geq 0$, $\F(X + a) \leq \F(X) + a$.
		\item\label{item:Fproperties_triangle} 
		For all non-increasing sequences $(X_i)_i$ of functions $X_i\colon \QLocs \times \Rpos\to \Rbar$ that 
		uniformly converge over each region\footnote{A sequence of functions $(f_i)_i$ from partial orders $X$ to $\Rbar$ uniformly converge over $A\subset X$ towards a function $f$ if for all $\varepsilon>0$, there exists $N\in \N$ such that for all $i\geq N$ and $x\in A$, $|f_i(x)-f(x)|\leq \varepsilon$.},
		$\inf_i \F(X_i) = \F(\inf_i X_i)$.
	\end{enumerate}
\end{lem}
\begin{proof}
	\begin{enumerate}
		\item Let $X, X' \colon \QLocs \times \Rpos\to \Rbar$ be two 
		functions such that $X \geq X'$ (i.e.~$X(\qloc,\val) \geq X'(\qloc, \val)$ for all configurations $(q,\val)$), 
		and let $(\qloc, \val)$ be a configuration.
		If $\qloc \in \QLocsT$, then $\F(X)(\qloc, \val) = \weightfin(\qloc, \val) 
		= \F(X')(\qloc, \val)$. Otherwise, since $X(\qloc', \val') \geq X'(\qloc', \val')$, 
		for all edges $(\qloc, \val) \moveto{\delay, \trans} (\qloc', \val')$, 
		we have:
		\begin{displaymath}
		\weight(\trans) + \delay \, \weight(\qloc) + X(\qloc', \val') \geq
		\weight(\trans) + \delay \, \weight(\qloc) + X'(\qloc', \val')
		\end{displaymath}
		Finally, we apply the infimum (resp. supremum) over all edges 
		in this inequality if $\qloc \in \QLocsMin$ (resp. $\qloc \in \QLocsMax$).

		\item Let $(\qloc, \val)$ be a configuration. 
		\begin{itemize}
			\item If $\qloc \in \QLocsT$, then, since $a \geq 0$, we have:   
			\begin{displaymath}
			\F(X + a)(\qloc, \val) = \weightfin(\qloc, \val) = \F(X)(\qloc, \val) 
			\leq \F(X)(\qloc, \val) + a \,.
			\end{displaymath} 
			\item If $\qloc \in \QLocsMin$, then, since $a$ does not depend on edges, 
			we have:  
			\begin{align*}
			\F(X + a)(\qloc, \val) 
			&= \inf_{(\qloc, \val) \moveto{\delay, \trans} (\qloc', \val')}
			\big(\weight(\trans) + \delay \, \weight(\qloc) + X(\qloc', \val') + a\big) \\
			&= \inf_{(\qloc, \val) \moveto{\delay, \trans} (\qloc', \val')}
			\big(\weight(\trans) + \delay \, \weight(\qloc) + X(\qloc', \val')\big) + a \\
			&= \F(X)(\qloc, \val) + a \,.
			\end{align*} 
			\item If $\qloc \in \QLocsMax$, then, for the same reason, we have  
			$\F(X + a)(\qloc, \val) = \F(X)(\qloc, \val) + a$.
		\end{itemize}

		\item 
		%Let $(\qloc, \val)$ be a configuration. 
		Since $\F$ is monotonous (by item~\eqref{item:Fproperties_monotonic}), 
		we remark that for all $j \in \N$, we have 
		$\F(X_j) \geq \F(\inf_i X_i)$. 
		In particular, as this inequality holds for all $j \in \N$, we obtain that 
		\begin{displaymath}
		\inf_i \F(X_i) \geq \F(\inf_i X_i) \,.
		\end{displaymath}
		Conversely, let $\varepsilon > 0$ and $I$ be a region. Since $(X_i)_i$ uniformly converges 
		over $\reg$ to $\inf_i X_i$ (since the sequence $(X_i)$ is non-increasing), there exists $j_{\reg} \in \N$ such that 
		$X_{j_\reg} \leq \inf_i X_i + \varepsilon$ over $\reg$. Now, since there are
		only a finite number of regions, we fix $j = \max_{\reg} j_\reg$. 
		Thus, since the sequence $(X_i)_i$ is non-increasing, for all regions $\reg$, 
		$X_{j} \leq \inf_i X_i + \varepsilon$. 
		Since $\F$ is monotonous and by item~\eqref{item:Fproperties_constant}, 
		\begin{displaymath}
		\F(X_j) 
		\leq \F\big(\inf_i X_i + \varepsilon\big)
		\leq  \F\big(\inf_i X_i\big) + \varepsilon \,.
		\end{displaymath}
		In particular, we deduce that $\inf_i \F(X_i) \leq  \F\big(\inf_i X_i\big) + \varepsilon$, 
		for all $\varepsilon > 0$.
		\qedhere
	\end{enumerate}
\end{proof}

As a corollary of this result, we prove that $\inf_i V_i$ is a fixpoint of $\F$ by proving that 
it is a non-increasing sequence that uniformly converges. 
In particular, we observe that this sequence of functions is non-increasing 
since $\F$ is monotonous (by Lemma~\ref{lem:Fproperties}.\eqref{item:Fproperties_monotonic}) and 
$V_0 \geq \F(V_0)$ (since $V_0(\qloc, \val) = +\infty$, or 
$V_0(\qloc, \val) = V_1(\qloc, \val) = \weightT(\qloc, \val)$). 
In particular, it (simply) converges to $\inf_i V_i$.
To prove that $(V_i)_i$ uniformly converges to $\inf_i V_i$, 
we will use Dini's theorem: a sequence of continuous 
functions that (simply) converges to a continuous function, 
uniformly converges. The main difficulty is to prove that 
$\inf_i V_i$ is continuous over regions. To do it, we note that 
if a sequence of $k$-Lipschitz-continuous functions (simply) 
converges, then its limit is a continuous function. In particular, 
we want to show that there exists 
$k \in \Rpos$ such that, for all $i \in \N$, $V_i$ is 
$k$-Lipschitz-continuous. 

\begin{defi}
	A function $f \colon \Rpos \to \Rbar$ is \emph{continuous (respectively, 
	$k$-Lipschitz-continuous, for $k \in \Rpos$) on regions} if 
	for all regions $\reg$, the restriction of $f$ over each region is a 
	continuous (respectively, $k$-Lipschitz-continuous) function. 
	A function $f \colon \Rpos \to \Rbar$ is 
	\emph{$k$-Lipschitz-continuous on regions}, for 
	$k \in \Rpos$, if for all regions $\reg$ and all 
	valuations $\val, \val' \in \reg$,
	$|f(\val) - f(\val')| \leq \Lambda|\val - \val'|$. 

	A function $X \colon \QLocs \times \Rpos \to \Rbar$ is 
	continuous (respectively, $k$-Lipschitz-continuous) on regions, 
	if for all locations $\qloc \in \QLocs$, the restriction of $X$ to
	$\qloc$ is continuous (respectively, $k$-Lipschitz-continuous).
\end{defi}

We let $\Lambda$ be the maximum absolute value of all weights of locations and of derivatives that appear 
in the piecewise-affine functions (the slopes of the affine pieces) of $\weightT$. Then, $V_0$ is trivially $\Lambda$-Lipschitz-continuous on regions. 
Indeed, being a $\Lambda$-Lipschitz-continuous function on regions when already being a continuous and piecewise 
affine function with finitely many pieces is equivalent to having its derivatives bounded by $\Lambda$ in absolute value. 
In~\cite[Lemma~10.10]{Bus19}, it is shown in all \WTG{s}, for all $i \in \N$, $V_i$ 
is $\Lambda_i$-Lipschitz-continuous on regions for a constant $\Lambda_i$ that depends on $i$.
 We now refine the proof, in our one-clock setting, to show that the same constant $\Lambda$ can be chosen for all $i$. 

\begin{lem}
	\label{lem:fixpoint_Lipschitz}
	For all $i \in \N$, $V_i$ 
	is $\Lambda$-Lipschitz-continuous on regions.
\end{lem}
\begin{proof}
	We rely on the knowledge that for all $i\in \N$, $V_i$ is continuous and piecewise 
	affine on each regions, with finitely many pieces, i.e. each $V_i$ has a finite number
	of cutpoints\footnote{We recall that a cutpoint is the value of the clock in-between two affine pieces of the function.}.

	We reason by induction on $i \in \N$ showing that the derivative of $V_i$ is bounded by $\Lambda$ in absolute values. The base case $i=0$ is trivially satisfied as seen above. Let $i\in \N$ be such that 
	$V_i$ is continuous on regions and piecewise affine with finitely many pieces that have  
	a derivative bounded by $\Lambda$. Let $\qloc \in \QLocs\setminus \QLocsT$ 
	(otherwise, we conclude as for $i=0$). 
	%The continuity of $\F(V_i)$ 
	%on regions holds because continuous functions are stable by translation, infimum and supremum. 
	By massaging the definition of $\F$, we have that
	\begin{displaymath}
  V_{i+1}(\qloc,\val) =
  \begin{cases}
    \min_\trans \inf_{(\qloc,\val) \moveto{\delay, \trans} (\qloc',\val')} 
	\big(\weight(\trans) + \delay\, \weight(\qloc) + V_i(\qloc', \val')\big) 
                          & \text{if } \qloc \in \QLocsMin \\
    \max_\trans \sup_{(\qloc,\val) \moveto{\delay, \trans} (\qloc',\val')} 
	\big(\weight(\trans) + \delay\, \weight(\qloc) + V_i(\qloc', \val')\big) 
                          & \text{if } \qloc \in \QLocsMax 
  \end{cases}
  \end{displaymath}
	For a fixed valuation $\val$, and once chosen the transition $\trans$ in the minimum or maximum,
	there are finitely many delays $\delay$ to consider in the
	infimum or supremum: since $V_i$ is piecewise affine, they are either delay $0$ or all delays $\delay$ such that $\val+\delay$ are cutpoints $\val^c$ of 
	$V_i(\qloc', \cdot)$. 
	In particular, since there 
	is only a finite number of such cutpoints, the function $\F(V_i)(\qloc, \cdot)$ can 
	be written as a finite nesting of $\min$ and $\max$ operations over affine terms, 
	each corresponding to a choice of delay and a transition to take. There 
	are several cases to define those terms, depending on the chosen transition $\delta$ and cutpoint $\val^c$. 
	If the transition $\trans$ 
	resets $\clockx$: 
	\begin{itemize}
	\item if a delay $0$ is chosen, then the affine term is $V_i(\qloc', 0)+\weight(\trans)$ that has derivative 0;
	\item otherwise, the affine term that it generates is of the form:
	\begin{displaymath}
	(\val^c-\val) \, \weight(\qloc) + \weight(\trans) + V_i(\qloc', 0)
	\end{displaymath}
	whose derivative is bounded by $\maxWeightLoc$ in absolute value, and thus by $\Lambda$.
	\end{itemize}
	 If the transition 
	$\trans$ does not reset $\clockx$: 
	\begin{itemize}
		\item if a delay $0$ is chosen, then the affine term is
		$\weight(\trans) + V_i(\qloc', \val)$, whose derivative is the same as in $V_i(\qloc',\cdot)$ and thus bounded by $\Lambda$ in absolute value;
		
		\item otherwise, the affine term that it generates is of the form:
		\begin{displaymath}
		(\val^c-\val) \, \weight(\qloc) + \weight(\trans) + V_i(\qloc', \val^c) \,.
		\end{displaymath}
		whose derivative is bounded by $\maxWeightLoc$ in absolute value, and thus by $\Lambda$.\qedhere
	\end{itemize}
\end{proof}

Now, we have tools to prove Proposition~\ref{prop:fixpoint}.
First, we prove that $\inf_i V_i$ is a fixed point of $\F$, 
i.e.~$\inf_i V_i = \F(\inf_i V_i)$.
By Lemma~\ref{lem:fixpoint_Lipschitz}, we know that for all $i \in \N$, 
$V_i$ is $\Lambda$-Lipschitz-continuous over regions. 
Thus, we deduce that $(V_i)_i$ converges to a continuous function over regions, 
i.e.~$\inf_i V_i(\qloc)$ is continuous over regions, for all locations $\qloc$. 
Now, by Dini's theorem, we deduce that $(V_i)_i$ uniformly converges over regions 
to $\inf_i V_i$. Finally, we apply 
Lemma~\ref{lem:Fproperties}.\eqref{item:Fproperties_triangle} to conclude that 
$\inf_i V_i = \inf_i \F(V_i) = \F(\inf_i V_i)$, and thus that 
$\inf_i V_i$ is a fixpoint of $\F$.
	
Finally, we prove that $\inf_i V_i$ is the greatest fixpoint $V$ of $\F$. 
As $V$ is the greatest fixpoint, 
we have $\inf_i V_i \leq V$. Conversely, we prove 
by induction on $i\in \N$ that $V \leq V_i$. If $i= 0$ and $\qloc \notin \QLocsT$, 
then $V_0(\qloc, \val) = +\infty$ and $V(\qloc, \val) \leq V_0(\qloc, \val)$; 
otherwise, $\qloc \in \QLocsT$ and $V_0(\qloc, \val) = \weightfin(\qloc, \val)$, while 
$V(\qloc, \val) = \weightfin(\qloc, \val)$ (since $V$ is a fixpoint of $\F$). 
If $i \in \N$ is such that 
$V \leq V_i$, as $\F$ is monotonous, 
we have $\F(V) \leq \F(V_i)$. 
Thus, since $V$ is a fixpoint of $\F$, we deduce that 
$V = \F(V)\leq \F(V_i) = V_{i+1}$ that concludes the proof 
of Proposition~\ref{prop:fixpoint}.

\subsection{\texorpdfstring{The greatest fixpoint of $\F$ is equal to the value function}{The greatest fixpoint of F is equal to the value function}}
\label{subsec:fixpoint_Val}

To conclude the proof of Theorem~\ref{theo:fixedpoint}, its remains to 
prove that $\inf_i V_i = \Value$. To do it, we adapt the 
proof given in~\cite{Bou16} to our context (turn-based games with 
negative and positive weights).

\begin{prop}
	\label{prop:fixpoint-limit}
	$\inf_i V_i = \Value$
\end{prop}

The main idea of this proof is the link between $V_i$ 
and the value obtained when we consider only plays with at most $i$ steps. 
We thus let $W_i$ be the configurations from where \MinPl 
can guarantee to reach a target location within $i$ steps: this is a very classical sequence of configurations that is traditionally called \emph{attractor}. Intuitively, for a configuration not in $W_i$, $V_i$ is equal to $+\infty$ since \MaxPl can avoid the target in the $i$ first 
steps. Formally, we define the sequence of $(W_i)_i$ by induction on 
$i \in \N$: $(\qloc, \val) \in W_0$ if $\qloc \in \QLocsT$, and for all $i \in \N$, 
$(\qloc, \val) \in W_{i+1}$ if $(\qloc, \val) \in W_i$, or
\begin{enumerate}
	\item\label{item:Wi_Min} $\qloc \in \QLocsMin$, and there exists an edge  
	$(\qloc, \val) \moveto{\delay, \trans} (\qloc', \val')$ such that  
	$(\qloc', \val') \in W_i$;
	
	\item\label{item:Wi_Max} $\qloc \in \QLocsMax$, and for all edges 
	$(\qloc, \val) \moveto{\delay, \trans} (\qloc', \val')$, we have 
	$(\qloc', \val') \in W_i$.
\end{enumerate}
The following lemma recalls the link between $V_i$ and $W_i$: 

\begin{lem}
	\label{lem:WinValueFinite}
	Let $i \in \N$ and $(\qloc, \val)$ be a configuration. Then, 
	$(\qloc, \val) \in W_i$ if and only if $V_i(\qloc, \val) < +\infty$.
\end{lem}
\begin{proof}
	We prove the equivalence by induction on $i \in \N$. If $i = 0$, 
	since the game has been modified so that final weight functions are finite, we conclude by 
	definitions of $W_0$ and $V_0$. Now, we fix $i \in \N$ such that 
	for all configurations $(\qloc', \val')$, we have 
	$(\qloc', \val') \in W_i$ if and only if $V_i(\qloc', \val') < + \infty$. 
	Let $(\qloc, \val)$ be a configuration. First, we suppose that 
	$\qloc \in \LocsT$. In this case, $(\qloc, \val) \in W^{i+1}$ and 
	$V_{i+1}(\qloc, \val) = \weightT(\qloc, \val) < +\infty$ (by 
	hypothesis). 
		
	Now, we suppose that $\qloc \in \QLocsMin$ and we have: 
	\begin{align*}
	V_{i+1}(\qloc, \val) &= \F(V_i)(\qloc, \val) = 
	\inf_{(\qloc, \val) \moveto{\delay, \trans} (\qloc', \val')} 
	\big(\weight(\trans) + \delay \, \weight(\qloc) + V_i(\qloc', \val') \big) 
	\end{align*}
	In particular, $V_{i+1}(\qloc, \val) < +\infty$ 
	if and only if there exists an edge $(\qloc, \val) \moveto{\delay, \trans} 
	(\qloc', \val')$ such that $V_i(\qloc', \val') < +\infty$. We deduce that 
	$V_{i+1}(\qloc, \val) < +\infty$ if and only if there exists 
	an edge $(\qloc, \val) \moveto{\delay, \trans} (\qloc', \val')$ such 
	that $(\qloc', \val') \in W_i$ (by applying the 
	inductive hypothesis on $(\qloc', \val')$). We conclude that 
	$V_{i+1}(\qloc, \val) < +\infty$ if and only if 
	$(\qloc, \val) \in W^{i+1}$, by item~\eqref{item:Wi_Min} of the 
	definition of~$W_{i+1}$. 
	
	Finally, we suppose that $\qloc \in \QLocsMax$ and we have: 
	\begin{align*}
	V_{i+1}(\qloc, \val) &= \F(V_i)(\qloc, \val) = 
	\sup_{(\qloc, \val) \moveto{\delay, \trans} (\qloc', \val')} 
	\big(\weight(\trans) + \delay \, \weight(\qloc) + V_i(\qloc', \val') \big) 
	\end{align*}
	In particular, by inductive hypothesis, $V_{i+1}(\qloc, \val) < +\infty$ 
	if and only if for all edges $(\qloc, \val) \moveto{\delay, \trans} 
	(\qloc', \val')$, we have $(\qloc', \val') \in W_i$. Thus, by 
	item~\eqref{item:Wi_Max} of the definition of~$W_{i+1}$, we obtain
	that $V_{i+1}(\qloc, \val) < +\infty$ if and only if $(\qloc, \val) \in W_{i+1}$.
\end{proof}

To prove that the value iteration converges to the value function, we relate configurations in $W_i$ with some particular strategies of \MinPl. Given a configuration 
$(\qloc, \val)$, we fix $\minStratWin_i(\qloc, \val)$ to be the set of 
strategies of \MinPl such that all plays from $(\qloc, \val)$ conforming to it
reach the target \emph{in at most $i$ steps}. More precisely, we require that for all plays starting from $(\qloc,\val)$ and conforming to a strategy of $\minStratWin_i(\qloc, \val)$, the $j$th configuration of the play belongs to $W_{i-j}$: in particular, the first configuration, $(\qloc,\val)$ must be in $W_i$, and the last one in $W_0$ (i.e.~with a location being a target).

For all $\varepsilon > 0$, we inductively define a sequence of 
strategies $(\minstrategy^{\varepsilon}_i)_i$ whose $i$-th strategy 
will be shown to belong to $\minStratWin_i(\qloc, \val)$ if $(\qloc, \val)\in W_i$, and
$\varepsilon$-optimal according to $V_i$. 
In particular, we prove that an almost-optimal strategy can be 
defined by choosing almost-optimal edges along the play. 
Intuitively, the $i$-th strategy chooses the first move as the best edge 
according to $V_{i}$, and then follows the $(i-1)$-th strategy
(applying in the suffix of the play except the first choice). 

Formally, we let $\minstrategy^{\varepsilon}_0$ be any fixed strategy of \MinPl.
For $i \in \N$, relying on $\minstrategy^{\varepsilon/2}_{i}$, 
we inductively define $\minstrategy^{\varepsilon}_{i+1}$ according 
to the length of all finite plays ending in a location of \MinPl. 
If the play contains only one configuration, we fix $\minstrategy^{\varepsilon}_{i+1}(\qloc, \val)$
be any decision $(\delay, \trans)$ such that $(\qloc, \val) \moveto{\delay, \trans} 
(\qloc', \val')$ and $\weight(\trans) + \delay \, \weight(\qloc) + 
V_{i}(\qloc', \val')\leq V_{i+1}(\qloc,\val)+\varepsilon/2$ (that exists by definition of $V_{i+1}(\qloc,\val)$ as an infimum).
%\begin{displaymath}
%\minstrategy^{\varepsilon}_{i+1}(\qloc, \val) = 
%\arginfepsilon[\varepsilon/2]_{(\qloc, \val) \moveto{\delay, \trans} 
%	(\qloc', \val')} \big(\weight(\trans) + \delay \, \weight(\qloc) + 
%V_{i+1}(\qloc', \val')\big) \,.
%\end{displaymath}
Otherwise, the play can be decomposed as $(\qloc, \val) \moveto{\delay, \trans} 
\play$ with $(\qloc', \val')$ the first configuration of $\play$, and we let:
\begin{displaymath}
\minstrategy^{\varepsilon}_{i+1}((\qloc, \val) \moveto{\delay, \trans} \play) 
= \minstrategy^{\varepsilon/2}_i(\play) \,.
\end{displaymath} 

\begin{lem}
	\label{lem:partialMinStrat}
	For all $i \in \N$, $\varepsilon > 0$ and $(\qloc, \val)\in W_i$, 
   	\begin{center}
   	$\minstrategy^{\varepsilon}_i \in \minStratWin_i(\qloc, \val)$ \quad and\quad
	$V_{i}(\qloc, \val) + \varepsilon \geq 
	\Value^{\minstrategy^{\varepsilon}_i}(\qloc, \val)$\,.
	\end{center}
\end{lem}
\begin{proof}
	We reason by induction on $i \in \N$. If $i = 0$, since $(\qloc, \val)\in W_0$, we have $\qloc \in \QLocsT$ and thus any strategy (and thus the fixed strategy $\minstrategy_0^\varepsilon$) is in 
	$\minStratWin_0(\qloc, \val)$, and 
	$V_0(\qloc, \val) = \weightT(\qloc, \val) = 
	\Value^{\minstrategy^{\varepsilon}_0}(\qloc, \val)$.
	
	Now, consider $i \in \N$ such that for all 
	configurations $(\qloc, \val)\in W_i$,
	$\minstrategy^{\varepsilon/2}_i \in \minStratWin_i(\qloc, \val)$, and 
	$V_{i}(\qloc, \val) + \varepsilon/2 \geq 
	\Value^{\minstrategy^{\varepsilon/2}_i}(\qloc, \val)$. We show that 
	$\minstrategy^{\varepsilon}_{i+1}$ satisfies the properties 
	for a given configuration $(\qloc, \val)\in W_{i+1}$. If $\qloc\in \QLocsT$, 
	we conclude as in the case~$i = 0$ (since $\qloc \in W_0$).
	Otherwise, we show
	$\minstrategy^{\varepsilon}_{i+1} \in \minStratWin_{i+1}(\qloc, \val)$ 
	by contradiction. We thus suppose 
	that there exists a finite play $\play'$ of length $i+1$ 
	conforming to $\minstrategy^{\varepsilon}_{i+1}$ that 
	does not reach $\QLocsT$. It can be 
	decomposed as $(\qloc, \val) \moveto{\delay, \trans} \play$ 
	where $\minstrategy^{\varepsilon}_{i+1}(\qloc, \val) = 
	(\delay, \trans)$ and $\play$ is conforming to 
	$\minstrategy^{\varepsilon/2}_i$. 
	We show that $(\qloc', \val') \in W_i$ where 
	$(\qloc', \val')$ is the first configuration of $\play$. 
	\begin{itemize}
		\item If $\qloc \in \QLocsMax$, then we conclude 
		that $(\qloc', \val') \in W_i$ by item~\eqref{item:Wi_Max} 
		of definition of $W_{i+1}$: all edges from $(\qloc, \val)$ 
		reach a configuration in $W_i$. 
		
		\item If $\qloc \in \QLocsMin$, then, by 
		item~\eqref{item:Wi_Min} of definition of $W_{i+1}$, 
		there exists an edge $(\qloc, \val) \moveto{\delay', \trans'} 
		(\qloc'', \val'')$ such that $(\qloc'', \val'') \in W_i$, i.e.~$V_i(\qloc'',\val'') < +\infty$ (by 
		Lemma~\ref{lem:WinValueFinite}). The choice of   
		$\minstrategy^{\varepsilon}_{i+1}(\qloc, \val)$ is taken along all possible edges 
		from $(\qloc, \val)$, at most $\varepsilon$ away of the infimum. Thus, 
		it chooses an edge $(\qloc, \val) \moveto{\delay, \trans} 
		(\qloc', \val')$ such that $V_i(\qloc',\val') < +\infty$, 
		i.e.~$(\qloc',\val') \in W_i$ (by Lemma~\ref{lem:WinValueFinite}). 
	\end{itemize}
	By induction hypothesis
	applied to $(\qloc', \val') \in W_i$, 
	$\minstrategy^{\varepsilon/2}_i \in 
	\minStratWin_{i}(\qloc', \val')$, and thus $\play$ reaches $\QLocsT$ within $i$ steps which  contradicts	the hypothesis.
	
%	Now, we prove that if $\minstrategy^{\varepsilon}_{i+1} \in 
%	\minStratWin_{i+1}(\qloc, \val)$, then $(\qloc, \val) \in W_{i+1}$. We 
%	reason by contradiction: we suppose that $(\qloc, \val) \notin W_{i+1}$, and 
%	we want to define a play from $(\qloc, \val)$ of length $i+1$ and conforming 
%	$\minstrategy^{\varepsilon}_{i+1}$ such that it does not reach the target. 
%	To do it, by letting $(\qloc, \val) \moveto{\delay, \trans} 
%	(\qloc', \val')$ be the first step of this play, and we note that the 
%	definition of $W_{i+1}$ ensures us that $(\qloc', \val') \notin W_i$.
%	Now, by inductive hypothesis applying on $(\qloc', \val')$, we obtain 
%	that $\minstrategy^{\varepsilon/2}_{i} \notin 
%	\minStratWin_{i}(\qloc', \val')$. In particular, there exists $\play$ 
%	be a play of length $i$ from $(\qloc', \val')$ and conforming 
%	$\minstrategy^{\varepsilon/2}_{i}$ that does not reach a target 
%	location. Thus, by definition of $\minstrategy^{\varepsilon}_{i+1}$, 
%	there exists $(\qloc, \val) \moveto{\delay, \trans} \play$ be a play 
%	of length $i + 1$ from $(\qloc, \val)$ and conforming 
%	$\minstrategy^{\varepsilon}_{i+1}$ that does not reach a target 
%	location. Finally, we conclude that 
%	$\minstrategy^{\varepsilon}_{i+1} \notin \minStratWin_{i+1}(\qloc, \val)$ and 
%	we obtain a contradiction.
	
	We then prove that $V_{i+1}(\qloc, \val) + \varepsilon \geq 
	\Value^{\minstrategy^{\varepsilon}_{i+1}}(\qloc, \val)$. By 
	definition of $\minstrategy^{\varepsilon}_{i+1}$ with 
	$\minstrategy^{\varepsilon/2}_i$, we remark that, for all finite 
	plays $(\qloc, \val) \moveto{\delay, \trans} \play$ of length at 
	least one, we have $\minstrategy^{\varepsilon}_{i+1}((\qloc, \val) 
	\moveto{\delay, \trans} \play) = \minstrategy^{\varepsilon/2}_i(\play)$.
	In particular, the weight of all plays from $(\qloc, \val) 
	\moveto{\delay, \trans} \play$ and conforming to 
	$\minstrategy^{\varepsilon}_{i+1}$ is equal to the weight of the play 
	from $\play$ and conforming to $\minstrategy^{\varepsilon/2}_{i}$ under 
	the same strategy of \MaxPl, i.e.~for all strategies of \MaxPl, 
	$\maxstrategy$, we have: 
	\begin{displaymath}
	\weight(\Play((\qloc, \val) \moveto{\delay, \trans} (\qloc', \val'), 
	\minstrategy^{\varepsilon}_{i+1}, \maxstrategy)) = 
	\weight(\Play((\qloc', \val'), \minstrategy^{\varepsilon/2}_i, 
	\maxstrategy)) \,.
	\end{displaymath} 
	Thus, by applying the supremum over strategies 
	of \MaxPl, we deduce~that 
	\begin{equation}
	\label{eq:app-eqVal}
	\Value^{\minstrategy^{\varepsilon}_{i+1}}((\qloc, \val) 
	\moveto{\delay, \trans} (\qloc', \val')) = 
	\Value^{\minstrategy^{\varepsilon/2}_i}(\qloc', \val') 
	\end{equation}
	\begin{itemize}
		\item If $\qloc \in \QLocsMax$, then we~have:  
		\begin{align*}
		V_{i + 1}(\qloc, \val) &= \F(V_i)(\qloc, \val)
		= \sup_{(\qloc, \val) \moveto{\delay, \trans} (\qloc', \val')} 
		\big(\weight(\trans) + \delay \, \weight(\qloc) + V_i(\qloc', \val')\big) 
		\end{align*}
		Since $(\qloc, \val) \in W_{i+1}$, then we have $(\qloc', \val') \in W_i$ 
		(by item~\eqref{item:Wi_Max}). 
		Moreover, by induction hypothesis applying on $V_i(\qloc', \val')$, we deduce that
		$V_{i}(\qloc', \val') + \varepsilon/2 \geq 
		\Value^{\minstrategy^{\varepsilon/2}_i}(\qloc', \val')$. 
		Thus, for all $(\delay, \trans)$, we have:
		\begin{align*}
		V_{i + 1}(\qloc, \val) &\geq \weight(\trans) + \delay \, \weight(\qloc) + 
		V_i(\qloc', \val')\\
		&\geq \weight(\trans) + \delay \, \weight(\qloc) + 
		\Value^{\minstrategy^{\varepsilon/2}_i}(\qloc', \val') - \varepsilon/2 \\
		&\geq \weight(\trans) + \delay \, \weight(\qloc) + 
		\Value^{\minstrategy^{\varepsilon}_{i+1}}((\qloc, \val) 
		\moveto{\delay, \trans }(\qloc', \val')) - \varepsilon/2
		\qquad \text{(by~\eqref{eq:app-eqVal})}
		\end{align*}
		Finally, since this inequality holds for all edges from $(\qloc, \val)$, we 
		deduce that  
		\begin{align*}
		V_{i + 1}(\qloc, \val) &\geq 
		\sup_{(\qloc, \val) \moveto{\delay, \trans} (\qloc', \val')} 
		\big(\weight(\trans) + \delay \, \weight(\qloc) + 
		\Value^{\minstrategy^{\varepsilon}_{i+1}}((\qloc, \val) 
		\moveto{\delay, \trans}(\qloc', \val'))\big) - \varepsilon/2\\
		&\geq \Value^{\minstrategy^{\varepsilon}_{i+1}}(\qloc, \val) - \varepsilon
		\qquad \text{(by Lemma~\ref{lem:ValueStrat}).}
		\end{align*}
		
		\item If $\qloc \in \QLocsMin$, then, by definition of 
		$\minstrategy^{\varepsilon}_{i+1}$, and letting $\minstrategy^{\varepsilon}_{i+1}(\qloc, \val) = (\delay, \trans)$: 
		\begin{align*}
		V_{i + 1}(\qloc, \val) 
		&\geq \weight(\trans) + \delay \, \weight(\qloc) + V_i(\qloc', \val') - 
		\varepsilon/2
		\end{align*} 
		Now, since $(\qloc', \val') \in W_i$ (as explain before to show that 
		$\minstrategy^{\varepsilon}_{i+1} \in \minStratWin_{i+1}(\qloc,\val)$), by induction hypothesis, 
		$V_{i}(\qloc', \val') + \varepsilon/2 \geq 
		\Value^{\minstrategy^{\varepsilon/2}_i}(\qloc', \val')$. Thus, we deduce that
		\begin{align*}
		V_{i + 1}(\qloc, \val) 
		&\geq \weight(\trans) + \delay \, \weight(\qloc) + 
		\Value^{\minstrategy^{\varepsilon/2}_i}(\qloc', \val') - \varepsilon \\
		&\geq \weight(\trans) + \delay \, \weight(\qloc) + 
		\Value^{\minstrategy^{\varepsilon}_{i+1}}((\qloc, \val) 
		\moveto{\delay, \trans}(\qloc', \val'))	- \varepsilon 
		\qquad \text{(by~\eqref{eq:app-eqVal})} \\
		&\geq \Value^{\minstrategy^{\varepsilon}_{i+1}}(\qloc, \val) - \varepsilon 
		\qquad \text{(by Lemma~\ref{lem:ValueStrat}).}
		\qedhere
		\end{align*}
	\end{itemize}
\end{proof}

As a corollary, we obtain: 
\begin{lem}
	\label{lem:WinValueStrat}
	For all $i \in \N$, and $(\qloc, \val) \in W_i$,
	%\begin{displaymath}
	$V_i(\qloc, \val) = 
	\inf_{\minstrategy \in \minStratWin_i(\qloc, \val)} 
	\Value^{\minstrategy}(\qloc, \val)$.% \,.
	%\end{displaymath}
\end{lem}
\begin{proof}
	We reason by induction on $i\in \N$. If $i = 0$, since $\qloc \in \QLocsT$ 
	for all strategies $\minstrategy \in \minStratWin_0(\qloc, \val)$, 
	 $V_0(\qloc, \val) = \weightfin(\qloc, \val) = \Value^{\minstrategy}(\qloc, \val)$.
	 
	For $i \in \N$ such that the property holds, let
	$(\qloc, \val) \in W_{i+1}$. If
	$\qloc \in \QLocsT$, we have $(\qloc,\val) \in W_0$ and we conclude as in the case 
	$i = 0$. Otherwise, Lemma~\ref{lem:partialMinStrat} directly implies that 
%	
%	First, we show that $V_{i+1}(\qloc, \val) \geq 
%	\inf_{\minstrategy \in \minStratWin_{i+1}(\qloc, \val)} 
%	\Value^{\minstrategy}(\qloc, \val)$ by applying :  
%	for all $\varepsilon > 0$, we have $\minstrategy^{\varepsilon}_{i+1} \in 
%	\minStratWin_{i+1}(\qloc, \val)$, such that 
%	$V_{i+1}(\qloc, \val) + \varepsilon \geq 
%	\Value^{\minstrategy^{\varepsilon}_{i+1} }(\qloc, \val)$. 
%	Thus, by applying the infimum, we deduce that 
	\begin{displaymath}
	V_{i+1}(\qloc, \val)+\varepsilon\geq \Value^{\minstrategy^{\varepsilon}_{i+1}}(\qloc, \val) \geq 
	\inf_{\minstrategy \in \minStratWin_{i+1}(\qloc, \val)} 
	\Value^{\minstrategy}(\qloc, \val) 
	\end{displaymath} 
	and $V_{i+1}(\qloc, \val)\geq 
	\inf_{\minstrategy \in \minStratWin_{i+1}(\qloc, \val)} 
	\Value^{\minstrategy}(\qloc, \val)$ since the inequality holds 
	for all $\varepsilon > 0$.
	
	Conversely, we show that $V_{i+1}(\qloc, \val) \leq 
	\inf_{\minstrategy \in \minStratWin_{i+1}(\qloc, \val)} 
	\Value^{\minstrategy}(\qloc, \val)$ by proving that for all 
	$\minstrategy \in \minStratWin_{i+1}(\qloc, \val)$, we have 
	$V_{i+1}(\qloc, \val) \leq \Value^{\minstrategy}(\qloc, \val)$. 
	Let $\minstrategy \in \minStratWin_{i+1}(\qloc, \val)$. 
	\begin{itemize}
		\item If $\qloc \in \QLocsMin$, then we let $(\delay,\trans)= 
		\minstrategy(\qloc, \val)$ with $(\qloc,\val)\xrightarrow{\delay,\trans} (\qloc', \val')$, so that $(\qloc', \val')\in W_i$. By induction hypothesis, we have $V_i(\qloc', \val') = 
		\inf_{\minstrategy' \in \minStratWin_i(\qloc', \val')} \Value^{\minstrategy}(\qloc', \val')$. Consider the strategy 
		$\minstrategy'$ obtained from $\minstrategy$ by adding as a first move the edge $(\qloc, \val) \moveto{\delay, \trans} (\qloc', \val')$. Formally, it is defined by: 
		\begin{displaymath}
		\minstrategy_{\qloc', \val'}(\play) = 
		\begin{cases}
		\minstrategy\big((\qloc, \val) \moveto{\delay, \trans} \play\big) & 
		\text{if $\play$ starts in $(\qloc', \val')$;} \\
		\minstrategy(\play) & \text{otherwise.}
		\end{cases} 
		\end{displaymath} 
		Given a play $\play'$ conforming to $\minstrategy'$
		starting from $(\qloc', \val')$, we remark that 
		$(\qloc, \val) \moveto{\delay, \trans} \play'$ is conforming to $\minstrategy$. 
		In particular, we obtain that 
		\begin{eqnarray}
		\label{eq:WinValueStrat-1}
		\Play((\qloc, \val), \minstrategy', \maxstrategy) = 
		\Play((\qloc, \val) \moveto{\delay, \trans} (\qloc', \val'), 
		\minstrategy, \maxstrategy)
		\end{eqnarray}   		
		Thus, from~\eqref{eq:WinValueStrat-1}, we deduce that 
		$\minstrategy' \in \minStratWin_{i}(\qloc', \val')$.
		Thus, $V_i(\qloc', \val')
		\leq  \Value^{\minstrategy'}(\qloc', \val')$ and we obtain that
		\begin{align*}
		V_{i + 1}(\qloc, \val) &= \F(V_i)(\qloc, \val) 
		\\
		&\leq \weight(\trans) + \delay \, \weight(\qloc) + V_i(\qloc', \val')\\ 
		&\leq \weight(\trans) + \delay \, \weight(\qloc) + 
		\Value^{\minstrategy'}(\qloc', \val') \,.
		\end{align*}
		Moreover, by~\eqref{eq:WinValueStrat-1}, we also obtain that, 
		for all strategies $\maxstrategy$ of \MaxPl, 
		\begin{displaymath}
		\weightP(\Play((\qloc, \val), \minstrategy', \maxstrategy)) = 
		\weightP(\Play((\qloc, \val) \moveto{\delay, \trans} (\qloc', \val'), 
		\minstrategy, \maxstrategy)) \,.
		\end{displaymath}  
		In particular, we deduce that $\Value^{\minstrategy'}(\qloc', \val') = 
		\Value^{\minstrategy}((\qloc, \val) \moveto{\delay, \trans} (\qloc', \val'))$, 
		and we can rewrite the previous inequality as:  
		\begin{align*}
		V_{i + 1}(\qloc, \val) 
		&\leq \weight(\trans) + \delay \, \weight(\qloc) + 
		\Value^{\minstrategy}((\qloc, \val) \moveto{\delay, \trans} (\qloc', \val')) \\
		&\leq \Value^{\minstrategy}(\qloc, \val) 
		\qquad \text{(by Lemma~\ref{lem:ValueStrat}).}
		\end{align*}
		
		\item If $\qloc \in \QLocsMax$, then, by Lemma~\ref{lem:ValueStrat}, 
		we have:
		\begin{align*}
		\Value^{\minstrategy}(\qloc, \val) &= 
		\sup_{(\qloc, \val) \moveto{\delay, \trans} (\qloc', \val')} 
		\big(\weight(\trans) + \delay \,\weight(\qloc) + 
		\Value^{\minstrategy}((\qloc, \val) \xrightarrow{\delay, \trans}
		(\qloc', \val')) \big) 
		\end{align*}
		Letting $(\qloc, \val) \moveto{\delay, \trans} (\qloc', \val')$ be 
		an edge from $(\qloc, \val)$, since $\minstrategy \in 
		\minStratWin_{i+1}(\qloc, \val)$, we have $(\qloc',\val')\in W_i$, and thus by induction hypothesis, $V_{i}(\qloc', \val') \leq 
	\inf_{\minstrategy' \in \minStratWin_{i}(\qloc', \val')} 
	\Value^{\minstrategy}(\qloc', \val')$. By considering the same strategy $\minstrategy'$ as the one defined in the case of \MinPl, we obtain that
	\[V_i(\qloc', \val') \leq \Value^{\minstrategy'}(\qloc', \val')  = 
	\Value^{\minstrategy}((\qloc, \val) \moveto{\delay, \trans} (\qloc', \val'))\]
		Thus, we deduce that 
		\begin{displaymath}
		\Value^{\minstrategy}(\qloc, \val) \geq 
		\weight(\trans) + \delay \,\weight(\qloc) + V_i(\qloc', \val')
		\end{displaymath}
		Since this holds for all edges 
		$(\qloc, \val) \moveto{\delay, \trans} (\qloc', \val')$, we deduce that 
		\begin{displaymath}
		\Value^{\minstrategy}(\qloc, \val) \geq 
		\sup_{(\qloc, \val) \moveto{\delay, \trans} (\qloc', \val')} 
		\big(\weight(\trans) + \delay \, \weight(\qloc) + V_i(\qloc', \val') \big) 
		= V_{i+1}(\qloc, \val) \,.
		\qedhere
		\end{displaymath}
	\end{itemize}
\end{proof}

Finally, we have tools to prove Proposition~\ref{prop:fixpoint-limit}. 
In particular, we fix $W$ be the set of configurations from where \MinPl can 
ensure to reach $\QLocsT$ (without restriction on the number of steps), that is
the limit of $(W_i)_i$: $W = \bigcup_i W_i$. 
By classical results~\cite[Theorem 103]{fijalkow-23} on the attractor computation in timed games, 
we know that there exists a finite $N\in \N$ such that $W = \bigcup_{i = 0}^N W_i$. 
Now, by letting $V = \inf_i V_i$, we can finally prove that $V = \Value$. 

We reason by double inequalities and we start by 
proving that $V \geq \Value = 
\inf_{\minstrategy} \Value^{\minstrategy}$. If $(\qloc,\val)\notin W$, we have for all $i\in \N$, $V_i(\qloc, \val) = +\infty$ 
(by Lemma~\ref{lem:WinValueFinite}), and thus $V(\qloc, \val) = +\infty$. Otherwise, $(\qloc, \val) \in W_N$. 
Let $\varepsilon > 0$. 
Since $(V_i)_i$ uniformly converges to its limit, there exists 
$k\geq N$ such that $V_k(\qloc, \val) \leq 
V(\qloc, \val) + \varepsilon$. By using Lemma~\ref{lem:WinValueStrat}, 
$\inf_{\minstrategy \in \minStratWin_k(\qloc, \val)} \Value^{\minstrategy}(\qloc, \val)\leq V(\qloc, \val) + \varepsilon$. 
By considering the infimum over all strategies, and since this holds for all $\varepsilon$, we get $\Value = 
\inf_{\minstrategy} \Value^{\minstrategy}\leq V$. 

Conversely, we prove that $V \leq \Value = 
\inf_{\minstrategy} \Value^{\minstrategy}$. 
By contradiction, we suppose that there exists a strategy $\minstrategy$ of \MinPl and an initial configuration $(\qloc,\val)$
such that $V(\qloc, \val) > \Value^{\minstrategy}(\qloc, \val)$. 
Since then $\Value^{\minstrategy}(\qloc, \val) < +\infty$, all plays 
conforming to $\minstrategy$ reach a target location. 
We (inductively) build a play $\play$ from $(\qloc, \val)$ 
conforming to $\minstrategy$ such that at each step we guarantee that 
$\last(\play) = (\qloc', \val')$ satisfies $\qloc' \notin \QLocsT$, and 
$V(\qloc', \val') > \Value^{\minstrategy}(\play)$. 
In particular, this implies that $\play$ is an infinite play that never reaches a target, 
and we get a contradiction. 

Now, to finish the proof, we provide the construction of a such $\play$.
First, we suppose that $\play = (\qloc, \val)$.
To initiate the inductive construction of $\play$, since $V(\qloc, \val) > 
\Value^{\minstrategy}(\qloc, \val)$, we deduce that $\qloc \notin \QLocsT$ (otherwise $V(\qloc, \val) = \weightfin(\qloc, \val) = 
\Value^{\minstrategy}(\qloc, \val)$ by Lemma~\ref{lem:ValueStrat}). 

Then, we suppose that $\play$ is a play from $(\qloc, \val)$ conforming to 
$\minstrategy$ such that $V(\qloc', \val') > \Value^{\minstrategy}(\play)$ where 
$\last(\play) = (\qloc', \val')$ and $\qloc' \notin \QLocsT$. We define a new 
step for $\play$ as follows. 
\begin{itemize}
	\item If $\qloc' \in \QLocsMin$, then we extend $\play$ by 
	$\play' = \play \moveto{\delay, \trans} (\qloc'', \val'')$, by letting 
	$\minstrategy(\play) = (\delay, \trans)$. Since $\play$ is conforming to $\minstrategy$, then 
	$\play'$ is also conforming to 
	$\minstrategy$.
	By induction hypothesis and Lemma~\ref{lem:ValueStrat}, 
	\begin{displaymath}
	V(\qloc', \val') > \Value^{\minstrategy}(\play) = \weight(\trans) + \delay \, 
	\weight(\qloc') + \Value^{\minstrategy}(\play') 
	\end{displaymath}
	Since $V$ is a fixpoint of $\F$ by Proposition~\ref{prop:fixpoint}, with $V(\qloc',\val')$ being thus equal to an infimum over all possible edges, we obtain 
	\[V(\qloc'',\val'') \geq V(\qloc',\val') - \weight(\trans) - \delay \, \weight(\qloc') > \Value^{\minstrategy}(\play') \]

	\item If $\qloc' \in \QLocsMax$, then we prove that there exists an edge  
	$(\qloc', \val') \moveto{\delay, \trans} (\qloc'', \val'')$ such that 
	$V(\qloc'', \val'') > \Value^{\minstrategy}(\play \moveto{\delay, \trans} 
	(\qloc'', \val''))$, and we define the new step of $\play$ with this edge 
	(the resulting play is conforming to $\minstrategy$). 
	To do that, we reason by contradiction, and we suppose that for all edges 
	$(\qloc', \val') \moveto{\delay, \trans} (\qloc'', \val'')$, we have 
	$V(\qloc'', \val'') \leq \Value^{\minstrategy}(\play \moveto{\delay, \trans} 
	(\qloc'', \val''))$. 
	In particular, we obtain a contradiction since: 
	\begin{align*}
	V(\qloc', \val') &> \Value^{\minstrategy}(\play)  \qquad \text{(by induction hypothesis)}\\
	&= 
	\sup_{(\qloc', \val') \moveto{\delay, \trans} (\qloc'', \val'')} 
	\big(\weight(\trans) + \delay \, \weight(\qloc') + 
	\Value^{\minstrategy}(\play \moveto{\delay, \trans} (\qloc'', \val''))\big) 
	\qquad \text{(by Lemma~\ref{lem:ValueStrat})}\\
	&\geq 
	\sup_{(\qloc', \val') \moveto{\delay, \trans} (\qloc'', \val'')} 
	\big(\weight(\trans) + \delay \, \weight(\qloc') + V(\qloc'', \val'')\big) 
	\qquad \text{(by monotonicity of $\F$)} \\
	&> V(\qloc', \val') \qquad \text{(since $V$ is a fixpoint of $\F$ by Proposition~\ref{prop:fixpoint}).}
	\end{align*}
\end{itemize}
%	Finally, to conclude this proof, we need to prove that $\qloc'' \notin \QLocsT$ 
%	(since $\play$ does not reach a target location, by inductive hypothesis 
%	applying to $\play$). We reason by contradiction and we suppose that 
%	$\qloc'' \in \QLocsT$. In this case, by Lemma~\ref{lem:ValueStrat}, 
%	we have 
%	\begin{displaymath} 
%	V(\qloc'', \val'') >
%	\Value^{\minstrategy}(\play \moveto{\delay, \trans} (\qloc'', \val'')) 
%	= V(\qloc'', \val'') = \weightfin(\qloc'', \val'')
%	\end{displaymath}
%	by previous arguments. Thus, we obtain a contradiction, i.e.~$\qloc'' \notin \QLocsT$ as expected.
%\end{proof}

This concludes the proof that $\lim_i V_i = \Value$, and thus of Theorem~\ref{theo:fixedpoint}.

\section{Conclusion}

We solve one-clock \WTG{s}
with arbitrary weights, an open problem for several years. 
We strongly rely on the determinacy of the game, taking the point of
view of $\MaxPl$, instead of the one of $\MinPl$ as was done in
previous work with only non-negative weights. We also use technical ingredients such as the closure of a
game, switching strategies for $\MinPl$, and acyclic unfoldings. 

Regarding the complexity, our algorithm runs in exponential time (with
weights encoded in unary), which does not match the known \PSPACE lower
bound with weights in unary~\cite{FearnleyIbsenJensenSavani-20}. 
Observe that this lower
bound only uses non-negative weights. 
This complexity gap deserves further study.

To compute the value function with a \PSPACE algorithm, a promising idea 
from a reviewer of this article consists in using the first-order theory over the reals with a fixed number of quantifier alternations where the satisfiability of a formula can be checked in \PSPACE \cite[remark 13.10]{BPR06}. The idea is to encode the greatest fixpoint of $\F$ 
(that is the value of the game, by Theorem~\ref{theo:fixedpoint}) in this logic. Indeed, since the value function of a one-clock \WTG is piecewise affine with a 
pseudo-polynomial number of cutpoints (according to~\cite{brihaye2021oneclock}), we can write such a formula by using a variable for each cutpoint and slope, and then expressing with inequalities and equalities that, for each cutpoint or line segment, the current valuation is at least as good as what can be obtained by either waiting until a later cutpoint, or jumping through a transition.
% It is worth observing that this formula is not enough to also compute the (almost-)optimal strategies of both players, contrary to our more technical approach that provides a complete representation of the value functions from which it is easy to obtain strategies.

Our work also opens three research directions.  
%We list someof them here. -> generic 
First, as we unfold the game into a finite tree, it
would be interesting to develop a symbolic approach that shares
computation between subtrees in order to obtain a more efficient
algorithm.  Second, playing stochastically in \WTG{s}
with shortest path objectives has been recently studied
in~\cite{MonmegeParreauxReynier-ICALP21}. One could study an extension
of one-clock \WTG{s} with stochastic transitions. In this context, \MinPl
aims at minimizing the expectation of the accumulated weight.  Third,
the analysis of cycles that we have done in the setting of one-clock
\WTG{s} can be an inspiration to identify new decidable classes of
\WTG{s} with arbitrarily many clocks.

\bibliographystyle{alphaurl}
\bibliography{biblio}

\newpage

\appendix

\section{Continuity of the value function on closure of regions}
\label{app:proof-val_continue}

\lemValContinue*

The main ingredient of our proof is, given a strategy $\minstrategy$
of \MinPl in $\rgame$, a location $\loc=(\qloc, \reg)$ of $\rgame$, and
valuations $\val,\val'\in \overline \reg$ (and not only $\val,\val' \in \reg$
as in the proof of \cite[Theorem~3.2]{brihaye2021oneclock}), to show how
to build a strategy $\minstrategy'$ in $\rgame$ and a
length-preserving function $g$ that maps plays of $\rgame$ starting in
$(\loc,\val')$ and conforming to $\minstrategy'$ to plays of
$\rgame$ starting in $(\loc,\val)$ conforming to $\minstrategy$ with
similar behaviour and weight. More precisely, we define $\minstrategy'$
and $g$ by induction on the length $k$ of the finite play that is
given as an argument and relies on the following set of induction
hypotheses:

\textbf{Induction hypothesis:} There exist a strategy $\minstrategy'$,
only defined on plays of length at most $k-1$ starting in
$(\loc,\val')$, and a function $g$ mapping plays of length $k$
starting in $(\loc,\val')$ conforming to $\minstrategy'$ to plays of
length $k$ starting in $(\loc,\val)$ conforming to $\minstrategy$
such that for all plays
$\play'=(\loc_0=\loc,\val'_0=\val')\moveto{\delay'_0,\trans'_0}
\cdots\moveto{\delay'_{k-1},\trans'_{k-1}} (\loc_k,\val'_k)$
conforming to $\minstrategy'$, letting
$(\loc_0,\val_0=\val)\moveto{\delay_0,\trans_0}
\cdots\moveto{\delay_{k-1},\trans_{k-1}} (\loc_{k},\val_{k})$
the play $g(\play')$, we have:
\begin{enumerate}
	\item \label{item:3}
	$|\val_k-\val'_k|\leq |\val-\val'|$;
	\item \label{item:4}
	$\weightC{(\play')} \leq \weightC{(g(\play'))} +
	\maxPriceLoc(|\val-\val'| -|\val_k-\val'_k|)$.
\end{enumerate}
We note that no property is required on the strategy $\minstrategy'$
for finite plays that do not start in~$(\loc,\val')$. Moreover, by the
invariants of~$\rgame$, we have that for every $i\in\{0,\ldots,k\}$,
$\val_i$ and $\val'_i$ belong to the interval $\overline{I_i}$ such
that $\loc_i = (\qloc_i, I_i)$.

Let us explain how this result would imply the desired result before 
going through the induction itself, i.e.~why 
$\val\mapsto \Value_\rgame((\qloc,I),\val)$ is continuous over 
$\overline \reg$. We remark first that the result directly implies that 
if the value of the game is finite for some valuation $\val$ in
$\overline \reg$, then it is finite for all other valuation $\val'$ in
$\overline \reg$. Indeed, a finite value of the game in $(\loc,\val)$ 
implies that there exists a strategy $\minstrategy$ such that every 
play starting in $(\loc,\val)$ and conforming to it reaches a target 
location in a final valuation such that the final weight function applying 
in this last configuration is finite. Moreover, denoting $\minstrategy'$ 
the strategy obtained from $\minstrategy$ thanks to the above result, 
any play $\play'$ starting in $(\loc,\val')$ and conforming to 
$\minstrategy'$ reaches a target location (since $g(\play')$ does as a 
play conforming to $\minstrategy$). Moreover, its final weight 
function is finite as the final valuation of $\play'$, and $g(\play')$ 
sit in the same region and, by hypothesis, a final weight function 
is either always finite or always infinite within a region.

Now, assuming the value of the game is finite over $\overline \reg$ and 
we show that the value is continuous over $\overline \reg$. To do it, 
we show that, for all $\val\in \overline \reg$, for all $\varepsilon>0$,
there exists $\delta>0$ such that for all $\val' \in \overline \reg$ with
$|\val-\val'|\leq \delta$, we have
$|\Value_\rgame(\loc,\val)-\Value_\rgame (\loc,\val')|\leq \varepsilon$. 
To this end, we can show that:
\begin{equation}
\label{eq:lemma-continuity}
|\Value_\rgame (\loc,\val)-\Value_\rgame(\loc,\val')| 
\leq (\maxPriceLoc + K)|\val-\val'|
\end{equation}
where $K$ is the greatest absolute value of the slopes appearing
in the piecewise affine functions within $\weightT$.  Indeed, assume
that this inequality holds, and consider $\val\in \overline I$ and a
positive real number $\varepsilon$. Then, we let
$\delta=\frac{\varepsilon}{\maxPriceLoc+ K}$, and we consider a
valuation $\val'$ such that $|\val - \val'|\leq\delta$. In this case,
\eqref{eq:lemma-continuity} becomes:
\begin{displaymath}
|\Value_\rgame (\loc,\val)-\Value_\rgame (\loc,\val')| \leq (\maxPriceLoc+
K)|\val-\val'|\leq (\maxPriceLoc+ K)
\frac{\varepsilon}{\maxPriceLoc+ K}\leq\varepsilon \,.
\end{displaymath}
Thus, proving~\eqref{eq:lemma-continuity} is sufficient to establish
continuity. 

On the other hand, \eqref{eq:lemma-continuity} is equivalent~to:
\begin{displaymath}
\Value_\rgame (\loc,\val)\leq \Value_\rgame (\loc,\val') +
(\maxPriceLoc+ K)|\val-\val'|
\quad \text{and} \quad
\Value_\rgame (\loc,\val')\leq \Value_\rgame (\loc,\val) + 
(\maxPriceLoc + K)|\val-\val'| \,.
\end{displaymath}
As those two last equations are symmetric with respect to $\val$ and
$\val'$, we only have to show either of them.  We thus focus on the
latter, which, by using the upper value, can be reformulated as: for
all strategies $\minstrategy$ of \MinPl, there exists a strategy
$\minstrategy'$ such that
\begin{displaymath}
\Value_\rgame^{\minstrategy'}(\loc,\val')\leq
\Value_\rgame^{\minstrategy} (\loc,\val)+(\maxPriceLoc + 
K)|\val-\val'| \,.
\end{displaymath} 
We note that this last equation is equivalent to
saying that there exists a function $g$ mapping plays $\play'$ from
$(\loc,\val')$ conforming to $\minstrategy'$ to plays
from~$(\loc,\val)$ conforming to $\minstrategy$ such that, for all
such $\play'$ the final valuations of $\play'$ and $g(\play')$ differ
by at most $|\val-\val'|$ and
\begin{displaymath}
\weightC{(\play')} \leq
\weightC{(g(\play'))}+\maxPriceLoc|\val-\val'|
\end{displaymath}
which is exactly what we claimed induction achieves since
$-|\val_k-\val'_k|\leq 0$.  Thus, to conclude this proof, let us now
define $\minstrategy'$ and $g$, by induction on the length $k$ of
$\play'$.

\textbf{Base case $k=0$: } In this case, $\minstrategy'$ does not have
to be defined since there are no plays of length $-1$. Moreover, in
that case, $\play' = (\loc,\val')$ and $g(\play')=(\loc,\val)$, in
which case both properties are trivial.

\textbf{Inductive case: } Let us suppose now that the construction is
done for a given $k\geq 0$ and perform it for $k+1$. We start with
the construction of $\minstrategy'$. To that extent, we consider a play
$\play'=(\loc_0=\loc,\val'_0=\val')\moveto{\delay'_0,\trans'_0}
\cdots\moveto{\delay'_{k-1},\trans'_{k-1}} (\loc_k,\val'_k)$
conforming to $\minstrategy'$ (provided by induction hypothesis)
such that $\loc_{k} \in \LocsMin$ . Let $(\delay, \trans)$ be the 
choice of delay and transition made by $\minstrategy$
on $g(\play')$, i.e.~$\minstrategy(g(\play'))=(\delay, \trans)$. 
Then, we define $\minstrategy'(\play')=(\delay', \transition)$ where
$\delay' = \max(0, \val_{k}+t - \val'_{k})$.  The delay $\delay'$ 
respects the guard of transition $\trans$, as can be seen from
\figurename~\ref{fig:deltaPrime}. Indeed, either
$\val_{k}+\delay = \val'_{k}+\delay'$ (cases (a) and (b) in 
\figurename~\ref{fig:deltaPrime}) or $\val_{k} \leq \val_{k} + 
\delay \leq \val'_{k}$ (case (c) in \figurename~\ref{fig:deltaPrime} 
where $\delay'=0$), in which case $\val'_{k}$ is in the same closure 
of region as $\val_{k}+\delay$ since $\val_k$ and $\val'_k$ are in 
the same closure of region by induction hypothesis: we conclude by 
noticing that the guard of $\trans$ is closed.

\begin{figure}[tbp]
	\centering
	\begin{tikzpicture}
	\begin{scope}
	\draw[->,dashed] (-0.5,0) -- (2.5,0);  
	\node  at (0,-0.3) {$\val'_{k}$};
	\node (vkp) at (0,0) {$\bullet$};
	\node  at (1,0.3) {$\val_{k}$};
	\node (vk) at (1,0) {$\bullet$};
	\node (vkdelta) at (2,0) {$\bullet$};
	
	\node at (1,-1.5) {(a)};
	
	\draw[->] (vk) to[bend left] node[midway, above] {$t$} (vkdelta);
	\draw[->] (vkp) to[bend right]  node[midway, below]  {$t'$} (vkdelta);
	\end{scope}
	
	\begin{scope}[xshift = 4cm]
	\draw[->,dashed] (-0.5,0) -- (2.5,0);  
	\node  at (0,-0.3) {$\val_{k}$};
	\node (vk) at (0,0) {$\bullet$};
	\node  at (1,0.3) {$\val'_{k}$};
	\node (vkp) at (1,0) {$\bullet$};
	\node (vkdelta) at (2,0) {$\bullet$};
	
	\node at (1,-1.5) {(b)};
	
	\draw[->] (vk) to[bend right] node[midway, below] {$t$} (vkdelta);
	\draw[->] (vkp) to[bend left]  node[midway,above]  {$t'$} (vkdelta);
	\end{scope}

	\begin{scope}[xshift = 8cm]
	\draw[->,dashed] (-0.5,0) -- (2.5,0);  
	\node  at (0,-0.3) {$\val_{k}$};
	\node (vk) at (0,0) {$\bullet$};
	\node (vkdelta) at (1,0) {$\bullet$};
	\node  at (2,-0.3) {$\val'_{k}$};
	\node (vkp) at (2,0) {$\bullet$};
	
	\node at (1,-1.5) {(c)};
	
	\draw[->] (vk) to[bend right] node[midway, below] {$t$} (vkdelta);
	\draw[->]  (vkp) edge[loop above] node[midway, above] {$t'$} (vkp);
	\end{scope}
	
	\end{tikzpicture}
	\caption{The definition of $\delay'$ when (a)
		$\val'_k \leq \val_k$; (b)
		$\val_k < \val'_k < \val_k + \delay$; (c)
		$\val_k < \val_k+\delay < \val'_k $.}
	\label{fig:deltaPrime}
\end{figure}
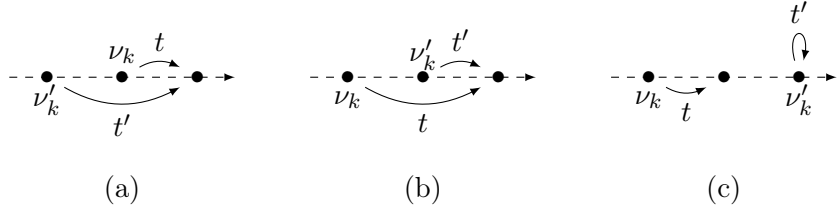

Let us now build the mapping $g$. Let
$\play'=(\loc_0=\loc,\val'_0=\val')\moveto{\delay'_0,\trans'_0}
\cdots\moveto{\delay'_{k},\trans'_{k}} (\loc_{k+1},\val'_{k+1})$
be a play conforming to $\minstrategy'$ and let
$\tilde{\play}' =(\loc_0,\val'_0)\moveto{\delay'_0,\trans'_0}
\cdots\moveto{\delay'_{k-1},\trans'_{k-1}} (\loc_{k},\val'_{k})$
its prefix of length~$k$. Using the construction of~$g$ over plays of
length $k$ by induction, the play $g(\tilde{\play}') =
(\loc_0,\val_0=\val)\moveto{\delay_0,\trans_0}
\cdots\moveto{\delay_{k-1},\trans_{k-1}} (\loc_{k},\val_{k})$
satisfies properties \eqref{item:3} and \eqref{item:4}. Then:
\begin{itemize}
	\item if $\loc_k \in \LocsMin$ and
	$\minstrategy(g(\tilde{\play}')) = (\delay,\transition)$, then
	$g(\play')=g(\tilde{\play}')\moveto{\delay,\transition}
	(\loc_{k+1},\val_{k+1})$ is obtained by applying those choices on
	$g(\tilde{\play}')$. By the construction of $\minstrategy'$, we
	moreover have~$\trans'_k=\trans$;
	
	\item if $\loc_k \in \LocsMax$, the last valuation $\val_{k+1}$ of 
	$g(\play')$ is rather obtained by choosing action
	$(\delay, \trans'_k)$ verifying $\delay = \max (0,
	\val'_{k} + \delay'_k  -\val_k)$. We note that transition 
	$\trans'_k$ is allowed since both $\val_k+ \delay$ and~$\val'_k + 
	\delay'_k$ are in the same closure of region (for similar 
	reasons as above).
\end{itemize}
Moreover, by induction hypothesis $\tilde{\play}'$ and $g(\tilde{\play}')$ 
have the same length. 

Now, to prove~\eqref{item:3}, we notice that we always have either 
\begin{displaymath}
\val_k + \delay = \val'_k + \delay'_k 
\quad \text{or} \quad
\val_k \leq \val_k + \delay \leq \val'_k = \val'_k + \delay'_k 
\quad \text{or} \quad
\val'_k \leq \val'_k + \delay \leq \val_k = \val_k + \delay \,.
\end{displaymath}  
In all of these possibilities, we have
$|(\val_k + \delay) - (\val'_k + \delay'_k)| \leq |\val_k - \val'_k|$. 

We finally check property~\eqref{item:4}. Either $\loc_k$ belongs to \MinPl or 
to \MaxPl, using the induction hypothesis, we have:
\begin{align*}
\weightC{(\play')} &=
\weightC{(\tilde{\play}')} + \weight(\trans'_k) + \delay'_k \, \weight(\loc_k) \\ 
&\leq \weightC{(g(\tilde{\play}'))} + \maxPriceLoc(|\val-\val'|
- |\val_k - \val'_k|) + \weight(\trans'_k) + \delay'_k \, \weight(\loc_k) \\ 
&=\weightC{(g(\play'))} + (\delay'_k - \delay) \, \weight(\loc_k) + 
\maxPriceLoc(|\val - \val'| - |\val_k - \val'_k|) \,. 
\end{align*}
To conclude, let us claim that
\begin{equation}
\label{eq:t'-t}
|\delay'_k - \delay|\leq |\val_k-\val'_k| -
|\val'_{k+1}-\val_{k+1}|
\end{equation}
Thus, since $|\weight(\loc_k)|\leq \maxPriceLoc$, we conclude that
\begin{displaymath}
\weightC{(\play')}\leq \weightC{(g(\play'))}
+ \maxPriceLoc(|\val-\val'|
- |\val_{k+1}-\val'_{k+1}|)
\end{displaymath}
which concludes the induction.

To conclude the proof, we prove~\eqref{eq:t'-t}.
First, we suppose that $\trans'_k$ does not contain a reset. In 
particular, we have  $\delay'_k = \val'_{k+1}-\val'_{k}$ and
$\delay = \val_{k+1}-\val_{k}$, thus $|\delay'_k-\delay| = 
|\val'_{k+1} - \val'_{k} - (\val_{k+1} - \val_{k})|$.  Then, two
cases are possible: either $\delay'_k = \max(0, \val_{k}+\delay - \val'_{k})$ or
$\delay = \max (0, \val'_{k}+\delay'_k-\val_k)$. So we have three different
possibilities:
\begin{itemize}
	\item if $\delay'_k + \val'_k = \delay +\val_k$, then
	$\val'_{k+1} = \val_{k+1}$, thus
	\begin{displaymath}
	|\delay'_k - \delay|= |\val_k-\val'_k| =
	|\val_k - \val'_k| - |\val'_{k+1} - \val_{k+1}| \,;
	\end{displaymath}
	
	\item if $\delay=0$, then   
	$\val_k = \val_{k+1} \geq \val'_{k+1} \geq \val'_k$, thus
	\begin{displaymath}
	|\delay'_k - \delay| =
	\val'_{k+1} - \val'_k = (\val_{k} - \val'_k) -
	(\val_{k+1} - \val'_{k+1}) = |\val_k - \val'_k| -
	|\val'_{k+1} - \val_{k+1}| \,;
	\end{displaymath}
	
	\item if $\delay'_k=0$, then
	$\val'_k = \val'_{k+1} \geq \val_{k+1 }\geq \val_k$, thus
	\begin{displaymath}
	|\delay'_k - \delay| =
	\val_{k+1} - \val_k = (\val'_{k} - \val_k) -
	(\val'_{k+1} - \val_{k+1}) = |\val_k - \val'_k| -
	|\val'_{k+1} - \val_{k+1}| \,.
	\end{displaymath}
\end{itemize}
Otherwise, $\trans'_k$ contains a reset, then
$\val'_{k+1} = \val_{k+1} = 0$. If
$\delay'_k = \val_{k} + \delay - \val'_{k}$, we have that
$|\delay'_k - \delay| = |\val_k - \val'_k| $. Otherwise, 
$\delay'_k=0$ and $\delay \leq \val'_k-\val_k$.
In all cases, we have proved \eqref{eq:t'-t}.

\end{document}